\documentclass[11pt]{article}
\usepackage{adjustbox, array, algorithm, amsmath, amsthm, amsfonts, bbm, caption, complexity, enumitem, standalone, subfig, geometry, physics, fullpage, footmisc, mathtools, lmodern, thmtools, thm-restate, tikz, xspace}
\usepackage{hyperref}
\usepackage[false]{anonymous-acm}
\usepackage[noend]{algpseudocode}
\usepackage[backend=biber,minalphanames=3,maxnames=99,url=false,doi=false,isbn=false,style=alphabetic]{biblatex}
\hypersetup{colorlinks=true, linkcolor=blue, urlcolor=black, citecolor=black, breaklinks=true}
\usetikzlibrary{matrix,decorations.pathreplacing}

\newtheorem{theorem}{Theorem}
\newtheorem{claim}{Claim}
\newtheorem{example}{Example}
\newtheorem{corollary}{Corollary}
\newtheorem{lemma}{Lemma}
\newtheorem{observation}{Observation}
\theoremstyle{definition}
\newtheorem{definition}{Definition}
\newtheorem{remark}{Remark}

\DeclareFontFamily{U}{mathc}{}
\DeclareFontShape{U}{mathc}{m}{it}%
{<->s*[1.03] mathc10}{}
\DeclareMathAlphabet{\mathcal}{U}{mathc}{m}{it}
\DeclareMathOperator{\Sol}{\mathcal{S\mkern-3mu ol}}

\newcommand{\ttt}[1]{\textnormal{\texttt{#1}}}
\newcommand\plabel[1]{\phantomsection\label{#1}}

\title{Shared Randomness in Locally Checkable Problems:\\The Role of Computational Assumptions 
\blockanon{\thanks{Research supported in part by grants from the Israel Science Foundation (no.\  and 2686/20), by the Simons Foundation Collaboration on the Theory of Algorithmic Fairness and by the Israeli Council for Higher Education (CHE) via the Weizmann Data Science Research Center}}}

\authoranon{
\author{Adar Hadad\thanks{
		Department of Computer Science and Applied Mathematics, Weizmann Institute of Science, Rehovot, Israel.
		Email: {hadad.hadad@weizmann.ac.il}.
	}
	\and
	Moni Naor\thanks{
		Department of Computer Science and Applied Mathematics, Weizmann Institute of Science, Rehovot, Israel.
		Incumbent of the Judith Kleeman Professorial Chair.
		Email: {moni.naor@weizmann.ac.il}.
  }
}
}

\date{}

\begin{document}

\maketitle

\begin{abstract}
Shared randomness is a valuable resource in distributed computing, allowing some form of coordination between processors without explicit communication. But what happens when the shared random string can affect the inputs to the system?

Consider the class of distributed graph problems where the correctness of solutions can be checked locally, known as \textit{Locally Checkable Labelings} ($\mathsf{LCL}$). 
$\mathsf{LCL}$ problems have been extensively studied in the $\mathsf{LOCAL}$ model, where nodes operate in synchronous rounds and have access only to local information. This has led to intriguing insights regarding the power of private randomness. E.g., for certain round complexity classes, derandomization does not incur an overhead (asymptotically).

This work considers a setting where the randomness is public. Recently, an $\mathsf{LCL}$ problem for which shared randomness can reduce the round complexity was discovered by Balliu \textit{et al.} (ICALP 2025) \cite{bgk25}.
This result applies to inputs set obliviously of the shared randomness, which may not always be a plausible assumption. 

We define a model where the inputs can be adversarially chosen, even based on the shared randomness, which we now call \textit{preset public coins}. 
We study $\mathsf{LCL}$ problems in the preset public coins model, under assumptions regarding the computational power of the adversary that selects the input. We show connections to hardness in the class \textsf{TFNP}. Our results are:

\begin{enumerate}

\item Assuming a hard-on-average problem in \textsf{TFNP}, we present an $\mathsf{LCL}$ problem that, in the preset public coins model, demonstrates a gap in the round complexity between polynomial-time and unbounded adversaries.

\item An $\mathsf{LCL}$ problem for which the error probability is significantly higher when facing unbounded adversaries implies a hard-on-average problem in $\mathsf{TFNP/poly}$.

\end{enumerate}
\end{abstract}


\newpage

\section{Introduction} 
In distributed algorithms, multiple entities collaborate to compute a function of the global inputs. In the \textsf{LOCAL} model, introduced by Linial \cite{lin92}, the processors can directly communicate only with their neighbors, implying that within $r$ steps a node may receive messages from at most $r$ hops away. The main complexity measure is the number of rounds required for successfully performing a computation. In this model, the topology of the network plays a key role in determining the complexity of algorithms. In this way, we capture the effect of \textit{locality} on our ability to evaluate functions in the distributed setting. Although the functions can be defined in a \textit{global} manner, taking into account the network as a whole, each individual processor has a restricted, local view.

We focus on problems for which solutions can be efficiently verified. Namely, given a solution, its correctness can be deterministically verified within $O(1)$ rounds of communication. This is the motivation behind the class \textit{locally checkable labelings} (\textsf{LCL}), introduced by Naor and Stockmeyer~\cite{ns95}. The solution to a problem on a bounded-degree graph is the output labels of the nodes, and we require that by inspecting neighborhoods with a constant radius it is possible to verify that the labels satisfy a list of constraints. The constraints are specified as a finite set of labeled constant-size graphs, and each node's constant-radius neighborhood must be isomorphic to one of these labeled graphs. A \textsf{LOCAL} algorithm that \textit{solves} an \textsf{LCL} problem assigns an output label to each node in such a way that none of the constraints are violated. Key problems studied in this context are maximal matching, maximal independent set, and graph coloring.

The algorithms for generating these output labelings can be randomized. That is, every processor is given access to a (separate) tape of private coins. In this case, we require the labels returned by all nodes to be valid with sufficiently high probability, taken over the coins of the whole network.

Another flavor of randomized communication protocols involves a public random string, visible to all the processors. At first glance (and even a second one), it is unclear whether this extra resource makes the model strictly stronger (see Ghaﬀari \cite[Open Problem 5]{Gha20}).
Balliu, Ghaffari, Kuhn, Modanese, Olivetti, Rabie, Suomela, and Uitto~\cite{bgk25} were the first to show the advantage of shared randomness in solving \textsf{LCL} problems in the \textsf{LOCAL} model. They present an \textsf{LCL} problem that demonstrates an exponential improvement in the round complexity once shared randomness is employed. 

To give a taste of their construction, think of a square grid. For the output labeling to be valid, in at least one of the rows, the nodes should return the same bit that was given as input to the rightmost node. Evidently, this problem is hard if the number of rounds is restricted to $O(\log n)$, since the information cannot travel quickly enough from one end of the grid to the other. But, if the nodes use shared randomness, a uniformly random bit coincides with an input bit with probability $\frac{1}{2}$. Repeating this over all the rows, the failure probability is sufficiently low.

This model assumes that the entity that selects the input is {\em oblivious} of the shared randomness. The independence of the two plays a crucial role in analyzing the success probability of the solution suggested above. We refer to this as the \textit{oblivious} model. 

However, implementing shared randomness in practice presents several difficulties. It necessitates the existence of a trusted party, or else the public string can be manipulated and may not be truly random. As pointed out by Mironov, Naor, and Segev~\cite{mns08}, even under the assumption that such a trusted party exists, it is not always true that the string can be revealed only {\em after} the inputs of the nodes are fixed. It is reasonable to think about the input as being chosen when the public randomness is already known. A malicious user can leverage this knowledge to choose the input adversarially. But even in more benign circumstances, if the same string is used in several successive computations, then the result of one instance may be influenced by the value of the shared randomness and may creep into the values of subsequent computations. 

For example, in the problem presented by \cite{bgk25}, the adversary can use its knowledge of the public coins and select the complement bits. Consequently, the strategy sketched earlier fails.

\subsection{Models}
In this work, we study algorithms specifically designed to remain robust even when the inputs are selected by an adversary who sees the public randomness. We refer to it as the \textit{preset public coins} model. Take note that the graph is fixed before the randomness is made public.\footnote{A potential criticism of this model is that the graph structure and the input labels can be seen as interchangeable, since input labels can be encoded as part of the graph (e.g., by attaching gadgets to the nodes). 
However, we contend that it is reasonable to assume that the network’s structure is determined before the randomness is revealed, as it's often the case that the topology of a cluster is planned upfront.}

In the field of distributed algorithms, it is common to allow processors to be unbounded in terms of their time and space complexities. 
In contrast, we restrict the processors' runtime to be polynomial in $n$ (the size of the graph). 
Furthermore, in the context of the \textit{preset public coins} model, we make a distinction between adversaries based on their computational resources. We consider efficient adversaries (i.e., polynomial-time) and computationally unbounded ones.
The models that we examine are listed below. For formal definitions, see Section~\ref{subsec:distributed}.

\begin{itemize}
\item \textbf{Private Coins.} No public randomness at all. The only source of randomness is the private coins tossed by the nodes.

\item \textbf{Oblivious.} The input provided to the nodes is independent of the public randomness. The success probability is taken over both the private coins and the public randomness.

\item \textbf{Preset Public Coins with Bounded Adversaries.} The public randomness is known in advance. A probabilistic adversary that may run for $\poly(n)$ steps chooses the inputs of the nodes after seeing the public random string. Success probability is taken over the private coins, the preset public randomness and the internal coin tosses of the adversary.

\item \textbf{Preset Public Coins with Unbounded Adversaries.} The public randomness is known in advance. A computationally unbounded adversary (assumed to be deterministic) chooses the inputs of the nodes after seeing the public random string. The success probability is taken over both the private coins and the preset public randomness.
\end{itemize}
The round complexities of solving an \textsf{LCL} problem $\Pi$ with probability $1 - \epsilon$ in each of the models are denoted by $\mathsf{ROUND}^{\mathsf{priv}}_{\epsilon}(\Pi)$, $\mathsf{ROUND}^{\mathsf{O}}_{\epsilon}(\Pi)$, $\mathsf{ROUND}^{\mathsf{B}}_{\epsilon}(\Pi)$ and $\mathsf{ROUND}^{\mathsf{U}}_{\epsilon}(\Pi)$, respectively. 

\paragraph{The Role of Computational Assumptions.} 
The choice to study bounded adversaries is motivated by results in related areas, which show that computational assumptions help when the inputs may depend on the public randomness. 
Cohen and Naor \cite{cn22} introduced the preset public coins model in the context of communication complexity. In this setting, multiple players jointly compute a function of their inputs, with the inputs being selected adversarially after the public randomness is revealed.
They demonstrated that by assuming the existence of \textit{collision-resistant hash functions} (CRH), communication costs can be reduced.
This leads us to ask whether computational assumptions might be beneficial in our domain as well.

We emphasize that this question becomes meaningful only in settings with shared randomness.  
When only private coins are allowed, both bounded and unbounded adversaries can be implemented by fixing an input labeling that minimizes the success probability of the \textsf{LOCAL} algorithm.
However, the presence of shared randomness introduces an element of unpredictability: while an unbounded adversary may adapt to it, a bounded adversary cannot anticipate or prepare for it in advance.

\paragraph{Landscape of \textsf{LCL} Problems.} 
Sub-classes of \textsf{LCL} can be defined with respect to the models discussed above. 
Let $\epsilon: \mathbb{N} \rightarrow [0,1]$ be the error probability, and $r: \mathbb{N} \rightarrow \mathbb{N}$ be the round complexity. $\text{\textsf{LCL}}^{\mathsf{priv}}_{\epsilon(n)}[r(n)]$ is the class of all \textsf{LCL} problems that can be solved on $n$~vertex graphs with probability $1- \epsilon(n)$ within $O(r(n))$ rounds in the private coins model. 
\[
\text{\textsf{LCL}}^{\mathsf{priv}}_{\epsilon(n)}[r(n)] := 
\big\{ \Pi \in \text{\textsf{LCL}} |  \mathsf{ROUND}^{\mathsf{priv}}_{\epsilon(n)}(\Pi) = O(r(n)) \big\}
\]
Analogous sub-classes can be defined for the oblivious model and the preset public coins model: $\text{\textsf{LCL}}^{\mathsf{O}}_{\epsilon(n)}[r(n)]$, $\text{\textsf{LCL}}^{\mathsf{B}}_{\epsilon(n)}[r(n)]$ and $\text{\textsf{LCL}}^{\mathsf{U}}_{\epsilon(n)}[r(n)]$.

Clearly, by ignoring the preset randomness, any algorithm in the private coins model can be viewed as an algorithm in the preset public coins model (with unbounded adversaries) with the same round complexity and success rate.
Moreover, by restricting the adversary, any algorithm that succeeds with probability $1- \epsilon(n)$ within $r(n)$ rounds still performs at least as well against the weaker adversary.

\begin{observation}[Inclusions]
\label{obs:inclusions}
For every error probability $\epsilon: \mathbb{N} \rightarrow [0,1]$ and round complexity $r: \mathbb{N} \rightarrow \mathbb{N}$, the following inclusions hold:
\begin{equation}
\label{eq:inclusions}
\textnormal{\textsf{LCL}}^{\mathsf{priv}}_{\epsilon(n)}[r(n)] \subseteq
\textnormal{\textsf{LCL}}^{\mathsf{U}}_{\epsilon(n)}[r(n)] \subseteq 
\textnormal{\textsf{LCL}}^{\mathsf{B}}_{\epsilon(n)}[r(n)] \subseteq 
\textnormal{\textsf{LCL}}^{\mathsf{O}}_{\epsilon(n)}[r(n)] 
\end{equation}
\end{observation}

The gap demonstrated in \cite{bgk25} is a separation between $\text{\textsf{LCL}}^{\mathsf{priv}}_{\epsilon(n)}[r(n)]$ and $\text{\textsf{LCL}}^{\mathsf{O}}_{\epsilon(n)}[r(n)]$, for certain choices of $\epsilon(n)$ and $r(n)$. Below, we point out that the rightmost and leftmost inclusions in Eq.~(\ref{eq:inclusions}) are strict. The proofs are straightforward adaptations of results from \cite{bgk25}, and are outlined in Appendices~\ref{subsec:proof_gap_oblivious} and~\ref{subsec:proof_gap_private_coins}.

\begin{restatable}[Separation between the Oblivious Model and the Preset Public Coins Model with Bounded Adversaries]{lemma}{oblivious}
\label{lemma:gap_oblivious}
There is an \textsf{LCL} problem $\Pi_1$ such that:
\[
\mathsf{ROUND}^{\mathsf{O}}_{1/n}(\Pi_1) = O(\log n) \text{ and } \mathsf{ROUND}^{\mathsf{B}}_{1/n}(\Pi_1) = \Omega(\sqrt{n})
\]
\end{restatable}

\begin{restatable}[Separation between the Private Coins Model and the Preset Public Coins Model with Unbounded Adversaries]{lemma}{privatecoins}
\label{lemma:gap_private}
There is an \textsf{LCL} problem $\Pi_2$ such that:
\[
\mathsf{ROUND}^{\mathsf{U}}_{1/n}(\Pi_2) = O(\log n) \text{ and } \mathsf{ROUND}^{\mathsf{priv}}_{1/n}(\Pi_2) = \Omega(\sqrt{n})
\]
\end{restatable}

As for the difference between bounded and unbounded adversaries, it becomes interesting if one is willing to make computational assumptions. That is, to assume the existence of computational tasks that cannot be solved efficiently. Our focus is on hard-on-average problems.
Intuitively, we can use such a problem to pose a challenge to the adversary: 
if it succeeds, it can generate an input labeling that forces any \textsf{LOCAL} algorithm to run for many rounds;
otherwise, the distributed problem is solvable in few rounds.

More formally, a problem is said to be $(T, \mu)$-hard-on-average if there exists a distribution over its instances such that a $T(\lambda)$-time probabilistic algorithm cannot solve the problem on a random instance with probability better than $\mu(\lambda)$, with $\lambda$ being the instance's length. 
Throughout this work we let the solver be non-uniform (that is, have access to advice of length $\poly(\lambda)$).~\footnote{This work focuses exclusively on hardness with respect to non-uniform solvers, so we sometimes treat this aspect as implicit.%See Remark~\ref{remark:non-uniform_adv} for more detail.
}
Moreover, when referring to a \textit{public-coin} $(T, \mu)$-hard-on-average distributional problem, we make the additional assumption that the problem remains hard even when the random bits used for sampling the instance are known.

Whenever we make a claim about distributed graph problems while relying on the existence of $(T, \mu)$-hard-on-average problems, we set $\lambda$ to be a function of $n$, the graph size.
We make a distinction between problems that are subexponentially hard, and ones that are only polynomially hard (see Section~\ref{subsec:complexity} for more details). However, in both cases, we set $\lambda := \lambda(n)$ such that $T(\lambda)$ is superpolynomial in $n$ and $\mu(\lambda)$ is negligible in $n$ (smaller than $1 / p(n)$ for any polynomial $p(n)$).  
We study the connection between \textsf{LCL} problems and the existence of $(T, \mu)$-hard-on-average problems. In particular, we pose the following question:
\begin{center}
\textit{In the preset public coins model, does the existence of public-coin hard-on-average problems help with \textsf{LCL} problems?}
\end{center}
This corresponds to a gap in the round complexity between two environments, one with adversaries that run in $\poly(n)$ time and one with computationally unbounded adversaries. 

Additionally, we ask whether \textsf{LCL} problems in the preset public coins model that satisfy certain properties might imply the existence of intractable problems. To this end, we consider \textsf{LCL} problems in which the success rate changes significantly depending on the computational power of the adversary.
\begin{center}
\textit{
Assume the existence of an \textsf{LCL} problem that, in the preset public coins model, can be solved with a significantly higher success probability when the adversary is computationally bounded.\newline
Does this imply any known computational assumptions?}
\end{center}
Specifically, if the gap is large enough, can it be leveraged to construct a hard-on-average problem?
Note that this is not the exact reverse of the first question, which was concerned with the round complexity, rather than with the success probability. 

More concretely, the (public-coin) hard-on-average problems that we consider are from the complexity class $\TFNP$ (or its non-uniform variant $\TFNP/\poly$), which stands for \textit{total} search problems that can be verified in polynomial time. A search problem is \textit{total} if for every instance of the problem, there exists a solution. Apparently, this class captures the desired properties that we need from hard-on-average problems.

We note that the complexity class $\TFNP$ has been extensively studied, with much of the interest primarily due to its subclasses. A notable example is the subclass $\mathsf{PPAD}$, introduced by Papadimitriou \cite{pap94}, which is known for being complete with respect to finding Nash equilibria (see Daskalakis, Goldberg, and Papadimitriou \cite{dgp09}). 
Moreover, average-case hardness in $\TFNP$ is known to be related to other known computational assumptions.
For example, Hub\'{a}cek, Naor, and Yogev~\cite{hny17} proved that the existence of (public-coin) hard-on-average problems in $\TFNP / \poly$ (or in $\TFNP$, under additional assumptions) is implied by hard-on-average problems in $\NP$. 
Bear in mind that the latter is implied by the existence of \textit{one-way functions} (OWFs), whereas it is unknown whether the converse holds.

\subsection{Contributions}
Our results are twofold. First, we present a problem that exhibits a gap between the best round complexity achievable when the adversary is computationally unbounded, compared to what can be obtained if the adversary is assumed to be efficient. In more detail, we show a lower bound that holds against unbounded adversaries and an upper bound that holds only in the case of efficient adversaries.

\begin{restatable}{theorem}{roundsgap}
\label{thm:rounds_gap}
Assume the existence of a public-coin $(T, \mu)$-hard-on-average problem in $\TFNP$.
Then, there exist an \textsf{LCL} problem $\Pi$ and a constant $c \in (0,1/2]$ such that: 
\[
\mathsf{ROUND}^{\mathsf{U}}_{1/n}(\Pi) = \Omega(n^c) \text{ and } \mathsf{ROUND}^{\mathsf{B}}_{\mu(\lambda) + \mathsf{negl}(n)}(\Pi) = O(\lambda)
\]
Where $n$ refers to the graph size and $\lambda$ is a parameter ($\lambda = \omega(\log n)$ but $\lambda = O(\poly(n))$) taken such that $T(\lambda)$ is superpolynomial in $n$, and $\mathsf{negl}(n)$ is some negligible function.
Moreover, against bounded adversaries, the length of the preset public random string is $O(\lambda)$.
\end{restatable}

\begin{corollary}[Polynomial Gap]
\label{cor:poly_hardness}
Assume the existence of a public-coin polynomially hard-on-average problem in $\TFNP$. Set $\lambda$ to be $\lceil n^{\delta} \rceil$ for any arbitrarily small $\delta \in (0,c)$. Then, there is negligible function $\mathsf{negl}(n)$ such that the \textsf{LCL} problem $\Pi$ from Theorem~\ref{thm:rounds_gap} satisfies: 
\[
\mathsf{ROUND}^{\mathsf{B}}_{\mathsf{negl}(n)}(\Pi) = O(n^{\delta})
\]
Moreover, the length of the preset public randomness is also $O(n^{\delta})$.
\end{corollary}

\begin{corollary}[Subexponential Gap]
\label{cor:subexp_hardness}
Assume the existence of a public-coin subexponentially hard-on-average problem in $\TFNP$. Thus, there are constants $\kappa, \rho > 0$ such that the problem is $(T, \mu)$-hard-on-average for $T(\lambda) = 2^{\lambda^{\kappa}}$ and $\mu(\lambda) = 2^{-\lambda^{\rho}}$. Set $\lambda$ to be $\lceil \log^{\delta} n \rceil$ for any $\delta > \max\{1/\kappa, 1 / \rho\}$. 
Then, there is negligible function $\mathsf{negl}(n)$ such that the \textsf{LCL} problem $\Pi$ from Theorem~\ref{thm:rounds_gap} satisfies: 
\[
\mathsf{ROUND}^{\mathsf{B}}_{\mathsf{negl}(n)}(\Pi) = O(\log^{\delta} n)
\]
Moreover, the length of the preset public randomness is also $O(\log^{\delta} n)$.
\end{corollary}

The constant $c$ depends on the hard problem we assume to exist. The first part of Theorem~\ref{thm:rounds_gap} says that any \textsf{LOCAL} algorithm that solves $\Pi$ with failure probability at most $1 / n$ when the inputs are chosen by an unbounded adversary, must run for $\Omega(n^c)$ rounds. On the other hand, whenever we consider a more relaxed environment, where the inputs are chosen efficiently, the second part of Theorem~\ref{thm:rounds_gap} (and Corollaries~\ref{cor:poly_hardness} and \ref{cor:subexp_hardness}) show that a negligible error probability can be obtained in $O(n^\delta)$ rounds for every $\delta \in (0,1)$ (assuming polynomial hardness) or in $O(\log^\delta n)$ rounds for every $\delta$ larger than some threshold (assuming subexponential hardness). This establishes an arbitrarily large polynomial gap in the number of rounds in the former case, and a superpolynomial gap in the latter. We stress that not only does the round complexity significantly improve, but the success rate does as well.

We remark that the constant $c$ can be explicitly determined given additional information about the $\TFNP$ problem. See Example~\ref{cor:explicit_lb} below.

\begin{example}[Explicit Lower Bound]
\label{cor:explicit_lb}
Assume the existence of a public-coin $(T, \mu)$-hard-on-average problem in $\TFNP$, where the solutions are of linear length, and both sampling and verification run in linear time.\footnote{A simple problem that exhibits this behavior is \textit{pigeonhole circuit}, introduced by Papadimitriou \cite{pap94}, coupled with the uniform distribution, under the assumption that  \textit{collision-resistance hash functions} (CRH) or \textit{one-way permutations} (OWP) exist.} Then, the \textsf{LCL} problem $\Pi$ from Theorem~\ref{thm:rounds_gap} satisfies: 
\[
\mathsf{ROUND}^{\mathsf{U}}_{1/n}(\Pi) = \Omega(\sqrt{n})
\]
\end{example}

Theorem~\ref{thm:rounds_gap} is essentially a separation between $\text{\textsf{LCL}}^{\mathsf{U}}_{\epsilon(n)}[r(n)]$ and $\text{\textsf{LCL}}^{\mathsf{B}}_{\epsilon(n)}[r(n)]$ for specific values of $\epsilon(n)$ and $r(n)$. Let $c$ be the constant mentioned in Theorem~\ref{thm:rounds_gap}.
Together with Lemmas~\ref{lemma:gap_oblivious},~\ref{lemma:gap_private}, it follows that assuming public-coin polynomial hardness-on-average of a distributional problem in $\TFNP$, for every $\delta_0 \in (0, c)$ we have a separation between the sub-classes in Eq.~(\ref{eq:inclusions}) that correspond to error probability $\epsilon(n) = 1/n$ and round complexity $r(n) = n^{\delta_0}$. In the case of subexponential hardness, this conclusion can be extended to $r(n) = \log^{\delta_1} (n)$ for all $\delta_1 > \max \{ 1/ \kappa, 1/\rho \}$ (where $\kappa$ and $\rho$ are as in Cor.~\ref{cor:subexp_hardness}).
See Fig.~\ref{fig:landscape}.

\begin{figure}[t]
\centering
\begin{tikzpicture}
% Ellipses
\draw[lightgray, line width=0.04cm] (0,-0.5) ellipse [x radius=2.5,y radius=0.75];
\node at (0,-0.5) {$\mathsf{LCL}^{\mathsf{priv}}_{1/n}[r(n)]$};
\draw[lightgray, line width=0.04cm] (0,0) ellipse [x radius=3.5,y radius=1.3];
\node at (0, 0.65) {$\mathsf{LCL}^{\mathsf{U}}_{1/n}[r(n)]$};
\draw[lightgray, line width=0.04cm] (0,0.55) ellipse [x radius=4.5,y radius=2];
\node at (0, 1.85) {$\mathsf{LCL}^{\mathsf{B}}_{1/n}[r(n)]$};
\draw[lightgray, line width=0.04cm] (0,1) ellipse [x radius=5.5,y radius=2.65];
\node at (0, 3) {$\mathsf{LCL}^{\mathsf{O}}_{1/n}[r(n)]$};
% Inequalities
\node at (0, 2.5) {$\neq$ \scriptsize(Lemma \ref{lemma:gap_oblivious})};
\node at (0, 1.3) {$\neq$ \scriptsize(Thm. \ref{thm:rounds_gap})};
\node at (0, 0.22) {$\neq$ \scriptsize(Lemma \ref{lemma:gap_private})};
\end{tikzpicture}
\caption{Landscape of the \textsf{LCL} problems that can be solved in $r(n)$ rounds with error at most $1/n$. Given a (public-coin) polynomial hardness-on-average of a distributional problem in \textsf{TFNP}, $r(n)$ can be any $n^{\delta_0}$, where $\delta_0 \in (0,c)$ (for $c$ mentioned in Theorem~\ref{thm:rounds_gap}). Assuming subexponential hardness, $r(n)$ can take any value from between $\log^{\delta_1}(n)$ (for $\delta_1$ larger than the constants $1 / \kappa$ and $1 / \rho$ specified in Cor.~\ref{cor:subexp_hardness}) and $n^{\delta_0}$.}
\label{fig:landscape}
\end{figure}
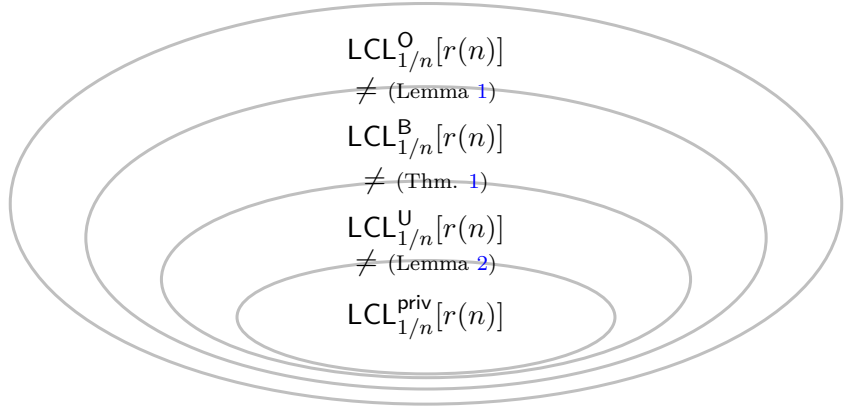

\begin{remark}[Randomness Complexity]
\label{remark:randomness_complexity}
Although we do not restrict the length of the public random string (other than requiring it to be polynomial), getting long public random strings is challenging to implement in practice. Hence, it is of central interest to show that the public randomness can be relatively short (e.g., Ghaffari and Kuhn \cite{gk19} studied distributed algorithms that use shared random strings of length no more than polylogarithmic).

By slightly modifying\footnote{At a high level, the nodes in the first $\lceil \log n \rceil$ rows of the grid behave the same, and the nodes in the remaining rows have some fixed output. The rows are consistent, and the probability that in all the rows the output disagrees with the input of the rightmost node is at most $1 / 2^{\Theta(\log n)} = n^{-\Theta(1)}$.}
the construction of \cite{bgk25}, success probability $1 - 1/n$ can be attained using just $O(\log n)$ shared random bits.
More generally, \cite{Gha20} observed that $O(\log n)$ shared random bits suffice for any algorithm in the oblivious model.
For algorithms that rely solely on shared randomness (i.e., no private coins are used) this is also optimal (see \cite[Footnote 5]{gk19}). 

% In Theorem~\ref{thm:rounds_gap}, for our upper bound we don't use private coins either. In the proof of the lower bound, we make no assumptions regarding the \textsf{LOCAL} algorithm being employed. In particular, it is allowed to use private coins.

In our construction, Corollaries~\ref{cor:poly_hardness} and \ref{cor:subexp_hardness} show a trade-off between the hardness assumption we make and the resulting round and randomness complexities. 
Unfortunately, when we make a mild hardness assumption, the length of the shared random string is at least $n^{\Omega(1)}$. But, under a stronger assumption, we are able to show that a string of polylogarithmic length suffices. 
Keep in mind that we also achieve a significantly higher success rate, so the randomness complexities of our construction and \cite{bgk25}'s are incomparable.
Furthermore, we note that $\omega(\log n)$ shared random bits are essential for any non-trivial result in the preset public coins model. When $|r_{\text{pub}}| = O(\log n)$, the description of the $\poly(n)$-time adversary can include a hardcoded input labeling for every one of the $2^{O(\log n)}$ possibilities for $r_{\text{pub}}$. In other words, if the shared randomness is too short, we lose the ``surprise element'' against the (non-uniform) bounded adversary. This implies that our randomness complexity (if we make the subexponential hardness assumption) is close to optimal.
\end{remark}

\begin{remark} [Alternative Model]
\label{remark:alternative_model}
In the preset public coins model, the graph is fixed before the randomness is revealed.
Nonetheless, we mention that the gap demonstrated in Theorem~\ref{thm:rounds_gap} can also be proved under an alternative formulation of the model, where the graph is chosen alongside the input labeling (after $r_{\text{pub}}$ is revealed).
We also note that under this formulation, the preset public coins model with unbounded adversaries is equivalent to the private coins model. Namely, Lemma~\ref{lemma:gap_private} becomes irrelevant in this case.
\end{remark}

\begin{remark}[Non-Uniform Hardness]
\label{remark:non-uniform_adv}
We define hardness with respect to non-uniform solvers.
While this choice is somewhat unconventional, it is motivated by the analogy we draw between the roles of the solver in the distributional problem and the adversary in the \textsf{LCL} problem.
Since the choice of the adversary in our model is allowed to depend on the graph structure, and, in particular, on its size $n$, it is natural to permit different solvers for different input lengths.

A uniform variant of the model can also be considered, wherein a single uniform adversary is required to produce hard inputs for all $n$~vertex graphs, for all large enough $n$'s. 
Under this formulation, the computational assumption in Theorem~\ref{thm:rounds_gap} can be relaxed to hardness against uniform solvers.
\end{remark}

In our second result, we show a necessary condition for \textsf{LCL} problems in which having the input selected by an efficient adversary helps to significantly improve the success probability.

\begin{restatable}{theorem}{probabilitygap}
\label{thm:probability_gap}
Let $\Psi$ be an \textsf{LCL} problem. Assume that there exist a negligible function $\mathsf{negl}(n)$ and a polynomial $p(n)$ such that:
\[
\mathsf{ROUND}^{\mathsf{B}}_{\mathsf{negl}(n)}(\Psi) < \mathsf{ROUND}^{\mathsf{U}}_{1/p(n)}(\Psi)
\]
Where $n$ refers to the graph size. Then, there is a public-coin infinitely-often polynomially hard-on-average problem in $\TFNP / \poly$.
\end{restatable}

That is, in the preset public coins model, for every polynomial-time adversary that decides on the inputs, $\Psi$ can be solved by a \textsf{LOCAL} algorithm with probability $1 -\mathsf{negl}(n)$, with strictly fewer rounds than required for solving $\Psi$ with probability $1 - \frac{1}{p(n)}$ when a computationally unbounded adversary selects the input. Such a problem $\Psi$ implies the hardness of a total search problem.

This gap is required to be huge. The error probability against an efficient adversary is negligible, whereas against an unbounded adversary, we should fail with probability at least $1 / p(n)$ for some polynomial $p(n)$. It is natural to ask whether this assumption is not too artificial. Saying that the success rate gets worse if an unbounded adversary is present shouldn't come as a surprise, but here we demand a significant gap, which might seem restrictive, but we beg to differ. If the gap were smaller, one could argue that the powerful adversary doesn't inflict substantial harm, and the algorithm remains practical even in the unbounded setting. 

We remark that the problem promised by Theorem~\ref{thm:probability_gap} is only \textit{polynomially} hard-on-average. It seems that our techniques do not enable us to go beyond polynomial hardness.\footnote{The instances of the search problem that we devise in Section~\ref{sec:probability_gap} take on the form of $(1^n, r_{\text{pub}})$.
The problem's hardness is closely related to the difficulty a bounded adversary faces when selecting the input labeling for the given \textsf{LCL} problem $\Psi$ on an $n$~vertex graph. 
We defined the bounded adversary to be any $\poly(n)$-time algorithm.
Thus, on instance of length $\poly(n)$ the problem is guaranteed to be hard for any $\poly(n)$-time solver, but not beyond that.}

The fact that the resulting search problem is only \textit{infinitely-often} hard-on-average means that every (non-uniform) \textsf{PPT} solver for the problem succeeds with at most negligible probability for infinitely many values of $n$ (the graph size). 
In fact, this can be strengthened to hold for all large enough $n$ if one is willing to assume a stronger condition. 
Specifically, that the lower bound on the round complexity against unbounded adversaries applies even to ``weaker'' \textsf{LOCAL} algorithms that are required to be correct for infinitely many $n$'s (rather than for all large enough $n$, as is typically required of a \textsf{LOCAL} algorithm).

% Lastly, we note that both of our results are expressed using complexity classes of search problems. It is natural to ask whether this is necessary, and if the results can rephrased using analogous classes of decision problems (e.g., $\NP \cap \coNP$). The proof of Theorem~\ref{thm:rounds_gap} can be easily modified to support hard-on-average decision problems (instead of search problems). On the other hand, the proof of Theorem~\ref{thm:probability_gap} relies on the hardness of finding an input labeling. Namely, not on the hardness of determining whether such an input labeling exists. Therefore, it is unclear whether the proof can be easily adapted so that it yields a hard-on-average \textit{decision} problem.

\subsection{Related Work} 

The role of (private) randomness in solving \textsf{LCL} problems has been thoroughly studied. \cite{ns95} showed that randomness does not help if the number of rounds is $O(1)$. 
Chang, Kopelowitz, and Pettie \cite{ckp16} extended this result to $2^{O(\log^* n)}$ round algorithms. Moreover, they demonstrated that there exists an \textsf{LCL} problem with an exponential gap in the round complexity between the randomized and the deterministic settings (see Ghaffari and Su \cite{gs17} for another such problem).
For a wider class of distributed problems, in which solutions can be verified in $\poly(\log n)$ rounds, Rozho\v{n} and Ghaffari \cite{rg20} proved that any $\poly(\log n)$ round \textsf{LOCAL} algorithm can be derandomized, yielding a deterministic algorithm whose round complexity remains $\mathsf{poly}(\log n)$. 
In other words, \textit{private} randomness doesn't help if the goal is to have polylogarithmic round complexity.  
This stands in sharp contrast to the case of \textit{shared} randomness, as the result of~\cite{bgk25} indicates that certain tasks can be performed by $O(\log n)$ round \textsf{LOCAL} algorithms when shared randomness is available, whereas such efficiency is unattainable without it.
This gives a negative answer to an open question raised in \cite{Gha20}.

Ghaffari and Kuhn \cite{gk19} investigated various settings that involve \textsf{LOCAL} algorithms that use randomness. 
%In the context of private randomness, they discussed a setting with randomness with limited independence. Another interesting environment they examined assumes that each node might not have its own source of randomness, but has 1 random bit somewhere in its neighborhood of radius $\poly(\log n)$. 
In particular, they discussed the use of shared randomness and presented a $\poly(\log n)$ round algorithm for \textit{network decomposition} that relies on $\poly(\log n)$ shared random bits (and no private coins).
Note that the subsequent result of \cite{rg20} suggests that these random bits are actually unnecessary.
In any case, we emphasize that our use of shared random bits is motivated by their usefulness in coordinating distant nodes, 
whereas the motivation of \cite{gk19} seems to be that they give rise to a cleaner argument about the total amount of randomness required across the entire network to accomplish certain tasks.

Our computation model is a modified version of \textsf{LOCAL}, where the nodes run in time polynomial in $n$ (the graph size). This prevents the nodes from outperforming the bounded adversary (which is defined to be a $\poly(n)$-time probabilistic algorithm). 
Likewise, Aldema-Tshuva and Oshman \cite{ao24} studied the closely related class of \textit{local decision} (\textsf{LD}) under this particular restriction. 
Moreover, the literature on distributed graph algorithms contains other types of computational restrictions that can be imposed on the processors. For example, Parter and Yogev \cite{py19a} considered nodes that run in polynomial time relative to their inputs.

Our model fixes the shared randomness before the input is selected. Cohen and Naor~\cite{cn22} defined the \textit{preset public coins} model in the context of two-party communication complexity protocols, while also making the distinction between efficient and unbounded adversaries. One of the goals of this work is to extend the discussion about preset public randomness to another setting, namely, distributed graph algorithms.

Both of our theorems are related to average-case hardness in total search problems.
Importantly, hard-on-average problems in $\NP$ unconditionally imply (public-coin) hard-on-average problems in $\TFNP / \poly$ , as proved by Hub\'{a}cek, Naor, and Yogev \cite{hny17}. 
Furthermore, under complexity-theoretic assumptions, it can be strengthened to average-case hardness in the uniform class $\TFNP$.\footnote{In more detail, under the assumption that there exists a function with complexity $\mathsf{DTIME}(2^{O(n)})$, and $\Pi_2$-circuit complexity $2^{\Omega(n)}$, where $\Pi_2$-circuits are oracle circuits with access to any language in $\Pi_2$, the second level of the polynomial hierarchy.}

While OWFs can be used to construct problems in $\NP$ that are hard-on-average (for uniform solvers), in Impagliazzo's Pessiland \cite{imp95} such hard problems exist, yet OWFs do not. 
%Still in the context of uniform solvers,
Pass and Venkitasubramaniam \cite{pv20} showed that if there are hard-on-average problems in $\NP$, either OWFs exist, or there is an (infinitely-often) hard-on-average problem in $\TFNP$. 
In other words, average-case hardness in $\TFNP$ holds unconditionally in Pessiland.
Furthermore, average-case hardness in $\TFNP$ can be based on other assumptions. For instance, \textit{one-way permutations} (OWP) imply average-case hardness in the sub-class $\PPP$, as shown by Papadimitriou \cite{pap94}.

\subsection{Future Directions}
An intriguing open question is whether a result like Theorem~\ref{thm:probability_gap} can be derived from a gap in the round complexity rather than a gap in the failure probability.
Such a result will be closer in spirit to what can be considered the inverse of Theorem~\ref{thm:rounds_gap}.

From a broader perspective, we believe that the main contribution of this work is pointing out connections between the performance of distributed graph algorithms and the capabilities of a centralized algorithm (which represents the globally knowledgeable adversary).
That is, we shed light on how complexity measures from the distributed and the centralized settings relate.
We hope this work inspires future studies that will explore links between complexity measures that are usually associated with different areas of theoretical computer science.

\section{Technical Overview}

\subsection{Outline of the Proof of Theorem~\ref{thm:rounds_gap}}

We begin with a high-level description of the problem described in \cite{bgk25}, followed by an explanation of the modifications we make in order to prove Theorem~\ref{thm:rounds_gap}.

\paragraph{Harnessing (Oblivious) Shared Randomness.}
%\cite{bgk25}'s Key Ideas.
To demonstrate in a gap in round complexity between settings with and without shared randomness, \cite{bgk25} introduced a special family of graphs. 
In these graphs, shared randomness allows distant nodes to coordinate, enabling them to produce a valid output labeling in $O(\log n)$ rounds with high probability. 
Without shared randomness, however, $\Omega(\sqrt{n})$ rounds are required.
On the other hand, graphs outside the family are exempted from solving the hard task.

The graphs in the family are constructed using simpler building blocks, called \textit{tree-like structures} and \textit{grid structures}. Tree-like structures are subgraphs that take on the form of perfect binary trees, with additional edges that connect nodes that reside in the same layer. A grid structure of dimensions $h \times w$ is a two-dimensional lattice of $h$ rows with $w$ nodes in each. Grid structures satisfying $h \geq w$ are called \textit{vertical grid structures}.
A graph in the family consists of a vertical grid structure of dimensions $h \times w$, on top of which we place $w$ tree-like structures such that they are aligned with the grid's columns. 
%That is, the nodes in every column correspond to the leaves of a tree-like structure. 
We refer to this construction as \textit{augmented grid structure} (or in short, AG). We define it formally in Section~\ref{subsubsec:augmented_grid}.

Nodes in graphs from the specified family must produce outputs where: (1) all nodes in any given row output the same bit, (2) in at least one row, the output must be equal to the input given to the rightmost node. 
To see why this can be described using the \textsf{LCL} framework, observe that requirement~(1) is easily locally verifiable (with each grid node inspecting the outputs of its neighbors), and that requirement~(2) can be checked with the help from the column trees. 
Roughly speaking, the output labeling is deemed valid if the roots output YES, grid nodes output YES either if they don't participate in the right column or if the input equals the output, and inner tree nodes return YES if and only if at least one of their children does so. 
These can be verified locally, and together they imply (2). 
See Fig.~\ref{fig:ag_with_input} for an example.

\begin{figure}[t]
\centering
\resizebox{35ex}{31ex}{
\begin{tikzpicture}[scale=1.75, every node/.style={circle,draw=black,inner sep=0pt,minimum size=6.5mm}]

% AG's grid:
% Column 1
\node (lg11) at (-5+0,0) {0,Y};
\node (lg21) at (-5+-0.1,-1) {1,Y};
\node (lg31) at (-5+-0.2,-2) {1,Y};
\node (lg41) at (-5+-0.3,-3) {0,Y};
\draw[color=black] (lg11) -- (lg21);
\draw[color=black] (lg21) -- (lg31);
\draw[color=black] (lg31) -- (lg41);
% Column 2
\node (lg12) at (-5+1,0) {0,Y};
\node (lg22) at (-5+1-0.1,-1) {1,Y};
\node (lg32) at (-5+1-0.2,-2) {1,Y};
\node (lg42) at (-5+1-0.3,-3) {0,Y};
\draw[color=black] (lg12) -- (lg22);
\draw[color=black] (lg22) -- (lg32);
\draw[color=black] (lg32) -- (lg42);
\draw[color=black] (lg11) -- (lg12);
\draw[color=black] (lg21) -- (lg22);
\draw[color=black] (lg31) -- (lg32);
\draw[color=black] (lg41) -- (lg42);
% Column 3
\node (lg13) at (-5+2,0) {0,Y};
\node (lg23) at (-5+2-0.1,-1) {1,Y};
\node (lg33) at (-5+2-0.2,-2) {1,Y};
\node (lg43) at (-5+2-0.3,-3) {0,Y};
\draw[color=black] (lg13) -- (lg23);
\draw[color=black] (lg23) -- (lg33);
\draw[color=black] (lg33) -- (lg43);
\draw[color=black] (lg12) -- (lg13);
\draw[color=black] (lg22) -- (lg23);
\draw[color=black] (lg32) -- (lg33);
\draw[color=black] (lg42) -- (lg43);
% Column 4
\node[label=right:{$1$}] (lg14) at (-5+3,0) {0,N};
\node[label=right:{$0$}] (lg24) at (-5+3-0.1,-1) {1,N};
\node[label=right:{$1$}] (lg34) at (-5+3-0.2,-2) {1,Y};
\node[label=right:{$0$}] (lg44) at (-5+3-0.3,-3) {0,Y};
\draw[color=black] (lg14) -- (lg24);
\draw[color=black] (lg24) -- (lg34);
\draw[color=black] (lg34) -- (lg44);
\draw[color=black] (lg13) -- (lg14);
\draw[color=black] (lg23) -- (lg24);
\draw[color=black] (lg33) -- (lg34);
\draw[color=black] (lg43) -- (lg44);

% AG's trees:
% Column 1
\node (t11) at (-5-0.4,0.4) {Y};
\node (t31) at (-5-0.6,-2+0.4) {Y};
\node (tu1) at (-5-0.8,0.8) {Y};
\draw[color=black] (lg11) -- (t11);
\draw[color=black] (lg21) -- (t11);
\draw[color=black] (lg31) -- (t31);
\draw[color=black] (lg41) -- (t31);
\draw[color=black] (t11) -- (t31);
\draw[color=black] (t11) -- (tu1);
\draw[color=black] (t31) -- (tu1);
% Column 2
\node (t12) at (-4-0.4,0.4) {Y};
\node (t32) at (-4-0.6,-2+0.4) {Y};
\node (tu2) at (-4-0.8,0.8) {Y};
\draw[color=black] (lg12) -- (t12);
\draw[color=black] (lg22) -- (t12);
\draw[color=black] (lg32) -- (t32);
\draw[color=black] (lg42) -- (t32);
\draw[color=black] (t12) -- (t32);
\draw[color=black] (t12) -- (tu2);
\draw[color=black] (t32) -- (tu2);
% Column 3
\node (t13) at (-3-0.4,0.4) {Y};
\node (t33) at (-3-0.6,-2+0.4) {Y};
\node (tu3) at (-3-0.8,0.8) {Y};
\draw[color=black] (lg13) -- (t13);
\draw[color=black] (lg23) -- (t13);
\draw[color=black] (lg33) -- (t33);
\draw[color=black] (lg43) -- (t33);
\draw[color=black] (t13) -- (t33);
\draw[color=black] (t13) -- (tu3);
\draw[color=black] (t33) -- (tu3);
% Column 4
\node (t14) at (-2-0.4,0.4) {N};
\node (t34) at (-2-0.6,-2+0.4) {Y};
\node (tu4) at (-2-0.8,0.8) {Y};
\draw[color=black] (lg14) -- (t14);
\draw[color=black] (lg24) -- (t14);
\draw[color=black] (lg34) -- (t34);
\draw[color=black] (lg44) -- (t34);
\draw[color=black] (t14) -- (t34);
\draw[color=black] (t14) -- (tu4);
\draw[color=black] (t34) -- (tu4);

\end{tikzpicture}
}
\caption{An example of an \textit{augmented grid structure} (AG) of dimensions $4 \times 4$ with a valid (partial) labeling. 
For brevity, the input labels that correspond to the graph structures themselves are omitted, leaving only the input bits given to the right column (written outside the nodes). The outputs are written inside the nodes. Y (resp. N) stands for YES (resp. NO).}
\label{fig:ag_with_input}
\end{figure}
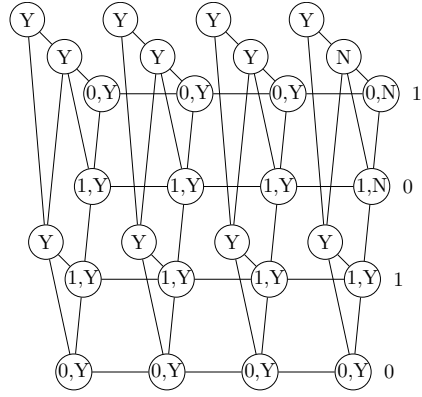

Observe that within $O(\log n)$ rounds in the \textsf{LOCAL} model, the view of the nodes contains the whole tree-like structure they are part of. As such, they are able to infer the index $i \in [h]$ of the row they are located in. If shared randomness $r_{\textnormal{pub}} \in \{ 0,1 \}^*$ is available, the nodes in each row $i$ of the grid can use the $i$-th bit of $r_{\textnormal{pub}}$ as their output. 
Obviously, this algorithm satisfies requirement~(1). Since the party that selects the inputs of the nodes is assumed to be oblivious of the shared randomness, the probability of the output colliding with the input of the rightmost node is $\frac{1}{2}$.
Over $h$ rows, the probability is at least $1 - n^{-c}$ if $h \geq c \log n$ for some $c > 0$. 
If the height of the grid is smaller than that, the use of \textit{vertical} grid ensures that the width is also at most $c \log n$. This allows the nodes to learn the inputs of the entire network within $O(\log n)$ rounds.

On the other hand, if shared randomness is not available, the best the nodes can do is wait for a message from the rightmost node in their row, which takes $\Omega(w)$ rounds. In a grid of width $w = \Omega(\sqrt{n})$, we get the desired lower bound.

We follow the same approach. 
We define a (similar) graph family $\mathcal{G}$ of hard instances and an \textsf{LCL} problem $\Pi$ such that solving $\Pi$ on $\mathcal{G}$ can be done efficiently in one model but not in another, thereby establishing a gap between the two.
Additionally, we show that when the graph doesn't belong to $\mathcal{G}$, the nodes are able to detect that, which helps to ensure that $\Pi$ is promise-free.

\paragraph{Locally Checking Turing Machine Computations.}

In Section~\ref{subsubsec:tm_encode_structure} we define \textit{\textsf{TM} encoding structures} to be grid structures that have an input labeling that represents an execution of a Turing machine. 
In more detail, the structure is defined with respect to a specific (one-tape) Turing machine and an input, and its main component will be a grid structure large enough for writing the whole history of the tape as part of the nodes' labels. 
The horizontal axis (from right to left) is associated with time, namely, each column refers to a specific step in the execution, and the vertical axis corresponds to the tape of the machine, such that each row is related to an index in the tape. This is illustrated in Fig.~\ref{fig:tm_structure_labels} (see Fig.~\ref{fig:tm_encoding} for a concrete example).

The idea of embedding transcripts of Turing machine computations into labelings of \textsf{LCL} problems builds on prior work. 
Notably, \cite{ns95} demonstrated that simulating a Turing machine can be formulated as an \textsf{LCL} problem. Similar techniques were used by Balliu, Brandt, Olivetti, and Suomela \cite{bbos20}, as well as Chang \cite{cha23}. 
In all these works, however, the nodes were required to generate the transcript of the computation themselves.
Our approach takes a different direction.
We assume that the input labels already contain the full transcript of the computation. The role of the nodes is reduced to simply \textit{verifying} the correctness of the computation.
A similar idea appears in the work of Balliu, Brandt, Chang, Olivetti, Rabie, and Suomela \cite{bbco19}, albeit with a different underlying graph structure--a path instead of a grid.
This verification can be performed locally, leveraging the fact that Turing machines are local in their nature (as the head of the machine can move only to adjacent cells). 
If the graph has an appropriate structure and the computation encoded in the nodes' labels is correct, the graph can qualify as a \textit{\textsf{TM} encoding structure}.

\begin{figure}[t]
\centering
\resizebox{41ex}{30ex}{
\begin{tikzpicture}[align=center]
\node [draw, xshift=0.05cm, yshift=-0.02cm, minimum width=0.9cm, minimum height=0.5cm]  {$\#_{\mathsf{L}}$};
\node [xshift=0.11cm, yshift=-1.67cm] {
\begin{tikzpicture}
  \node [draw, minimum width=2.7cm, minimum height=0.9cm, rotate=-90] {\texttt{input}};    
\end{tikzpicture}
};
\node [draw, yshift=-3.74cm, xshift=0.05cm, minimum width=0.9cm, minimum height=1.46cm] {$\#_{\mathsf{R}}^+$};
\node [xshift=-0.74cm, yshift=-2.1cm] {
\begin{tikzpicture}
  \node [draw, minimum width=4.75cm, minimum height=0.8cm, rotate=-90] {$\textnormal{TAPE}^{\mathcal{M}}_{\texttt{input}}(1)$};
\end{tikzpicture}
};
\node [xshift=-1.54cm, yshift=-2.1cm] {
\begin{tikzpicture}
  \node [draw, minimum width=4.75cm, minimum height=0.8cm, rotate=-90] {$\textnormal{TAPE}^{\mathcal{M}}_{\texttt{input}}(2)$};
\end{tikzpicture}
};
\node [draw, xshift=-3.5cm, yshift=-2.1cm, minimum width=3cm, minimum height=4.75cm] {...};
\node [xshift=-5.34cm, yshift=-2.1cm] {
\begin{tikzpicture}
  \node [draw, minimum width=4.75cm, minimum height=0.8cm, rotate=-90] {$\textnormal{TAPE}^{\mathcal{M}}_{\texttt{input}}(T^{\mathcal{M}}(\texttt{input}))$};
\end{tikzpicture}
};
% Arrows
\draw[->,blue] (0.75, 0.28) -- (0.75,-4.5) node [midway, sloped, above] {tape};
\draw[->,blue] (0.36, -4.75) -- (-5.85,-4.75) node [midway, sloped, below] {time};
\end{tikzpicture}
}
\caption{An input labeling associated with a \textit{\textsf{TM} encoding structure}, for a Turing machine $\mathcal{M}$ that on input $\texttt{input}$ runs for $T^{\mathcal{M}}(\texttt{input})$ steps. Assume that the input is enclosed by a single $\#_{\mathsf{L}}$ on the left, and sufficiently many $\#_{\mathsf{R}}$'s on the right (in order to fill a tape large enough to support $\mathcal{M}$'s space requirements). $\textnormal{TAPE}^{\mathcal{M}}_{\texttt{input}}(j)$ stands for the contents of the tape at step $j$.}
\label{fig:tm_structure_labels}
\end{figure}
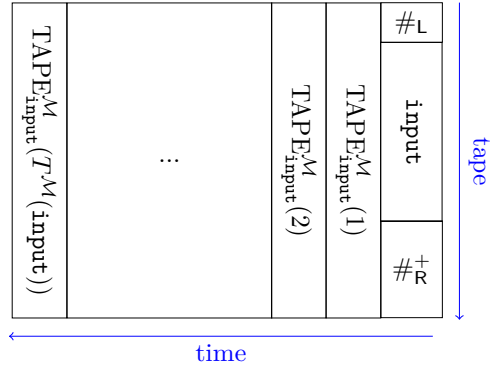

Thus far, we haven't mentioned what Turing machine we aim to implement as part of our graphs. 
In Theorem~\ref{thm:rounds_gap} we assume the existence of a public-coin hard-on-average problem in $\TFNP$. Let $\mathcal{S}$ be the relation that defines the search problem, and $\mathcal{D}$ be the distribution over the instances. $\mathcal{D}$ is associated with an efficient sampler $D$ that receives as input a uniformly random string and outputs an instance. 
Moreover, since $\mathcal{S} \in \TFNP$, there is an efficient verifier $V$ that on input $(x, w)$ checks whether the pair belongs to $\mathcal{S}$. 
On input $(r, w)$, our Turing machine $\mathcal{M}$ samples an instance by invoking the sampler $D$ on $r$. Let $x$ be the resulting instance. It then uses $V$ to verify whether $(x, w) \in \mathcal{S}$. If this is the case, the output is $r$. Otherwise, it returns a default string $\Vec{1}$.

The leftmost column of the grid of the \textit{\textsf{TM} encoding structure} contains the contents of the tape when the machine reaches a final state. That is, it stores the output. To set the labels of this column to some string $r$, the adversary must find a solution to the instance $x = D(r)$.
The hardness of $(\mathcal{S}, \mathcal{D})$ implies that this task should be computationally challenging when $x \sim \mathcal{D}$ (equivalently, when $r$ is sampled uniformly at random).

In practice, we use a somewhat more intricate graph structure, which we refer to as \textit{small-diameter augmented \textsf{TM} encoding structure} (SATM), defined in Section~\ref{subsec:satm}. It resembles an AG with an additional tree-like structure attached on the top of it. The technical details are omitted from this simplified exposition.

\paragraph{Hard Instances.} 

The graphs in the family $\mathcal{G}$ of hard instances that we define are not exactly the same as those in \cite{bgk25}. 
We extend their \textit{augmented grid structure} (AG) with a \textit{small-diameter augmented \textsf{TM} encoding structure} (SATM), a subgraph dedicated to simulating the execution of a Turing machine. 
This gadget is glued to the AG, such that its rightmost column is connected to nodes with labels that stand for the output of the Turing machine. 
We give a formal definition of the graph family $\mathcal{G}$ in Def.~\ref{def:graph_family_g}.
Fig.~\ref{fig:simplified_hard_instance} depicts such a graph along with a partial input labeling of the SATM, corresponding to an unspecified Turing machine computation. 
For a version that highlights the relevant structural components, see Fig.~\ref{fig:graph_family_g}.

\begin{figure}[t]
\centering
%\resizebox{41ex}{30ex}{
\includestandalone{figures/full_with_tm_inputs}
%}
\caption[]{An example illustrating the structure of graphs in $\mathcal{G}$. Namely, a SATM connected to an AG. The input labeling refers to the transcript of some unspecified Turing machine.\footnotemark~Each node is labeled $(t,h,s)$, where $t$ is a symbol in the tape, $h \in \{\textsf{Head}, \bot\}$ indicates head presence, and $s$ contains the state if $h = \textsf{Head}$, otherwise $s = \bot$. For brevity, the figure shows only $t$, with $h$ implied by the node color.}\label{fig:simplified_hard_instance}
\end{figure}
\footnotetext{Note that to for a graph to truly belong to $\mathcal{G}$, the computation must be consistent with the transition function of the Turing machine $\mathcal{M}$ defined earlier, which depends on the distributional problem $(\mathcal{S}, \mathcal{D})$.}

In the original problem, the inputs of the nodes in the right column play a significant role. Now, the string obtained from concatenating the inputs of the nodes in this column corresponds to the output of the Turing machine. 
Over the \textit{augmented grid structure}, our \textsf{LCL} problem $\Pi$ is defined mostly the same as \cite{bgk25}'s problem, with the difference that it now depends on the output of the Turing machine. 

Above, we noted that setting the output of the machine to a string $r$ is likely to be a hard computational task whenever $r$ is sampled uniformly at random.
A key observation is that the nodes have a strategy that forces $r$ to be distributed this way.
This is achieved by adopting the same strategy as in \cite{bgk25}. That is, using $r_{\textnormal{pub}}$ for the output. 
In order to win, the adversary must set the output of the Turing machine to the complement $\overline{r_{\textnormal{pub}}}$, which is also uniformly distributed. This is the main idea behind our upper bound in Theorem~\ref{thm:rounds_gap} (see Section~\ref{subsec:ub}).

Truthfully, as we aim to minimize the number of shared random bits, the nodes' strategy will be to use $r_{\textnormal{pub}}$'s bits only in a limited number of rows, while the remaining rows will use some fixed value. We will show that this is enough for achieving a sufficiently high success probability.

Unlike efficient adversaries, a computationally unbounded adversary should be able to solve the distributional problem. 
Without modifications, it could have been that the sampled instance $x$ has no solutions at all. We rule this out by requiring the problem to be \textit{total}. This is necessary for proving the lower bound in Section~\ref{subsec:lb}. 
The lower bound we prove is $\Omega(n^c)$ for some constant $c$, which depends on the problem $\mathcal{S}$.
Note that this holds regardless of the length of the preset public random string $r_{\textnormal{pub}}$ we have access to.

\paragraph{Keeping $\Pi$ Promise-Free.}
To complete the proof of the upper bound, it is crucial to ensure that $\Pi$ is easy on graphs outside $\mathcal{G}$.
%We show that for any such graph, a valid output labeling can be generated in a round-efficient manner. 
Since membership in $\mathcal{G}$ is characterized by local constraints, any such graph $G$ must contain at least one node whose neighborhood serves as evidence for the fact that $G \notin \mathcal{G}$. 
Adopting a beautiful technique demonstrated in \cite{bgk25}, we allow the output labels to encode chains of pointers leading to violations of the local constraints. 
That is, nodes in erroneous components may output pointers that, when followed, eventually reach a witness of the violation.
This can be viewed as a locally verifiable proof for the existence of an error.

The remaining nodes--those not participating in the pointer chains--must lie in subgraphs free of local errors. 
The constraints are designed so that such a subgraph is always an \textit{augmented graph structure} (AG) or a \textit{small-diameter augmented \textsf{TM} structure} (SATM). 

In other words, nodes that do not belong to an AG or to a SATM can produce a proof of incorrectness, while nodes within such structures are unable to forge one.
Importantly, a key feature of this output labeling is that it can be obtained in $O(\log n)$ communication rounds. We distinguish between the following cases:

\begin{itemize}

\item \textit{Erroneous components}. The proof of incorrectness serves as a valid output labeling for our problem, so no further work is needed and the nodes may terminate. 

\item \textit{Valid AG and SATM structures, properly connected.} Namely, the graph is a member of the hard instances family $\mathcal{G}$. This scenario was addressed earlier.

\item \textit{AG without a SATM.}
The nodes of the AG still follow the strategy outlined earlier, but in this case $\Pi$'s constraints don't require the outputs to be related to the inputs in the right column. Clearly, no additional communication is needed.

\item \textit{SATM without an AG.}
The problem is defined such that no further steps are required. 

\end{itemize}

To summarize, for every $G \notin \mathcal{G}$, a valid output labeling for the problem $\Pi$ can be generated within $O(\log n)$ rounds, completing the proof.

\subsection{Outline of the Proof of Theorem~\ref{thm:probability_gap}}

In Theorem~\ref{thm:probability_gap}, we are given an \textsf{LCL} problem $\Psi$ that can be solved within $r(n)$ rounds with negligible error probability $\mathsf{negl}(n)$, assuming that the adversary selecting the inputs is efficient, for some function $r(n)$ and sufficiently large graph sizes $n$. 
However, against unbounded adversaries, no \textsf{LOCAL} algorithm running for the same number of rounds $r(n)$ can ensure that for all large enough values of $n$ the failure probability is at most inverse polynomial $1 / p(n)$.
We turn this problem $\Psi$ into a public-coin infinitely-often (polynomially) hard-on-average distributional search problem. That is, for any non-uniform probabilistic polynomial-time (\textsf{PPT}) algorithm, the success probability is at most negligible, for infinitely many instance lengths.

\paragraph{Intuition.}

Let $\mathcal{A}$ be the \textsf{LOCAL} algorithm that performs well against bounded adversaries. 
As for unbounded adversaries, there is an infinite sequence of graph sizes $\mathcal{H} = \{ n_i \}_{i=1}^{\infty}$, such that for every $n_i$ there is a ``hard'' $n_i$~vertex graph $G_{n_i}$ on which $\mathcal{A}$ fails with probability $1 / p(n_i)$. 
Denote by $\mathcal{G} = \{ G_n \}_{n \in \mathbb{N}}$ a family of these graphs (for $n \notin \mathcal{H}$, use arbitrary $n$~vertex graphs).

The search problem $\mathcal{S}$ will depend on the algorithm $\mathcal{A}$, the sequence $\mathcal{H}$ and the graph family $\mathcal{G}$. The instances will be pairs of a graph size and a preset public random string $r_{\textnormal{pub}}$. 
Conceptually, the problem seeks to capture the difficulty faced by the adversary that chooses the input labelings. The solver takes on the role of the adversary. 

\paragraph{A Flawed Attempt.}

First, consider the following (flawed) attempt to define $\mathcal{S}$. A solution for an instance $(1^n, r_{\textnormal{pub}})$, for $n$ that belongs to $\mathcal{H}$, is an input labeling $x$ that satisfies: algorithm $\mathcal{A}$ fails to generate a valid output labeling when invoked on graph $G_n \in \mathcal{G}$ with input labeling $x$ and preset public coins $r_{\textnormal{pub}}$ with probability that is at least an inverse polynomial. Whenever $n \notin \mathcal{H}$, we permit a trivial solution (say, $x = 0^n$).

Informally, the hardness of the problem arises from the intractability of finding ``bad'' input labelings that make $\mathcal{A}$ fail with non-negligible probability.

The issue with this approach is that it ignores the nodes' private coins. This implies that the verification of solutions should be randomized; however, this needs to be avoided, as we aim to devise a search problem in $\TFNP/ \poly$, which means that solutions must be deterministically verifiable.

\paragraph{Characterizing Multisets.}

Our workaround involves sampling a collection of private coins beforehand, which will be a part of the definition of our search problem $\mathcal{S}$. Then, the fraction of samples on which the algorithm $\mathcal{A}$ fails will approximate the actual error probability. We formulate this idea through the framework of \textit{characterizing multisets}.
This technique has been widely used, e.g., in the works of Babai and Kimmel~\cite{bk97} (who introduced it) and \cite{cn22}. 
It can be viewed as a method for describing the strategy of a randomized player while using a relatively small multiset of its outputs.

In our context, we want to succinctly characterize the way $\mathcal{A}$ behaves when the nodes may toss private coins, unknown beforehand.
To this end, we prepare a polynomial-sized multiset $R_n = \{ r_{\textnormal{sec}}^{(i)} \}_{i}^{l(n)}$ of strings that stand for the private coins of the nodes. 
On a graph $G_n \in \mathcal{G}$, input labeling $x$ and preset public randomness $r_{\textnormal{pub}}$, the \textit{characterizing multiset} consists of the output labels of $G_n$'s nodes if $\mathcal{A}$ was invoked in this setup, with $r_{\textnormal{sec}}^{(i)}$ being the private coins of the nodes, for all $i \in [l(n)]$.

By ``characterizing'', we mean that if the nodes produce an invalid output labeling with some probability, then we want the fraction of invalid output labelings in the characterizing multiset to approximate that probability. In general, this framework allows us to estimate the probability that an event occurs with respect to the actual distribution over the private coins $r_{\textnormal{sec}}$, while relying solely on a publicly known multiset $R_n$ for our computations. 

\paragraph{Corrected Problem Definition.}

Let $\mathcal{R} = \{ R_n \}_{n \in \mathbb{N}}$ be the collection of multisets of private coins chosen above. 
The search problem $\mathcal{S}$ depends on $\mathcal{A}$, $\mathcal{G}$, $\mathcal{H}$ and $\mathcal{R}$. 
Now, to verify whether $((1^n, r_{\textnormal{pub}}), x) \in \mathcal{S}$ (for $n \in \mathcal{H}$), we compute the \textit{characterizing multiset} and count how many of the resulting output labelings are invalid. If it is greater than some threshold, we accept.
For $n \notin \mathcal{H}$, we accept if the solution is $0^n$.

\paragraph{Hardness.}

Our goal is to leverage the fact that a bounded adversary is unlikely to generate a ``bad'' input, in order to argue that our search problem is hard-on-average.
Note that in the context of \textsf{LCL} problems the failure probability is defined over the choice of the private coins $r_{\textnormal{sec}}$, the preset public coins $r_{\textnormal{pub}}$ and the adversary's coins. 
In contrast, for a solver of the distributional problem, success probability is defined only over $r_{\textnormal{pub}}$ (and the solver's coins).
To address this inconsistency, the main technical part of the proof is a transformation from solvers of the distributional search problem to adversaries for the \textsf{LCL} problem $\Psi$. 
Specifically, if the solver finds an input labeling $x$ on which inverse polynomial fraction of the multiset is invalid (where the success probability is taken over the choice of the instance $r_{\textnormal{pub}}$ and the solver's coin tosses), then there is an adversary for $\Psi$ that succeeds with inverse polynomial probability (over the choice of $r_{\textnormal{pub}}$, $r_{\textnormal{sec}}$ and the adversary's coins). 
Moreover, we show that this adversary can be implemented in polynomial time. 
Since we assume that no efficient adversary can cause the \textsf{LOCAL} algorithm $\mathcal{A}$ to fail with probability higher than $\mathsf{negl}(n)$, we get a contradiction. 
This is the core idea behind showing why our search problem $\mathcal{S}$, coupled with the uniform distribution, yields a hard-on-average problem. 
As a by-product, we get that the problem is \textit{public-coin} hard-on-average.

\paragraph{Totality.}

The resulting search problem is \textit{non-trivial}, meaning that for all $n \in \mathbb{N}$, there is a non-negligible probability that a solution exists for the uniformly sampled instance $(1^n, r_{\text{pub}})$. 
For $n \in \mathcal{H}$, our choice of $\mathcal{H}$ implies that with at least such probability, an input labeling on which $\mathcal{A}$ fails is guaranteed to exist. As for $n \notin \mathcal{H}$, a solution exists simply due to the definition of $\mathcal{S}$.

Be aware that this does not yet mean that it is a \textit{total} search problem, i.e., a problem in which every instance has a solution, and an additional step (that we omit from this exposition) is required.

\section{Preliminaries}

\subsection{Distributed Algorithms}
\label{subsec:distributed}

\paragraph{Graph Notation.} 
An undirected graph $G$ is a tuple $(V, E)$, where $V$ is the set of nodes and $E$ is the set of edges. Typically, the size of $V$ is denoted by $n$. Two nodes $u$ and $v$ are connected in $G$ if and only if $\{u,v\} \in E$. The \textit{diameter} of $G$ is the maximal shortest path between two nodes in the graph. The \textit{radius} of a node $v$ in a connected graph $G$ is the distance from $v$ to the node farthest away from it. For $t \in \mathbb{N}$, the \textit{$t$-neighborhood} of a node $v\in V$ is the set of nodes reachable from $v$ through paths of length at most $t$. The subgraph induced by the nodes in the $t$-neighborhood of $v$ is denoted by $N(v,t)$. A \textit{half-edge} is a pair $(v,e)$ for $v \in V$, $e \in E$, where $v \in e$ (that is, edge $e$ is connected to node $v$). 

\paragraph{\textsf{LOCAL} Model.} 
A network of processors is modeled as an undirected graph $G$ where the nodes can communicate only with their neighbors in synchronous rounds. In each round, messages of unlimited length are exchanged between pairs of nodes connected via an edge. After receiving messages from the neighbors, a node performs a local computation (based on its own input and the messages it received up to this point), at the end of which it either halts and returns its output, or determines the messages to be sent to its neighbors in the next round. Notice that within $t$ rounds, a message can reach a destination at most $t$ edges away. Therefore, in a $t$ round \textsf{LOCAL} algorithm, the nodes can collect information about their $t$-neighborhoods. Such an algorithm can be viewed as a mapping from $t$-neighborhoods to output labels. The nodes are assigned with $O(\log n)$ bits long unique identifiers. No failures take place during the execution. The graph size $n$ (or some polynomial upper bound on it) is assumed to be provided as input to all nodes (although we may not always explicitly indicate this in our notation). Each node may toss (private) coins. That is, its strategy can be randomized.

It is common to assume that the nodes are computationally unbounded. We deviate from the original model defined by Linial \cite{lin92} (see also Peleg \cite{pel00}), and restrict the processors to be probabilistic $\poly(n)$-time machines. As a result, the length of the messages and the randomness complexity of the nodes is at most $\poly(n)$. Whenever we say \textsf{LOCAL}, we refer to this special variant.

For concreteness, a \textsf{LOCAL} algorithm $\mathcal{A}$ can be modeled as separate Turing machines running in each node, equipped with 4 tapes: an input tape (that contains the input given to the node), a private randomness tape (with the private coins of the node), a communication tape (where neighboring nodes write their messages after each round) and a work tape. In settings where public randomness is available, we assume that the nodes have a fifth (read-only) tape with that randomness written on it.

\paragraph{\textsf{LCL}.} 
One line of research of distributed graph problems considers distributed tasks that can be formulated using the framework of \textit{locally checkable labelings} (\textsf{LCL}). The motivation is to have solutions that the network can verify ``quickly''. By ``quickly'', we mean $O(1)$ rounds of communication.

\begin{definition}
\label{def:lcl}
A \textit{locally checkable labeling} (\textsf{LCL}) $\Pi$ is a tuple $(\mathcal{V}_{\textnormal{input}}, \mathcal{E}_{\textnormal{input}}, \mathcal{V}_{\textnormal{output}}, \mathcal{E}_{\textnormal{output}}, \mathcal{C}, t)$, $t \in \mathbb{N}$.

\begin{itemize}[nosep]

\item $\mathcal{V}_{\textnormal{input}}, \mathcal{E}_{\textnormal{input}}$ are finite sets of input labels for the nodes and the half-edges, respectively.
Similarly, $\mathcal{V}_{\textnormal{output}}, \mathcal{E}_{\textnormal{output}}$ are finite sets of output labels.

\item $\mathcal{C}$ is the finite set of pairs $(H, s)$, where $H$ is a labeled\footnote{Every node $v$ in $H$ has one label in $\mathcal{V}_{\textnormal{input}}$ and one in $\mathcal{V}_{\textnormal{output}}$. Every pair $(v, e)$ of a node $v$ and an edge $e$ from $H$, where $v \in e$, has one label in $\mathcal{E}_{\textnormal{input}}$ and one in $\mathcal{E}_{\textnormal{output}}$.} graph $H$ and $s$ is a node in $H$ with radius $t$.
\end{itemize}

An \textit{instance} of a problem $\Pi$ is a graph $G=(V,E)$ with maximal degree $O(1)$, together with an input labeling $\mathcal{I} = (\mathcal{I}_V, \mathcal{I}_E)$, which is a pair of assignments $\mathcal{I}_V: V \rightarrow \mathcal{V}_{\textnormal{input}}$, $\mathcal{I}_E: V \times E \rightarrow \mathcal{E}_{\textnormal{input}}$.

A \textit{solution} is an output labeling $\mathcal{O} = (\mathcal{O}_V, \mathcal{O}_E)$, where $\mathcal{O}_V: V \rightarrow \mathcal{V}_{\textnormal{output}}$ and $\mathcal{O}_E: V \times E \rightarrow \mathcal{E}_{\textnormal{output}}$, such that for every $v \in V$ the subgraph induced by its $t$-neighborhood $N(v,t)$ (as an $(\mathcal{I}, \mathcal{O})$-labeled graph) satisfies $(N(v,t), v) \in \mathcal{C}$.

A (deterministic) algorithm \textit{solves} the problem $\Pi$ if for every graph $G$, assignment of unique identifiers and input labeling $\mathcal{I}$, the outputs associated with the nodes and the half-edges $\mathcal{O}$ is a valid \textit{solution}.
\end{definition}

When there is no interest in having a labeling for both the nodes and the half-edges, an \textsf{LCL} problem can be defined while specifying only one of them, and ignoring the other.

\paragraph{Path Tracing Function.} 
For a sequence of edges $e_1, ..., e_k$ where $e_i = \{ v_i, v_{i+1} \}$, i.e. a path of length $k$, and $\{L_i\}_{i =1}^k$ where $L_i$ is the label of the half-edge $(v_i, e_i)$ for any $i$, the \textit{path tracing function} $f$ satisfies:
\[
f(v_1, L_1, ..., L_k) = \begin{cases}
v_{k+1} &\quad \text{if the path $(v_1, ..., v_{k+1})$ exists and unique} \\
\bot &\quad \text{otherwise}
\end{cases}
\]

\paragraph{Notation.}
For every $v \in V$, we use $x_v \in \mathcal{V}_{\text{input}}$ to refer to its input label. Similarly, for every half-edge $(v,e) \in V \times E$ (such that $v \in e$) the input label is $x_{(v, e)} \in \mathcal{E}_{\text{input}}$. To simplify the notation, we refer to the concatenation of them all, $x:= x_{v_1} \circ ... \circ x_{v_n} \circ x_{h_1} \circ ... \circ x_{h_m}$, where $V = \{v_1, ..., v_n\}$, and $\{ h_1, ..., h_m\}$ is the set of half-edges. We do the same with the output labeling. $y = y_{v_1} \circ ... \circ y_{v_n} \circ y_{h_1} \circ ... \circ y_{h_m}$, where $y_{v_i} \in \mathcal{V}_{\textnormal{output}}$ and $y_{h_j} \in \mathcal{E}_{\textnormal{output}}$.

\paragraph{Private Coins Randomness.} 
In a randomized \textsf{LOCAL} algorithm $\mathcal{A}$, every node $v \in V$ has access to independent random bits $r_{\textnormal{sec}, v}$. Oftentimes, we denote by $r_{\textnormal{sec}}$ the concentration of all the nodes' private coins (i.e.\ $r_{\textnormal{sec}} := r_{\textnormal{sec}, v_1} \circ ... \circ r_{\textnormal{sec}, v_n}$). We usually take the number of nodes $n$ as a parameter. Then, we denote by $b^{\mathcal{A}}_{\textnormal{sec}}(n)$ the upper bound on the number of random bits that can be used by a single node.\footnote{In most of the \textsf{LOCAL} model literature, there is typically no bound on the growth rate of $|r_{\textnormal{sec}, v}| \leq b^{\mathcal{A}}_{\textnormal{sec}}(n)$. However, in our setting, the nodes run in $\poly(n)$ time, implying that $b^{\mathcal{A}}_{\textnormal{sec}}(n) \leq \poly(n)$.}

\begin{definition}[Private Coins Randomized Algorithm]
\label{def:private_coin}
Let $\Pi$ be an \textsf{LCL} problem.

When the nodes are provided with private randomness $r_{\textnormal{sec}}$ of length $b^{\mathcal{A}}_{\textnormal{sec}}(n)$, a \textsf{LOCAL} algorithm $\mathcal{A}$ \textit{solves $\Pi$ with probability $1-\epsilon(n)$ in $r(n)$ rounds} if for all large enough $n\in \mathbb{N}$, for every $n$~vertex graph $G$ and input labeling $x$, the output labeling generated by $\mathcal{A}$ within $r(n)$ rounds is valid with probability $1 - \epsilon(n)$ (taken over the private coins). That is,
\begin{equation*}
\label{eq:private_coins}
\forall G, x: \quad
\Pr_{r_{\textnormal{sec}}}[\mathcal{A}(G, x, r_{\textnormal{sec}})\text{ is valid}] \geq 1 - \epsilon(n)
\end{equation*}
In this case, it holds that $\mathsf{ROUND}^{\mathsf{priv}}_{\epsilon(n)}(\Pi) \leq r(n)$.
\end{definition}

\paragraph{The Oblivious Model.}
In a randomized \textsf{LOCAL} algorithm $\mathcal{A}$ in the \textit{oblivious mode}, we allow all of the above (in particular, private coins are available), and on top of that, all the nodes have access to a $b^{\mathcal{A}}_{\textnormal{pub}}(n)$ bits long random string $r_{\textnormal{pub}}$.\footnote{As in the case of the private coins, $|b^{\mathcal{A}}_{\textnormal{pub}}(n)| \leq \poly(n)$. We usually aim to minimize $|b^{\mathcal{A}}_{\textnormal{pub}}(n)|$. See Remark~\ref{remark:randomness_complexity} for more details.}
The public string is guaranteed to be independent of the private coins, and, for simplicity, we assume that it is uniformly distributed. Furthermore, the adversary that selects the inputs is assumed to be oblivious to the public string. We put no restrictions on the computational power of the adversary in this case.

\begin{definition}[Algorithm in the Oblivious Model]
\label{def:oblivious}
Let $\Pi$ be an \textsf{LCL} problem. 

A \textsf{LOCAL} algorithm $\mathcal{A}$ \textit{solves $\Pi$ with probability $1-\epsilon(n)$ in $r(n)$ rounds in the oblivious model} if for all large enough $n\in \mathbb{N}$, for every $n$~vertex graph $G$ and input labeling $x$, the output labeling generated by $\mathcal{A}$ within $r(n)$ rounds is valid with probability $1 - \epsilon(n)$ (taken over both the public and the private coins). That is,
\begin{equation*}
\label{eq:oblivious}
\forall G, x: \quad
\Pr_{r_{\textnormal{pub}}, r_{\textnormal{sec}}}[\mathcal{A}(G, x, r_{\textnormal{pub}}, r_{\textnormal{sec}})\text{ is valid}] 
\geq 1 - \epsilon(n)
\end{equation*}
In this case, it holds that $\mathsf{ROUND}^{\mathsf{O}}_{\epsilon(n)}(\Pi) \leq r(n)$.
\end{definition}

\paragraph{Preset Public Coins Model.} 
Unlike the oblivious model, now the shared randomness and the input are no longer independent. As the name suggests, the public coins are determined in advance, allowing the adversary that selects the inputs to make decisions based on $r_{\textnormal{pub}}$. This can be described as a game between the adversary and the nodes in a graph $G$ that runs a \textsf{LOCAL} algorithm $\mathcal{A}$, where the adversary may depend on $G$ and $\mathcal{A}$.
\begin{enumerate}

\item A public string $r_{\textnormal{pub}}$ is sampled according to the uniform distribution.

\item 
After seeing $r_{\textnormal{pub}}$ (and possibly tossing its own coins), the adversary selects an input labeling $x$ ($O(1)$ bits per node).

\item The nodes toss their private coins $r_{\textnormal{sec}}$.

\item On input $x$, preset public randomness $r_{\textnormal{pub}}$ and private randomness $r_{\textnormal{sec}}$, the nodes in the network $G$ execute the distributed algorithm $\mathcal{A}$. In the end, they return an output labeling $y$ ($O(1)$ bits per node).
\end{enumerate}

Note that the preset public coins are still assumed to be independent of the private coins tossed by the nodes. 

Pay attention to the change in the variables associated with the universal quantifier. In the oblivious model, the success probability (over the choice of $r_{\textnormal{pub}}$ and $r_{\textnormal{sec}}$) is required to be sufficiently high \textit{for all} input labelings. On the other hand, in the preset public coins model the success probability is required to be high only for input labelings that can be generated by the adversary.

We model the adversary as an algorithm $\mathcal{B}_{\mathcal{A}, G}$, that may depend on the graph size $n$, the graph itself $G$ and the description of the \textsf{LOCAL} algorithm $\mathcal{A}$ (to be executed by the nodes).~\footnote{We remark that this adversarial model is quite powerful, because there could be a scenario where no efficient adversary succeeds on a large portion of the $n$~vertex graphs, yet for every $n$~vertex graph there exists an adversary that causes $\mathcal{A}$ to fail. We argue that this is the appropriate notion for our setting, as the topology of the network is allowed to be fixed well in advance, and our focus is on the computational effort required for coming up with an input labeling only after $r_{\text{pub}}$ is revealed.} 
% To keep the notation concise, we stick with writing $\mathcal{B}$ (and avoid mentioning $n$, $G$ and $\mathcal{A}$).
$\mathcal{B}_{\mathcal{A}, G}$ gets the preset public randomness $r_{\textnormal{pub}}$ as an input. Efficiency is measured with respect to the graph size $n$. Thus, the input is augmented with $1^n$.
% Be aware that the three can be represented in $\poly(n)$ bits. 

Our goal is to make a distinction between the case where $\mathcal{B}_{\mathcal{A}, G}$ is efficient and the case where it has unlimited computational resources. 
Keep in mind that whenever $\mathcal{B}_{\mathcal{A}, G}$ is said to be efficient, it can also be assumed to be randomized. In the unbounded case, $\mathcal{B}_{\mathcal{A}, G}$ is deterministic.

\begin{definition}[Algorithm in the Preset Public Coins Model with Bounded Adversaries]
\label{def:bounded_adv}
Let $\Pi$ be an \textsf{LCL} problem. 

A \textsf{LOCAL} algorithm $\mathcal{A}$ \textit{solves $\Pi$ with probability $1-\epsilon(n)$ in $r(n)$ rounds in the preset public coins model with bounded adversaries} if for all large enough $n\in \mathbb{N}$, for every $n$~vertex graph $G$ and any probabilistic polynomial-time (\textsf{PPT}) adversary $\mathcal{B}_{\mathcal{A}, G}$ that selects the input after seeing the preset public coins $r_{\textnormal{pub}}$, the output labeling generated by $\mathcal{A}$ within $r(n)$ rounds is valid with probability $1 - \epsilon(n)$ (taken over the preset public coins, the private coins and the internal coin tosses of the adversary). That is,
\begin{equation*}
\label{eq:bounded}
\forall G , \text{ } \mathsf{PPT} \text{ } \mathcal{B}_{\mathcal{A}, G}: \quad
\Pr_{\substack{\text{$\mathcal{B}_{\mathcal{A}, G}$'s coins}\\r_{\textnormal{pub}}, r_{\textnormal{sec}}}}[\mathcal{A}(G, \mathcal{B}_{\mathcal{A}, G}(1^n, r_{\textnormal{pub}}), r_{\textnormal{pub}}, r_{\textnormal{sec}})\text{ is valid}] \geq 1 - \epsilon(n)
\end{equation*}
In this case, it holds that $\mathsf{ROUND}^{\mathsf{B}}_{\epsilon(n)}(\Pi) \leq r(n)$.
\end{definition}

\begin{definition}[Algorithm in the Preset Public Coins Model with Unbounded Adversaries]
\label{def:unbounded_adv}
Let $\Pi$ be an \textsf{LCL} problem. 

A \textsf{LOCAL} algorithm $\mathcal{A}$ \textit{solves $\Pi$ with probability $1-\epsilon(n)$ in $r(n)$ rounds in the preset public coins model with unbounded adversaries} if for all large enough $n\in \mathbb{N}$, for every $n$~vertex graph $G$ and any computationally unbounded adversary $\mathcal{B}_{\mathcal{A}, G}$ that selects the input after seeing the preset public coins $r_{\textnormal{pub}}$, the output labeling generated by $\mathcal{A}$ within $r(n)$ rounds is valid with probability $1 - \epsilon(n)$ (taken over the preset public and the private coins). That is,
\begin{equation*}
\label{eq:unbounded}
\forall G ,  \mathcal{B}_{\mathcal{A}, G}: \quad
\Pr_{\substack{r_{\textnormal{pub}}, r_{\textnormal{sec}}}}[
\mathcal{A}(G, \mathcal{B}_{\mathcal{A}, G}(1^n, r_{\textnormal{pub}}), r_{\textnormal{pub}}, r_{\textnormal{sec}})\text{ is valid}
] \geq 1 - \epsilon(n)
\end{equation*}
In this case, it holds that $\mathsf{ROUND}^{\mathsf{U}}_{\epsilon(n)}(\Pi) \leq r(n)$.
\end{definition}

\subsection{Computational Complexity}
\label{subsec:complexity}

\subsubsection*{Search Problems.}

\begin{definition}[$\NP$ Search Problems]
\label{def:fnp}
$\mathcal{S} \subseteq \{ 0,1 \}^* \times \{ 0,1 \}^*$ a binary relation is in $\FNP$ if:
\begin{itemize}
\item \textit{Polynomially bounded.} There is a polynomial $p(\cdot)$ such that for every $(x, w) \in \mathcal{S}$, $|w| \leq p(|x|)$.
\item \textit{Efficiently Verifiable.} There is a deterministic polynomial-time Turing machine $\mathcal{A}$ satisfying $\mathcal{A}(x,w) = \mathbbm{1}\{ (x,w) \in \mathcal{S} \}$ for every $(x,w) \in \{ 0,1 \}^* \times \{ 0,1 \}^*$.
\end{itemize}
\end{definition}
We use the notation $\mathcal{S}_x$ to refer to the set $\{ w | (x,w) \in \mathcal{S} \}$ (which can be empty).

In $\FNP/\poly$ (the non-uniform variant of the class $\FNP$), the deterministic polynomial-time Turing machine that verifies the solutions is given a polynomial length advice. Alternatively, it is replaced with a family of Boolean circuits of polynomial size. \newline

If any instance of a problem has a solution, we say that the search problem is \textit{total}. Megiddo and Papadimitriou \cite{mp89} defined a subclass of $\FNP$ that includes only total search problems.

\begin{definition}[Total $\NP$ Search Problems]
\label{def:tfnp}
$\mathcal{S} \subseteq \{ 0,1 \}^* \times \{ 0,1 \}^*$ a binary relation is in $\TFNP$ if:
\begin{itemize}
\item $\mathcal{S} \in \FNP$ (as in Def.~\ref{def:fnp}).
\item $\mathcal{S}$ is \textit{total}, i.e. for every $x \in \{ 0,1 \}^*$ there is a $w \in \{ 0,1 \}^*$ such that $(x,w) \in \mathcal{S}$.
\end{itemize}
\end{definition}
As before, we can define a non-uniform variant of $\TFNP$, where the solutions are required to be verifiable by a family of polynomial-size Boolean circuits. It is denoted by $\TFNP / \poly$.

\subsubsection*{Average-Case Complexity.}

\begin{definition}[Ensemble of Distributions]
\label{def:ensemble}
An ensemble of distributions $\mathcal{D} = \{ \mathcal{D}_{\lambda} \}_{\lambda \in \mathbb{N}}$ is an infinite set, where $\mathcal{D}_{\lambda}$ is a probability distribution over $\{ 0,1 \}^{\lambda}$. 
\end{definition}

When sampling from a distribution $\mathcal{D}_{\lambda}$, we denote it using an arrow $x \leftarrow \mathcal{D}_{\lambda}$. When it is the uniform distribution over a domain $S$, we use $x \stackrel{\$}{\leftarrow} S$.

\begin{definition}[Samplable Ensemble] 
\label{def:samplable_ensemble}
Ensemble $\mathcal{D}$ is \textit{(polynomial-time) samplable} if there is a polynomial-time Turing machine $D$ that on uniformly random coins $r \in \{ 0,1 \}^{\lambda}$, the random variable $D(r)$ is identically distributed as $\mathcal{D}_{\lambda}$. That is, $\Pr_{r \stackrel{\$}{\leftarrow} \{ 0,1 \}^{\lambda}}[D(r) = x] = \Pr_{x \leftarrow \mathcal{D}_{\lambda}}[x]$.
\end{definition}

Take note that we assume that the sampler $D$ is \textit{length-preserving}. Meaning, on $r \in \{ 0,1 \}^{\lambda}$ it returns an instance $x \in \textnormal{Supp}(\mathcal{D}_{\lambda}) \subseteq \{ 0,1 \}^{\lambda}$.\footnote{Otherwise, the output of $D$ can be padded in the following manner. Let $p(\cdot)$ be a polynomial such that $|D(r)| \leq p(|r|)$ for all $r$ (it exists because $D$ runs in polynomial-time). Then, define $D'(r \circ \tilde{r}) = D(r) \circ 10^{p(|r|) - |D(r)|}$ where $|r \circ \tilde{r}| = p(|r|) + 1$.}

\begin{definition}[Distributional Search Problem]
\label{def:dist_search}
A \textit{distributional search problem} is a pair $(\mathcal{S}, \mathcal{D})$, where $\mathcal{S} \subseteq \{ 0,1 \}^* \times \{ 0,1 \}^*$ and $\mathcal{D}$ is a probability ensemble.
\end{definition}

\begin{definition}[Distributional Problem with a Samplable Distribution]
\label{def:samplable_distributional_problem}
$(\mathcal{S}, \mathcal{D})$ is a \textit{distributional problem in class $\mathcal{C}$ with a samplable distribution} if:
\begin{itemize}
\item The problem $\mathcal{S}$ is in $\mathcal{C}$.
\item The probability ensemble $\mathcal{D}$ is samplable.
\end{itemize}
\end{definition}

In this work, we usually take $\mathcal{C}$ to be a class of search problems, e.g., $\FNP$, $\TFNP$, or their non-uniform variants.
Also, since we focus on (polynomial-time) samplable distributions, whenever we say \textit{distributional problems} we refer to \textit{distributional problems with samplable distributions}.\footnote{Note that this is inconsistent with some of the traditional definitions in the average-case complexity literature (e.g., Levin \cite{lev86}), where distributional problems refer to \textit{computable} distributions, rather than to \textit{samplable} ones. We believe that the latter is a better fit for our setting.} \newline

Intuitively, by saying that a search problem is \textit{hard-on-average}, we mean that any solver with limited resources succeeds only with a small probability in the following 2-player game (described, for example, in a note by Gurevich \cite{gur89}): A \textsf{PPT} challenger samples an instance $x$ of the problem. The solver wins either if he finds a solution for $x$, or if there was no solution for $x$ to begin with. This idea is formalized below.

\begin{definition}[Hard-on-Average Distributional Search Problem]
\label{def:hard_on_average}
Let $(\mathcal{S}, \mathcal{D})$ be a distributional search problem.
$(\mathcal{S}, \mathcal{D})$ is \textit{$(T(\lambda), \mu(\lambda))$-hard-on-average} if for every probabilistic algorithm $\mathcal{A}$ that runs for $T(\lambda)$ steps it holds that for all large enough $\lambda$:
\begin{equation}
\label{eq:hardness}
\Pr_{\substack{\text{$\mathcal{A}$'s coins}\\x \leftarrow \mathcal{D}_{\lambda}}} 
\left[ \mathcal{S}_x\neq \emptyset \Rightarrow (x, \mathcal{A}(x)) \in \mathcal{S} \right] < \mu(\lambda)
\end{equation}
\end{definition}

Specifically, we focus our attention on two flavors of hardness:

\begin{itemize}
\item \textit{Polynomial hardness.} 
A distributional problem $(\mathcal{S}, \mathcal{D})$ is \textit{polynomially hard-on-average} if for every polynomial $T(\lambda)$ there is a negligible function $\mu(\lambda)$ (smaller than $1 / p(\lambda)$ for every polynomial $p(\lambda)$) such that it is $(T, \mu)$-hard-on-average.

\item \textit{Subexponential hardness.} 
A distributional problem $(\mathcal{S}, \mathcal{D})$ is \textit{subexponentially hard-on-average} if there are constants $\kappa, \rho > 0$ such that the problem is $(T, \mu)$-hard-on-average for $T(\lambda) = 2^{\lambda^{\kappa}}$ and $\mu(\lambda) = 2^{-\lambda^{\rho}}$.

\end{itemize}

\noindent \textit{Infinitely-often.} A distributional search problem is said to be \textit{infinitely-often} hard-on-average if Eq.~(\ref{eq:hardness}) holds for infinitely many values of $\lambda \in \mathbb{N}$ (rather than for all sufficiently large values). \newline

\noindent \textit{Non-uniformity.} If the inequality Eq.~(\ref{eq:hardness}) holds even when the algorithm $\mathcal{A}$ is given access to a polynomial-length advice, the problem is said to be hard-on-average \textit{against non-uniform solvers}.
Since all the results in this work consider non-uniform solvers, we henceforth avoid stating this explicitly.
\newline

\noindent \textit{Public coins.} If the inequality Eq.~(\ref{eq:hardness}) holds even when the algorithm $\mathcal{A}$ is given access to the random coins $r$ used for sampling $x = D(r)$, we say that the distributional problem is \textit{public-coin} hard-on-average. \newline

For the proof of Theorem~\ref{thm:probability_gap}, it would be convenient to define two additional properties (that are closely related to the definitions above). Both of them are relevant when the problems in the class $\mathcal{C}$ are not total, e.g., $\FNP$, $\FNP/\poly$. The motivation for the first one is to rule out distributions that give too much weight to instances with no solution. The second can be viewed as a relaxed version of the definition of a (polynomially) hard distributional search problem.

\begin{definition}[Non-Triviality]
\label{def:non_trivial}
Let $(\mathcal{S}, \mathcal{D})$ be a distributional search problem. $(\mathcal{S}, \mathcal{D})$ is \textit{non-trivial} if the probability of sampling an instance with a solution is non-negligible. Meaning, there is a polynomial $p(\lambda)$ such that for all $\lambda \in \mathbb{N}$:
\[
\Pr_{x \leftarrow \mathcal{D}_{\lambda}} [\mathcal{S}_x \neq \emptyset] > \frac{1}{p(\lambda)}
\]
\end{definition}

\begin{definition}[Potentially-Vacuous (Polynomial) Hardness-on-Average]
\label{def:vacuous_hardness}
Let $(\mathcal{S}, \mathcal{D})$ be a distributional search problem. $(\mathcal{S}, \mathcal{D})$ is \textit{potentially-vacuously (polynomially) hard-on-average} if there is a negligible function $\mathsf{negl}(\lambda)$ such that for all \textsf{PPT} $\mathcal{A}$ it holds that for all large enough $\lambda$:
\begin{equation}
\label{eq:vacuous_hardness}
\Pr_{\substack{\text{$\mathcal{A}$'s coins}\\x \leftarrow \mathcal{D}_{\lambda}}} 
\left[ (x, \mathcal{A}(x)) \in \mathcal{S} \right] <  \mathsf{negl}(\lambda)
\end{equation}
\end{definition}

Since we will use the notion of \textit{potentially-vacuous (polynomial) hardness-on-average} only in the context of polynomial hardness, whenever we use the term \textit{potentially-vacuous hardness-on-average} we refer to polynomial hardness.

The last definition is weaker compared to standard (polynomial) hardness-on-average. The standard definition implies that the probability of sampling an instance without a solution is at most negligible, whereas the definition of \textit{potentially-vacuous hardness-on-average} imposes no such restriction. In particular, it contains vacuously hard problems, where many of the instances have no solutions.
To illustrate, in the 2-player game described earlier, the solver no longer receives credit whenever the challenger samples an instance that cannot be solved. Naturally, in the case of a \textit{total} problem, this coincides with the standard definition.

The terms \textit{``infinitely-often''}, \textit{``against non-uniform solvers''} and \textit{``public-coin''} (introduced with Def.~\ref{def:hard_on_average} in mind) can be analogously defined in the context of potentially vacuous hardness-on-average.
\newline

Intentionally, \textit{non-triviality} is defined as a rather weak requirement. Lemma~\ref{lemma:transformation_total} shows that it can be strengthened to \textit{totality}.

\begin{restatable}{lemma}{transformation}
\label{lemma:transformation_total}
Assume the existence of a non-trivial, public-coin infinitely-often potentially-vacuously (polynomially) hard-on-average distributional problem in $\FNP/\poly$.
Then, there exists a public-coin infinitely-often (polynomially) hard-on-average distributional problem in $\TFNP/\poly$.
\end{restatable}

Hub\'{a}cek, Naor and Yogev~\cite{hny17} proved a transformation from public-coin polynomially hard-on-average distributional decision problems in $\NP$ to public-coin polynomially hard-on-average distributional problems in $\TFNP/\poly$. 
We observe that the same proof works even if we begin with a \textit{non-trivial}, \textit{public-coin infinitely-often potentially-vacuously hard-on-average} search problem in $\FNP/ \poly$, and even if the hardness that we consider is with respect to non-uniform solvers.
For completeness, we present a full proof in Appendix~\ref{subsec:proof_transformation_total}. In short, the transformation takes polynomially many ``shifts'' of the original problem and considers their disjunction. It can be shown that for certain choices of ``shifts'' the problem always has a solution, thereby making the problem \textit{total}.

\section{The \texorpdfstring{\textsf{LCL}}{LCL} Problem \texorpdfstring{$\Pi$}{Pi}}
\label{sec:lcl_problem}

In Theorem~\ref{thm:rounds_gap}, computational assumptions are shown to imply faster \textsf{LOCAL} algorithms for solving a specific \textsf{LCL} problem $\Pi$. 
Before proving the theorem (in Section~\ref{sec:assumptions_imply_gap}), this section provides a detailed description of the problem $\Pi$.

\paragraph{Roadmap.}
Section~\ref{subsec:low_blocks} reviews several simple graph structures that can be locally checked (given suitable input labelings). 
Section~\ref{subsec:high_blocks} combines them into more complex structures. 
Many of these structures have been previously defined in earlier works (e.g., \cite{bcm21, bgk25}). 
Section~\ref{subsec:problem_def} defines a graph family $\mathcal{G}$ of ``hard instances'', which have a crucial role in the definition of $\Pi$. We then formally define $\Pi$ and discuss its properties.
Fig.~\ref{fig:flow_chart} illustrates how the various components interact with each other.

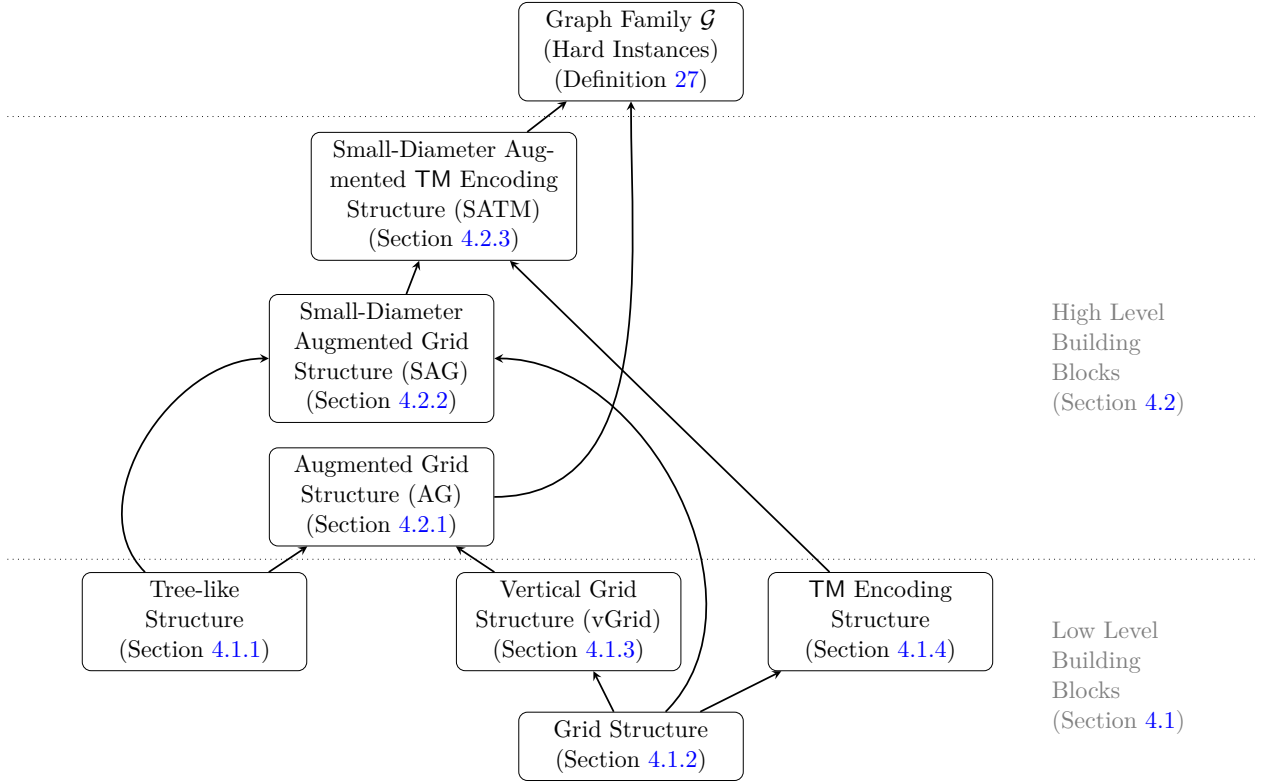
\begin{figure}[t]
\centering
\tikzstyle{rounded_rect} = [rectangle, rounded corners, minimum width=3.6cm, minimum height=1cm, text centered, draw=black, text width=3cm]
\begin{adjustbox}{max width=\textwidth}
\begin{tikzpicture}[node distance=2cm]
% Nodes
\node (tree) [rounded_rect] {Tree-like Structure \\ (Section \ref{subsubsec:tree_like})};
\node (grid) [rounded_rect, xshift=7cm, yshift=-2cm] {Grid Structure (Section \ref{subsubsec:grid})};
\node (vgrid) [rounded_rect, above of=grid, xshift=-1cm] {Vertical Grid Structure (vGrid) (Section \ref{subsubsec:vgrid})};
\node (ag) [rounded_rect, above of=tree, xshift=3cm] {Augmented Grid Structure (AG) (Section \ref{subsubsec:augmented_grid})};
\node (sag) [rounded_rect, above of=ag, yshift=0.22cm] {Small-Diameter Augmented Grid Structure (SAG) (Section \ref{subsubsec:sag})};
\node (tm) [rounded_rect, above of=grid, xshift=4cm] {\textsf{TM} Encoding Structure (Section \ref{subsubsec:tm_encode_structure})};
\node (satm) [rounded_rect, above of=sag, xshift=1cm, yshift=0.6cm, text width=4cm] {Small-Diameter Augmented \textsf{TM} Encoding Structure (SATM) \\ (Section \ref{subsec:satm})};
\node (hard) [rounded_rect, above of=satm, xshift=3cm, yshift=0.32cm] {Graph Family $\mathcal{G}$ (Hard Instances) (Definition \ref{def:graph_family_g})};
\node[text width=2.5cm, color=gray] at (15,-0.9) {Low Level Building Blocks \\ (Section \ref{subsec:low_blocks})};
\node[text width=2.5cm, color=gray] at (15,4.2) {High Level Building Blocks \\ (Section \ref{subsec:high_blocks})};
% Edges
\draw [thick,->,>=stealth] (tree) to (ag);
\draw [thick,->,>=stealth] (vgrid) to (ag);
\draw [thick,->,>=stealth] (grid) to (vgrid);
\draw [thick,->,>=stealth] (tree) to [out=135, in=180] (sag);
\draw [thick,->,>=stealth] (grid) to [in=0] (sag);
\draw [thick,->,>=stealth] (grid) to (tm);
\draw [thick,->,>=stealth] (sag) to (satm);
\draw [thick,->,>=stealth] (tm) to (satm);
\draw [thick,->,>=stealth] (satm) to (hard);
\draw [thick,->,>=stealth] (ag) to [out=0, in=270] (hard);
\draw [dotted] (-3,1) to (17,1);
\draw [dotted] (-3,8.1) to (17,8.1);
\end{tikzpicture}
\end{adjustbox}
\caption{A flow chart demonstrating how the graph structures presented in Section~\ref{subsec:low_blocks} (``low level'') relate to the structures presented in Section~\ref{subsec:high_blocks} (``high level''). The latter are then used for defining a family of hard instances $\mathcal{G}$ (Def.~\ref{def:graph_family_g}).}
\label{fig:flow_chart}
\end{figure}

\subsection{Low Level Building Blocks}
\label{subsec:low_blocks}

\subsubsection{Tree-Like Structures}
\label{subsubsec:tree_like}
\textit{Tree-like structures} are perfect binary trees, with paths connecting all the nodes located in the same layer. 

\begin{definition}[\text{\cite[Def. 6.3]{bcm21}}]
\label{def:tree_like}
A graph $G = (V,E)$ is a \textit{tree-like structure} if there's an assignment of coordinates $(d_v, i_v)$ to each $v \in V$ such that $d_v$ stands for the depth in the tree and $i_v$ for the position in the $d_v$-th layer. Given two nodes $(d_v, i_v)$ and $(d_u, i_u)$, $d_v \leq d_u$, $i_v \leq i_u$, there's an edge between them iff $(d_v, i_v) = (d_u - 1, \lfloor \frac{d_u}{2} \rfloor)$ or $(d_v,i_v) = (d_u, i_v - 1)$.
\end{definition}

\paragraph{Local Verification.}
Balliu, Censor-Hillel, Maus, Olivetti, and Suomela \cite{bcm21} defined a finite set of labels for half-edges $\mathcal{E}^{\textnormal{tree}}$ and a corresponding set of constraints $\mathcal{C}^{\textnormal{tree}}$, such that a graph equipped with such labeling satisfies $\mathcal{C}^{\textnormal{tree}}$ if and only if the graph is a \textit{tree-like structure}. See Fig.~\ref{fig:tree} for an illustration of the way the labels $\mathcal{E}^{\textnormal{tree}}$ are used for locally verifying a tree-like structure.

\begin{itemize}

\item Half-edge labels: $\mathcal{E}^{\textnormal{tree}} = \{ \mathsf{L}, \mathsf{R}, \P, \mathsf{Ch}_{\mathsf{L}}, \mathsf{Ch}_{\mathsf{R}} \}$ (standing for ``left'', ``right'', ``parent'', ``left child'' and ``right child'').
\item Constraints $\mathcal{C}^{\textnormal{tree}}$. See \cite[Sec. 6.2]{bcm21}.

\end{itemize}

\begin{figure}[t]
\centering
\begin{tikzpicture}[scale=1.6, every node/.style={circle,inner sep=2pt,fill=orange,font=\scriptsize}]
% Tree-like structure:
\node[label={30:$\mathsf{P}$}, label={0:${\mathsf{R}}$}, label={-30:$\mathsf{Ch}_{\mathsf{R}}$}, label={-120:$\mathsf{Ch}_{\mathsf{L}}$}] (top1) at (0.5-1.2,1.8) {};
\node[label={150:$\mathsf{P}$}, label={0:${\mathsf{R}}$}, label={180:${\mathsf{L}}$}, label={-30:$\mathsf{Ch}_{\mathsf{R}}$}, label={-120:$\mathsf{Ch}_{\mathsf{L}}$}] (top2) at (2.5-1.2,1.8) {};
\node[label={30:$\mathsf{P}$}, label={0:${\mathsf{R}}$}, label={-30:$\mathsf{Ch}_{\mathsf{R}}$}, label={-120:$\mathsf{Ch}_{\mathsf{L}}$}, label={-120:$\mathsf{Ch}_{\mathsf{L}}$}, label={180:${\mathsf{L}}$}] (top3) at (4.5-1.2,1.8) {};
\node[label={150:$\mathsf{P}$}, label={180:${\mathsf{L}}$}, label={-30:$\mathsf{Ch}_{\mathsf{R}}$}, label={-120:$\mathsf{Ch}_{\mathsf{L}}$}] (top4) at (6.5-1.2,1.8) {};
\node[label={20:$\mathsf{P}$}, label={0:${\mathsf{R}}$}, label={-10:$\mathsf{Ch}_{\mathsf{R}}$}, label={-170:$\mathsf{Ch}_{\mathsf{L}}$}] (ttop1) at (0.3,2.4) {};
\node[label={160:$\mathsf{P}$}, label={180:${\mathsf{L}}$}, label={-10:$\mathsf{Ch}_{\mathsf{R}}$}, label={-170:$\mathsf{Ch}_{\mathsf{L}}$}] (ttop2) at (4.3,2.4) {};
\node[label={-3:$\mathsf{Ch}_{\mathsf{R}}$}, label={-177:$\mathsf{Ch}_{\mathsf{L}}$}] (tttop) at (2.3,3) {};
\node[label={45:$\mathsf{P}$}, label={0:${\mathsf{R}}$}] (tt1) at (-1.2,1.2) {};
\node[label={135:$\mathsf{P}$}, label={0:${\mathsf{R}}$}, label={180:${\mathsf{L}}$}] (tt2) at (1-1.2,1.2) {};
\node[label={45:$\mathsf{P}$}, label={0:${\mathsf{R}}$}, label={180:${\mathsf{L}}$}] (tt3) at (2-1.2,1.2) {};
\node[label={135:$\mathsf{P}$}, label={0:${\mathsf{R}}$}, label={180:${\mathsf{L}}$}] (tt4) at (3-1.2,1.2) {};
\node[label={45:$\mathsf{P}$}, label={0:${\mathsf{R}}$}, label={180:${\mathsf{L}}$}] (tt5) at (4-1.2,1.2) {};
\node[label={135:$\mathsf{P}$}, label={0:${\mathsf{R}}$}, label={180:${\mathsf{L}}$}] (tt6) at (5-1.2,1.2) {};
\node[label={45:$\mathsf{P}$}, label={0:${\mathsf{R}}$}, label={180:${\mathsf{L}}$}] (tt7) at (6-1.2,1.2) {};
\node[label={135:$\mathsf{P}$}, label={180:${\mathsf{L}}$}] (tt8) at (7-1.2,1.2) {};

\draw[color=black] (tt1) -- (tt2);
\draw[color=black] (tt2) -- (tt3);
\draw[color=black] (tt3) -- (tt4);
\draw[color=black] (tt4) -- (tt5);
\draw[color=black] (tt5) -- (tt6);
\draw[color=black] (tt6) -- (tt7);
\draw[color=black] (tt7) -- (tt8);
\draw[color=black] (top1) -- (tt1);
\draw[color=black] (top1) -- (tt2);
\draw[color=black] (top2) -- (tt3);
\draw[color=black] (top2) -- (tt4);
\draw[color=black] (top3) -- (tt5);
\draw[color=black] (top3) -- (tt6);
\draw[color=black] (top4) -- (tt7);
\draw[color=black] (top4) -- (tt8);
\draw[color=black] (top1) -- (top2);
\draw[color=black] (top2) -- (top3);
\draw[color=black] (top3) -- (top4);
\draw[color=black] (ttop1) -- (top1);
\draw[color=black] (ttop1) -- (top2);
\draw[color=black] (ttop2) -- (top3);
\draw[color=black] (ttop2) -- (top4);
\draw[color=black] (ttop1) -- (ttop2);
\draw[color=black] (tttop) -- (ttop1);
\draw[color=black] (tttop) -- (ttop2);

\end{tikzpicture}
\caption{An example of a \textit{tree-like structure} of depth $4$ with a valid labeling.}
\label{fig:tree}
\end{figure}
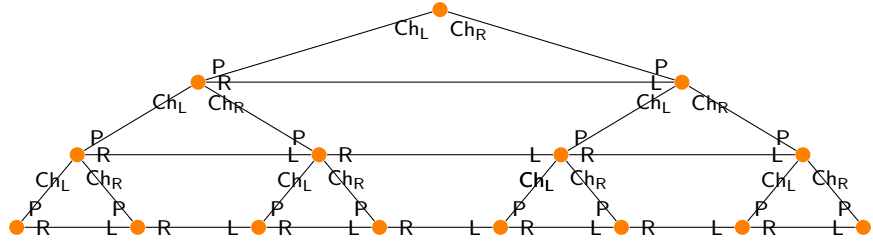

\begin{lemma}[\text{\cite[Lemma 6.6]{bcm21}}]
\label{lemma:tree_labels_sufficient}
If a graph equipped with $\mathcal{E}^{\textnormal{tree}}$ labels satisfies $\mathcal{C}^{\textnormal{tree}}$, then it is a tree-like structure. 
\end{lemma} 

\begin{lemma}[\text{\cite[Sec. 6.2]{bcm21}}]
\label{lemma:tree_labels_exist}
If a graph is a tree-like structure, then there is a labeling $\mathcal{E}^{\textnormal{tree}}$ such that $\mathcal{C}^{\textnormal{tree}}$ holds.
\end{lemma}

\paragraph{Proving Violations of \boldmath$\mathcal{C}^{\textnormal{tree}}$.} 
\cite{bgk25} defined an \textsf{LCL} problem $\Pi^{\textnormal{badTree}}$, with the purpose of not only verifying whether the input labeling satisfies the constraints mentioned above, but also to provide a ``proof'', accessible by all nodes, of local violations of the constraints (if there are any). In more detail, the input labeling includes labels from $\mathcal{E}^{\textnormal{tree}}$. If all the constraints in $\mathcal{C}^{\textnormal{tree}}$ are satisfied, then the only valid output labeling is $\bot$ (in all nodes). If $\mathcal{C}^{\textnormal{tree}}$ doesn't hold, there is a valid output labeling other than having $\bot$ in every node. This output labeling provides a ``proof'' for the graph not being tree-like. This is done in the form of pointers that lead to the violation. In a valid output labeling of $\Pi^{\textnormal{badTree}}$ for this scenario, every node has at least one neighbor with a non-$\bot$ label. This neighbor's label is a pointer to the next step in the chain leading to some violation of the tree constraints. As a result, given the output labeling of $\Pi^{\textnormal{badTree}}$, it is enough for a node to consider its local neighborhood in order to infer that the graph it belongs to is not tree-like. A desirable property of $\Pi^{\textnormal{badTree}}$ is having a round-efficient \textsf{LOCAL} algorithm for solving it.

\begin{definition}[\text{\cite[Sec. 4.2]{bgk25}}]
\label{def:bad_tree}    
$\Pi^{\textnormal{badTree}}$ is an \textsf{LCL} problem $(\mathcal{V}^{\textnormal{badTree}}_{\textnormal{input}}, \mathcal{E}^{\textnormal{badTree}}_{\textnormal{input}}, \mathcal{V}^{\textnormal{badTree}}_{\textnormal{output}}, \mathcal{C}^{\textnormal{badTree}})$.

\begin{itemize}

\item $\mathcal{V}^{\textnormal{badTree}}_{\textnormal{input}} = \{ \mathsf{marked}, \mathsf{unmarked} \}$ 

\item $\mathcal{E}^{\textnormal{badTree}}_{\textnormal{input}} = \mathcal{E}^{\textnormal{tree}}$

\item $\mathcal{V}^{\textnormal{badTree}}_{\textnormal{output}} = \{ \textsf{Err}, \bot \} \cup \{ (\mathsf{pointer}, p) | p \in \{ \mathsf{L}, \mathsf{R}, \P, \mathsf{Ch}_{\mathsf{R}} \} \}$

\item Constraints $\mathcal{C}^{\textnormal{badTree}}$: see \cite[Sec. 4.2]{bgk25}.
\end{itemize}
\end{definition}

\begin{lemma}[\text{\cite[Lemma 4.4]{bgk25}}]
\label{lemma:bad_tree_valid}
If a graph equipped with $(\mathcal{V}^{\textnormal{badTree}}_{\textnormal{input}}, \mathcal{E}^{\textnormal{badTree}}_{\textnormal{input}})$ labels that satisfy $\mathcal{C}^{\textnormal{badTree}}$ and all of the nodes were given input label \textsf{unmarked}, then the only valid solution for $\Pi^{\textnormal{badTree}}$ is setting all the nodes' output labels to $\bot$.
\end{lemma}

\begin{lemma}[\text{\cite[Lemma 4.5]{bgk25}}]
\label{lemma:bad_tree_invalid}
If a graph equipped with $(\mathcal{V}^{\textnormal{badTree}}_{\textnormal{input}}, \mathcal{E}^{\textnormal{badTree}}_{\textnormal{input}})$ labels either has a \textsf{marked} node or it doesn't satisfy $\mathcal{C}^{\textnormal{badTree}}$, then there's a valid $\mathcal{V}^{\textnormal{badTree}}_{\textnormal{output}}$ output labeling with consistent pointer chains\footnote{Consistency means that for a node labeled with $(\mathsf{pointer}, p)$, where $v := f(u,p) \neq \bot$ is the next node on the chain, if $v$'s label is $(\mathsf{pointer}, p')$ (and not $\mathsf{Err}$) one of the following must hold: (1)~$p \in \{ \mathsf{L}, \mathsf{R} \} \Rightarrow p' = p$ (2)~$p = \mathsf{P} \Rightarrow p' \in \{ \mathsf{P} ,\mathsf{L}, \mathsf{R}  \}$ (3)~$p = \mathsf{Ch}_{\mathsf{R}} \Rightarrow p' \in \{ \mathsf{Ch}_{\mathsf{R}}, \mathsf{L}, \mathsf{R}  \}$.}
leading to the marked nodes.\footnote{By ``marked'', we mean nodes that either violate $\mathcal{C}^{\textnormal{tree}}$, or have an input label \textsf{marked}. These are the nodes that in a valid $\mathcal{V}^{\textnormal{badTree}}_{\textnormal{output}}$ output labeling can output \textsf{Err}.}
The output labeling can be computed by an $O(\log n)$ round \textsf{LOCAL} algorithm.
\end{lemma}

\subsubsection{Grid Structures}
\label{subsubsec:grid}

\textit{Grid structures} of size $h \times w$ are two-dimensional grids, consisting of $h$ rows of length $w$. The endpoints of a column aren't connected to each other, and the same applies to the rows. To be put differently, the structure doesn't wrap around itself. The indices of the rows grow when going downwards, and the indices of the columns grow as we get closer to the left end.

\begin{definition}[\text{\cite[Def. 6.2]{bcm21}}]
\label{def:grid}
A graph $G = (V,E)$ is a \textit{grid structure} (of dimensions $h \times w$) if there's an assignment of coordinates $(i_v, j_v)$ to each $v \in V$ such that $i_v \in [h]$ stands for its vertical position, and $j_v \in [w]$ for its horizontal position. Two nodes $(i_v, j_v)$ and $(i_u, j_u)$ ($i_v \leq i_u$, $j_v \leq j_u$) have an edge iff $(i_v, j_v) = (i_u-1, j_u)$ or $(i_v, j_v) = (i_u, j_u - 1)$.
\end{definition}

\paragraph{Local Verification.}
$\mathcal{E}^{\textnormal{grid}}$ is a finite set of labels for the half-edges, and $\mathcal{C}^{\textnormal{grid}}$ is a corresponding set of constraints. Whenever the graph of interest is a valid grid structure, there's an assignment of labels such that the constraints are all satisfied. The opposite direction is not true, and we'll need additional conditions to hold as well. 
See Fig.~\ref{fig:grid} for an example of a valid $\mathcal{E}^{\textnormal{grid}}$ labeling of a grid structure.

\begin{itemize}

\item Half-edge labels: $\mathcal{E}^{\textnormal{grid}} = \{\mathsf{U}, \mathsf{D},  \mathsf{L}, \mathsf{R} \}$ (standing for ``up'', ``down'', ``left'' and ``right'').
\item Constraints $\mathcal{C}^{\textnormal{grid}}$. See \cite[Sec. 6.2]{bcm21}.

\end{itemize}

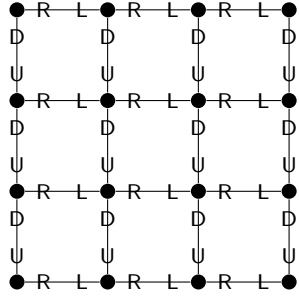
\begin{figure}[t]
\centering
\begin{tikzpicture}[scale=1.2, every node/.style={circle,inner sep=2pt,fill=black,font=\scriptsize}]
% 4x4 Grid:
% Column 1
\node[label={0:${\mathsf{R}}$}, label={-90:${\mathsf{D}}$}] (lg11) at (-5,0) {};
\node[label={90:${\mathsf{U}}$}, label={0:${\mathsf{R}}$}, label={-90:${\mathsf{D}}$}] (lg21) at (-5,-1) {};
\node[label={90:${\mathsf{U}}$}, label={0:${\mathsf{R}}$}, label={-90:${\mathsf{D}}$}] (lg31) at (-5,-2) {};
\node[label={90:${\mathsf{U}}$}, label={0:${\mathsf{R}}$}] (lg41) at (-5,-3) {};
\draw[color=black] (lg11) -- (lg21);
\draw[color=black] (lg21) -- (lg31);
\draw[color=black] (lg31) -- (lg41);
% Column 2
\node[label={180:${\mathsf{L}}$}, label={0:${\mathsf{R}}$}, label={-90:${\mathsf{D}}$}] (lg12) at (-5+1,0) {};
\node[label={90:${\mathsf{U}}$}, label={0:${\mathsf{R}}$}, label={-90:${\mathsf{D}}$}, label={180:${\mathsf{L}}$}] (lg22) at (-5+1,-1) {};
\node[label={90:${\mathsf{U}}$}, label={0:${\mathsf{R}}$}, label={-90:${\mathsf{D}}$}, label={180:${\mathsf{L}}$}] (lg32) at (-5+1,-2) {};
\node[label={90:${\mathsf{U}}$}, label={0:${\mathsf{R}}$}, label={180:${\mathsf{L}}$}] (lg42) at (-5+1,-3) {};
\draw[color=black] (lg12) -- (lg22);
\draw[color=black] (lg22) -- (lg32);
\draw[color=black] (lg32) -- (lg42);
\draw[color=black] (lg11) -- (lg12);
\draw[color=black] (lg21) -- (lg22);
\draw[color=black] (lg31) -- (lg32);
\draw[color=black] (lg41) -- (lg42);
% Column 3
\node[label={180:${\mathsf{L}}$}, label={0:${\mathsf{R}}$}, label={-90:${\mathsf{D}}$}] (lg13) at (-5+2,0) {};
\node[label={90:${\mathsf{U}}$}, label={0:${\mathsf{R}}$}, label={-90:${\mathsf{D}}$}, label={180:${\mathsf{L}}$}] (lg23) at (-5+2,-1) {};
\node[label={90:${\mathsf{U}}$}, label={0:${\mathsf{R}}$}, label={-90:${\mathsf{D}}$}, label={180:${\mathsf{L}}$}] (lg33) at (-5+2,-2) {};
\node[label={90:${\mathsf{U}}$}, label={0:${\mathsf{R}}$}, label={180:${\mathsf{L}}$}] (lg43) at (-5+2,-3) {};
\draw[color=black] (lg13) -- (lg23);
\draw[color=black] (lg23) -- (lg33);
\draw[color=black] (lg33) -- (lg43);
\draw[color=black] (lg12) -- (lg13);
\draw[color=black] (lg22) -- (lg23);
\draw[color=black] (lg32) -- (lg33);
\draw[color=black] (lg42) -- (lg43);
% Column 4
\node[label={180:${\mathsf{L}}$}, label={-90:${\mathsf{D}}$}] (lg14) at (-5+3,0) {};
\node[label={90:${\mathsf{U}}$}, label={-90:${\mathsf{D}}$}, label={180:${\mathsf{L}}$}] (lg24) at (-5+3,-1) {};
\node[label={90:${\mathsf{U}}$}, label={-90:${\mathsf{D}}$}, label={180:${\mathsf{L}}$}] (lg34) at (-5+3,-2) {};
\node[label={90:${\mathsf{U}}$}, label={180:${\mathsf{L}}$}] (lg44) at (-5+3,-3) {};
\draw[color=black] (lg14) -- (lg24);
\draw[color=black] (lg24) -- (lg34);
\draw[color=black] (lg34) -- (lg44);
\draw[color=black] (lg13) -- (lg14);
\draw[color=black] (lg23) -- (lg24);
\draw[color=black] (lg33) -- (lg34);
\draw[color=black] (lg43) -- (lg44);
\end{tikzpicture}
\caption{An example of a \textit{grid structure} of dimensions $4 \times 4$ with a valid labeling.}
\label{fig:grid}
\end{figure}

\begin{lemma}[\text{\cite[Lemma 6.5]{bcm21}}]
\label{lemma:grid_labels_sufficient}
Let $G$ be a $\mathcal{E}^{\textnormal{grid}}$ labeled graph that satisfies:
\begin{enumerate}[nosep]
\item Its labels satisfy $\mathcal{C}^{\textnormal{grid}}$.
\item \label{item:def-no_d_or_no_u} There is at least a single node that either has no incident $\mathsf{D} $ half-edges or $\mathsf{U} $ half-edges.
\item \label{item:def-no_l_or_no_r} There is at least a single node that either has no incident $\mathsf{L}$ half-edges or $\mathsf{R} $ half-edges.
\end{enumerate}
Then, the graph $G$ is a grid structure. 
\end{lemma}

\begin{lemma}[\text{\cite[Sec. 6.2]{bcm21}}]
\label{lemma:grid_labels_exist}
For every grid structure there is a $\mathcal{E}^{\textnormal{grid}}$ labeling such that the constraints $\mathcal{C}^{\textnormal{grid}}$ are satisfied.
\end{lemma}

\subsubsection{Vertical Grid Structures}
\label{subsubsec:vgrid}

\begin{definition}[\text{\cite[Def. 5.4]{bgk25}}]
A \textit{vertical grid structure} is a \textit{grid structure} of dimensions $h \times w$ that satisfies $h \geq w$.
\end{definition}

\paragraph{Local Verification.}
Vertical grid structures are associated with the following labels and constraints.
\begin{itemize}
\item Node labels: $\mathcal{V}^{\textnormal{vGrid}} = \{ 0, 1 \}$
\item Half-edge labels: $\mathcal{E}^{\textnormal{grid}}$ (no additional half-edge labels other than the grid structure's).
\item Constraints $\mathcal{C}^{\textnormal{vGrid}}$. Includes $\mathcal{C}^{\textnormal{grid}}$, with few additional constraints. See \cite[Sec. 5.2]{bgk25}.
\end{itemize}

\begin{lemma}[\text{\cite[Lemma 5.5]{bgk25}}]
If a graph is equipped with $(\mathcal{V}^{\textnormal{vGrid}}, \mathcal{E}^{\textnormal{grid}})$ labels that satisfy $\mathcal{C}^{\textnormal{vGrid}}$ along with conditions~\ref{item:def-no_d_or_no_u} and~\ref{item:def-no_l_or_no_r} of Lemma~\ref{lemma:grid_labels_sufficient}, and contains at least one node labeled $1$, then it is a vertical grid structure.
\end{lemma}

\begin{lemma}[\text{\cite[Lemma 5.6]{bgk25}}]
For every vertical grid structure there's an assignment of $(\mathcal{V}^{\textnormal{vGrid}}, \mathcal{E}^{\textnormal{grid}})$ labels such that $\mathcal{C}^{\textnormal{vGrid}}$ is satisfied, and at least one node is labeled $1$.
\end{lemma}

\subsubsection{\texorpdfstring{\textsf{TM}}{TM} Encoding Structure}
\label{subsubsec:tm_encode_structure}

Several works in the past (e.g., \cite{ns95, bbos20, cha23}) defined \textsf{LCL} problems that represent executions of Turing machines. We follow a similar path, but a key difference compared to these studies lies in the fact that they let the nodes perform the computation themselves, whereas our implementation considers the verification of the computation. 
Balliu, Brandt, Chang, Olivetti, Rabie, and Suomela \cite{bbco19} also defined an \textsf{LCL} problem centered on verification of computations, but used a path as the underlying graph structure, which inherently limited the length of the Turing machine's tape. 
In contrast, our use of grids provides much greater flexibility. 
Another difference compared to previous works is that we take into account executions over arbitrary input strings, while prior work focused on computations that begin with a blank tape.

\begin{remark}[Uniformity in Theorem~\ref{thm:rounds_gap}]
To verify solutions of the hard $\TFNP$ problem of interest, we rely on a Turing machine, a uniform model of computation.
This is where the uniformity in Theorem~\ref{thm:rounds_gap} originates from. 
We don't know if a non-uniform approach can work here. 
\end{remark}

\begin{definition}[Turing Machine]
\label{def:tm}
A single-tape, deterministic \textit{Turing machine} (\textsf{TM}) is a tuple $(Q, \Sigma, F, q_0, \delta)$.

\begin{itemize}

\item $Q$ is a finite set of states.

\item $\Sigma$ is a finite alphabet of all the symbols that may appear on the tape. Among others, contains $B = \{ \#_{\mathsf{L}}, \#_{\mathsf{R}} \}$ special blank symbols. One $\#_{\mathsf{L}}$ symbol appears at the left end of the tape. At least one $\#_{\mathsf{R}}$ (but potentially more) is written at the right end of the tape. In this work we are mostly concerned with $\Sigma = B \cup \{ 0 ,1  \}$.

\item $F \subseteq Q$ set of final states.

\item $q_0 \in Q$ the initial state.

\item $\delta: (Q \backslash F) \times \Sigma \rightarrow Q \times \Sigma \times \text{D}$ transition function. $\text{D} : = \{ \mathsf{L}, \mathsf{R}, \mathsf{S} \}$  (standing for ``left'', ``right'' and ``stay'').
Given a non-final state and the symbol in the head position, $\delta$ returns the next state, the symbol to be written before changing position, and the direction. 

\end{itemize}
\end{definition}

\paragraph{Notation.}
Let $N$ be the size of the input, and $\mathcal{M} = (Q, \Sigma, F, q_0, \delta)$ be the polynomial-time \textit{Turing machine} (\textsf{TM}) of interest.
Let $\ttt{input} \in ( \Sigma \backslash \{ \#_{\mathsf{L}}, \#_{\mathsf{R}}\})^N$ be an input string.
$S^{\mathcal{M}}(\ttt{input})$ bounds the length of tape. That is, the space complexity of $\mathcal{M}$ on $\ttt{input}$.
The runtime of $\mathcal{M}$ is exactly $T^{\mathcal{M}}(\ttt{input})$.
At time $j \in [T^{\mathcal{M}}(\ttt{input})]$, the tape stores the string $\textnormal{TAPE}^{\mathcal{M}}_{\ttt{input}}(j)$, the state is $\textnormal{STATE}^{\mathcal{M}}_{\ttt{input}}(j) \in Q$ and the head is located at position $\textnormal{HEAD}^{\mathcal{M}}_{\ttt{input}}(j)$.

\paragraph{Conventions.}
Suppose that the input and the output are aligned to the left end of the tape. Fill the rest of the cells with blank symbols $\#_{\mathsf{R}}$. 
To illustrate, at both the beginning and the end of the execution the tape takes on the form $\#_{\mathsf{L}} \circ \ttt{str} \circ \#_{\mathsf{R}}^+$, where $\#_{\mathsf{R}}^+$ stands for a string of unspecified (non-zero) length of $\#_{\mathsf{R}}$'s, and $\ttt{str} \in (\Sigma \backslash \{ \#_{\mathsf{L}}, \#_{\mathsf{R}} \})^*$ is a string of non-blank symbols, also of unspecified length (which may be empty).
Another useful property that we are going to assume is that the machine stops at a final state only when the head is located at the leftmost cell of the tape (the one with the value $\#_{\mathsf{L}}$). 
At the end of the computation the head of the \textsf{TM} is repeatedly moved to the left until it reaches the beginning of the tape, and only then we enter a state in $F$. 

\paragraph{Graph Structure.}
The \textit{\textsf{TM} encoding structure} is, first of all a grid structure, as defined in Section~\ref{subsubsec:grid}. Additionally, there is an assignment of labels to the nodes such that they hold information about the computation executed by the Turing machine. Every node is given the value of a cell in the tape at some point in time, and in every column there is a node whose input label indicates that the head is stationed there. This node's label also contains the state of the machine at that moment of the execution. 
We refer to the horizontal axis of the grid as the time dimension, and to the vertical axis as the tape. 
Due to our conventions, the input and the output (which are aligned to the left end of the tape) are written in the upper part of the vertical axis of the grid.
Observe that this labeling gives rise to the verification of the computation while communicating only with nodes that are $O(1)$-hops away.

\begin{definition}[\textsf{TM} Encoding Structure]
\label{def:tm_encoding}
Let $\mathcal{M} = (Q, \Sigma, F, q_0, \delta)$ (that obeys the conventions discussed above), $N \in \mathbb{N}$ and $\ttt{input} \in (\Sigma \backslash \{ \#_{\mathsf{L}}, \#_{\mathsf{R}} \})^N$.
A graph $G$ is \textit{\textsf{TM} encoding structure with respect to $\mathcal{M}$ and $\ttt{input}$} if $G$ is a grid structure of dimensions $h \times w$ for $h \geq S^{\mathcal{M}}(\ttt{input})$ and $w \geq T^{\mathcal{M}}(\ttt{input})$.
\end{definition}

A detail worth paying attention to is the fact that the width $w$ (resp. height $h$) is allowed to be larger than the time (resp. space) complexity of the machine $\mathcal{M}$. 
The labels for the extra columns must match those written in the $T^{\mathcal{M}}(\ttt{input})$-th column, meaning there are no additional computation steps; we just copy the output of the Turing machine as is until we reach the leftmost column.
The nodes at the excess bottom rows are labeled with $\#_{\mathsf{R}}$. 
In both cases, it's straightforward to see that the requirements can be implemented as locally checkable constraints.

At first glance, Def.~\ref{def:tm_encoding} seems redundant as it merely refers to a grid structure with restrictions over its dimensions. The intention is to make it clear that when referring to such grids, there exists an assignment of labels that corresponds to the computation of $\mathcal{M}$ on the input string specified by the labels of the nodes located in the right column.

\paragraph{Local Verification.}
Let $\mathcal{M} = (Q, \Sigma, F, q_0, \delta)$.
We use the following labels and constraints (inspired by \cite{bbos20}). The labels and the constraints depend on the Turing machine in interest $\mathcal{M}$. They are denoted by $\mathcal{V}^{\mathcal{M}}$, $\mathcal{C}^{\mathcal{M}}$, respectively.
Together with the node labels $\mathcal{V}^{\mathcal{M}}$, the half-edges are equipped with the grid structure's half-edge labels $\mathcal{E}^{\text{grid}}$. Then, $\mathcal{C}^{\mathcal{M}}$ are meant to be checked along with $\mathcal{C}^{\text{grid}}$. 

It's important to note \textit{\textsf{TM} encoding structures} with respect to $\mathcal{M}$ and $\ttt{input}$ are not locally checkable. Every such structure has a valid labeling (i.e., respects $\mathcal{C}^{\mathcal{M}}$ and $\mathcal{C}^{\text{grid}}$). The opposite is not necessarily true; graphs equipped with valid labelings may not be \textit{\textsf{TM} encoding structures} with respect to $\mathcal{M}$ and $\ttt{input}$. 
One challenge arises from machines that enter infinite loops. As a simple example, consider a \textsf{TM} $\mathcal{M}$ that doesn't edit the contents of the tape and never moves the head. Suppose that after $\omega(1)$ steps it returns to the initial state $q_0$. For such $\mathcal{M}$'s, the history of the tape can be written over a cylinder graph. The local constraints verifying the computation are satisfied, and the nodes cannot infer that they belong to a cylinder (rather than a grid). We tackle this issue by introducing additional conditions, similar to what we did with the \textit{grid structures}.

The other difficulty has to do with our inability to locally verify that the input string mentioned in a given labeling is truly $\ttt{input}$. Because of that, given a correctly labeled graph (plus the additional conditions) we say that it is a \textit{\textsf{TM} encoding structure} with respect to $\mathcal{M}$ and \textit{some} string $\ttt{input}$, where $\ttt{input}$ is directly determined by the labels of the right column. 

See Fig.~\ref{fig:tm_encoding} for an illustration of the labeling for a $6 \times 10$ grid and a simple Turing machine $\mathcal{M}$ (that obeys our conventions) which flips the bits of any given input string. E.g., on $\#_{\mathsf{L}} 001 \#_{\mathsf{R}}^+$, the output is $\#_{\mathsf{L}} 110 \#_{\mathsf{R}}^+$.

\begin{figure}[t]
\begin{minipage}{\textwidth}
\begin{minipage}[h]{0.74\textwidth}
\centering
\resizebox{\textwidth}{!}{
\begin{tikzpicture}
\matrix [matrix of nodes, nodes={draw, minimum width=4.65em, minimum height=4.45em, inner sep=0pt}]
{
% First Row
|[fill=blue!20]|$(\#_{\mathsf{L}},\mathsf{H},q_3)$ & |[fill=blue!20]|$(\#_{\mathsf{L}},\mathsf{H},q_2)$ & $(\#_\mathsf{L}, \bot, \bot)$ & $(\#_\mathsf{L}, \bot, \bot)$ & $(\#_\mathsf{L}, \bot, \bot)$ & $(\#_\mathsf{L}, \bot, \bot)$ & $(\#_\mathsf{L}, \bot, \bot)$ & $(\#_\mathsf{L}, \bot, \bot)$ & $(\#_\mathsf{L}, \bot, \bot)$ & |[fill=blue!20]|$(\#_\mathsf{L}, \mathsf{H}, q_0)$ \\
% 2nd Row
$(1,\bot,\bot)$ & $(1,\bot,\bot)$ & |[fill=blue!20]|$(1,\mathsf{H},q_2)$ & $(1,\bot,\bot)$ & $(1,\bot,\bot)$ & $(1,\bot,\bot)$ & $(1,\bot,\bot)$ & |[fill=red!20]|$(1,\bot,\bot)$ & |[fill=blue!20]|$(0,\mathsf{H},q_1)$ & $(0,\bot,\bot)$ \\
% 3nd Row
$(1,\bot,\bot)$ & $(1,\bot,\bot)$ & $(1,\bot,\bot)$ & |[fill=blue!20]|$(1,\mathsf{H},q_2)$ & $(1,\bot,\bot)$ & $(1,\bot,\bot)$ & |[fill=red!20]|$(1,\bot,\bot)$ & |[fill=blue!20]|$(0,\mathsf{H},q_1)$ & $(0,\bot,\bot)$ & $(0,\bot,\bot)$ \\
% 4th Row
$(0,\bot,\bot)$ & $(0,\bot,\bot)$ & $(0,\bot,\bot)$ & $(0,\bot,\bot)$ & |[fill=blue!20]|$(0,\mathsf{H},q_2)$ & |[fill=red!20]|$(0,\bot,\bot)$ & |[fill=blue!20]|$(1,\mathsf{H},q_1)$ & $(1,\bot,\bot)$ & $(1,\bot,\bot)$ & $(1,\bot,\bot)$ \\
% 5th Row
$(\#_\mathsf{R}, \bot, \bot)$ & $(\#_\mathsf{R}, \bot, \bot)$ & $(\#_\mathsf{R}, \bot, \bot)$ & $(\#_\mathsf{R}, \bot, \bot)$ & $(\#_\mathsf{R}, \bot, \bot)$ & |[fill=blue!20]|$(\#_\mathsf{R}, \mathsf{H}, q_1)$ & $(\#_\mathsf{R}, \bot, \bot)$ & $(\#_\mathsf{R}, \bot, \bot)$ & $(\#_\mathsf{R}, \bot, \bot)$ & $(\#_\mathsf{R}, \bot, \bot)$ \\
% 6th Row
$(\#_\mathsf{R}, \bot, \bot)$ & $(\#_\mathsf{R}, \bot, \bot)$ & $(\#_\mathsf{R}, \bot, \bot)$ & $(\#_\mathsf{R}, \bot, \bot)$ & $(\#_\mathsf{R}, \bot, \bot)$ & $(\#_\mathsf{R}, \bot, \bot)$ & $(\#_\mathsf{R}, \bot, \bot)$ & $(\#_\mathsf{R}, \bot, \bot)$ & $(\#_\mathsf{R}, \bot, \bot)$ & $(\#_\mathsf{R}, \bot, \bot)$ \\
};

\end{tikzpicture}
}
\captionof{figure}{$\mathcal{V}^{\mathcal{M}}$ labeling associated with a \textsf{TM} $\mathcal{M}$ that flips the bits of the input. 
The tape is initially set to $\#_{\mathsf{L}} 001 \#_{\mathsf{R}}^+$. 
$q_3 \in F$. 
\textsf{H}~is an abbreviation for \textsf{Head}.}
\label{fig:tm_encoding}
\end{minipage}
\begin{minipage}[h]{0.25\textwidth}
\centering
{\fontsize{9}{12}\selectfont
\begin{tabular}{|c c | c c c|} 
\hline
$Q$ & $\Sigma$ & $Q$ & $\Sigma$ & D \\ 
\hline
$q_0$ & $\#_{\mathsf{L}}$ & $q_1$ & $\#_{\mathsf{L}}$ & $\mathsf{R}$ \\
$q_1$ & $0$ & $q_1$ & $1$ & $\mathsf{R}$ \\
$q_1$ & $1$ & $q_1$ & $0$ & $\mathsf{R}$ \\
$q_1$ & $\#_{\mathsf{R}}$ & $q_2$ & $\#_{\mathsf{R}}$ & $\mathsf{L}$ \\
$q_2$ & $0$ & $q_2$ & $0$ & $\mathsf{L}$ \\
$q_2$ & $1$ & $q_2$ & $1$ & $\mathsf{L}$ \\
$q_2$ & $\#_{\mathsf{L}}$ & $q_3$ & $\#_{\mathsf{L}}$ & $\mathsf{S}$ \\[0.5ex] 
\hline
\end{tabular}
}
\captionof{table}{Partial transition function of the \textsf{TM} $\mathcal{M}$ depicted in Fig.~\ref{fig:tm_encoding}.}
\end{minipage}
\end{minipage}
\end{figure}

\begin{itemize}

\item Node labels: $\mathcal{V}^{\mathcal{M}} := 
\left\{ (t, h, s) \big| t \in \Sigma, h \in \{ \mathsf{Head}, \bot \}, s \in Q \cup \{ \bot \} \right\}$

$t$ is a symbol on the tape. $h$ indicates whether this is the current head position. $s$ is a state, if this is the current head position and $\bot$ otherwise. Recall, the set of states $Q$ and the alphabet $\Sigma$ are finite, so they can be part of a constant-size label.

\item Half-edge labels: $\mathcal{E}^{\textnormal{grid}}$

\item Constraints $\mathcal{C}^{\mathcal{M}}$: let $v \in V$, labeled with $(t,h,s)$.

\begin{enumerate}[label=\textbf{L\arabic*}, nosep]

\item \label{item:con-first} If $v$ doesn't have a $\mathsf{R} $ half-edge:

\begin{itemize}[nosep]

\item If $f(v, \mathsf{U}) = \bot$, then $(t,h,s) = (\#_{\mathsf{L}}, \mathsf{Head}, q_0)$. Otherwise, $h = \bot \land s = \bot$

\end{itemize}

\item \label{item:con-non_final}
If $s \in Q \backslash F$, then $u = f(v, \mathsf{L} )$ exists and labeled $(t^{\mathsf{L}}, h^{\mathsf{L}}, s^{\mathsf{L}})$.

\begin{enumerate}[nosep]

\item \label{item:con-transition} 
If $h = \mathsf{Head}$, then $\delta(s, t) = (\tilde{s}, t^{\mathsf{L}}, \texttt{dir})$ for $\texttt{dir} \in \text{D}$. 
Additionally, \textbf{exactly} one of the following holds.

\begin{itemize}[nosep]

\item $h^{\mathsf{L}} = \mathsf{Head}$. In this case: $\texttt{dir} = \mathsf{S}$, $s^{\mathsf{L}} = \tilde{s}$.

\item $f(u, \mathsf{L}, \mathsf{U} )$ exists and labeled $(t^{\mathsf{LU} }, \mathsf{Head}, s^{\mathsf{LU} })$. In this case: $\texttt{dir} = \mathsf{\mathsf{L}}, s^{\mathsf{LU}} = \tilde{s}$.

\item $f(v, \mathsf{L}, \mathsf{D} )$ exists and labeled $(t^{\mathsf{LD} }, \mathsf{Head}, s^{\mathsf{LD} })$. In this case: $\texttt{dir} = \mathsf{\mathsf{R}}, s^{\mathsf{LD}} = \tilde{s}$.

\end{itemize}

\item \label{item:con-nothing} 
If $h \neq \mathsf{Head}$, all of the following must hold.

\begin{itemize}[nosep]
\item \label{item:con-unchanged} $t = t^{\mathsf{L}}$.
\item \label{item:con-neighboring_head} If $h^{\mathsf{L}} = \mathsf{Head}$, then either $f(u, \mathsf{U} )$ or $f(u, \mathsf{D} )$ exists and labeled with $\mathsf{Head}$.
\end{itemize}

\end{enumerate}

\item \label{item:con-state} $h = \mathsf{Head}$ iff $s \neq \bot$.

\item \label{item:con-accepting_state} 
If $s \in F$:
\begin{itemize}[nosep]
\item $f(v, \mathsf{U}) = \bot$
\item $(t,h,s) = (\#_{\mathsf{L}}, \mathsf{Head}, s)$
\item If $u = f(v, \mathsf{L})$ for $u \neq \bot$, then $u$'s label is $(t,h,s)$.
\end{itemize}

\item \label{item:con-borders} $t = \#_{\mathsf{L}}$ iff $v$ doesn't have a $\mathsf{U} $ half-edge. If $t = \#_{\mathsf{R}}$, either $v$ doesn't have a $\mathsf{D} $ half-edge, or $f(v, \mathsf{D})$ is labeled with $(t^{\mathsf{D}}, h^{\mathsf{D}}, s^{\mathsf{D}})$ for $t^{\mathsf{D}} = \#_{\mathsf{R}}$.
\end{enumerate}

\end{itemize}
  
\begin{lemma}
\label{lemma:tm_labels_sufficient}
Let $\mathcal{M} = (Q, \Sigma, F, q_0, \delta)$.
A graph $G$ and a $(\mathcal{V}^{\mathcal{M}}, \mathcal{E}^{\textnormal{grid}})$ labeling satisfy $\mathcal{C}^{\mathcal{M}}$, $\mathcal{C}^{\textnormal{grid}}$ and conditions~\ref{item:def-no_d_or_no_u}, \ref{item:def-no_l_or_no_r} of 
Lemma~\ref{lemma:grid_labels_sufficient}.

Then, $G$ is a \textit{\textsf{TM} encoding structure} with respect to $\mathcal{M}$ and $\ttt{input}$, for \textit{some} string $\ttt{input} \in (\Sigma \backslash \{ \#_{\mathsf{L}}, \#_{\mathsf{R}} \})^*$.
\end{lemma}

\begin{proof}
$\mathcal{C}^{\text{grid}}$ is known to hold, and so are conditions~\ref{item:def-no_d_or_no_u}, \ref{item:def-no_l_or_no_r}. From Lemma~\ref{lemma:grid_labels_sufficient}, $G$ is ought to be a \textit{grid structure}. 
The string obtained from concatenating the labels in the right column takes on the form $\#_{\mathsf{L}} \circ \ttt{input} \circ \#_{\mathsf{R}}^+$, or else constraint~\ref{item:con-borders} is violated.
What remains is to confirm that the dimensions align with Def.~\ref{def:tm_encoding} when the input string is $\ttt{input}$. 
The main step in the proof is to justify that the labeling accurately follows the transition function $\delta$ of $\mathcal{M}$. 

% first column
Constraint~\ref{item:con-first} ensures that the computation begins with the initial state $q_0$, and with the head located at the left end of the tape (i.e., the topmost row of the grid).

% transition steps.
From \ref{item:con-transition}, local computation steps necessarily reflect the rules dictated by $\delta$. Moreover, constraint~\ref{item:con-nothing} verifies that whenever the head of the machine is not stationed in some cell, this cell's contents remain unchanged. 
Together, \ref{item:con-transition} and \ref{item:con-nothing} guarantee that any local computation step of $\mathcal{M}$ is implemented correctly.
Unfortunately, it still doesn't mean that the computation as a whole is right. Several heads of the machine may exist and operate simultaneously.

% a single head
From constraint~\ref{item:con-first}, there is a single $\mathsf{Head}$ in the first column. Constraint~\ref{item:con-transition} justifies why a $\mathsf{Head}$ node may lead to exactly one $\mathsf{Head}$ in the following column, whereas constraint~\ref{item:con-neighboring_head} explains why the label $\mathsf{Head}$ cannot show up unless there was a $\mathsf{Head}$ labeled node in the previous column.
Combined, these constraints suggest that each column has exactly one \textsf{Head} labeled node. This allows us to conclude that the labels constitute a correct simulation of $\mathcal{M}$ on $\ttt{input}$.

% final state.
The simulation necessarily reaches a final state. Otherwise, when reaching the left column of the grid $\delta$ would indicate that an additional step is required, which would violate \ref{item:con-non_final}.

Now that we are convinced that the simulation adheres to the transition function $\delta$, the argument for the dimensions being appropriate is straightforward. If the width is less than $T^{\mathcal{M}}(\ttt{input})$ or the height is less than $S^{\mathcal{M}}(\ttt{input})$, it would either result in an error in one of the steps recorded by the node labels (i.e., labels that contradict $\delta$), or it would indicate that at some point the rule implied by constraint~\ref{item:con-non_final} refers to a non-existing node and therefore fails.
\end{proof}

\begin{remark}[Constraints \ref{item:con-state}-~\ref{item:con-borders}]
Assuming that the input string follows the conventions, constraints \ref{item:con-first}-~\ref{item:con-non_final} suffice for proving Lemma~\ref{lemma:tm_labels_sufficient}. Constraint~\ref{item:con-state} helps with the formatting of the labels, but doesn't have an actual effect on the correctness of the simulation of $\mathcal{M}$ (since \ref{item:con-non_final} anyway ignores the state $s$ whenever $h = \bot$). \ref{item:con-accepting_state} has two primary roles. First, it verifies that $\mathcal{M}$ follows our convention of allowing the machine to enter a final state only when the head is positioned at the left end of the tape. Second, it ensures (together with \ref{item:con-nothing}) that if there are any additional columns after $\mathcal{M}$ has reached a final state, they must be identical (this gives us freedom in setting the width of the grid). Lastly, constraint~\ref{item:con-borders} ensures that the ends of the tape are consistent with our conventions. Additionally, it allows us to increase the height arbitrarily.
\end{remark}

\begin{lemma}
\label{lemma:tm_labels_exist}
Let $\mathcal{M} = (Q, \Sigma, F, q_0, \delta)$, $\ttt{input} \in (\Sigma \backslash \{ \#_{\mathsf{L}}, \#_{\mathsf{R}} \})^*$ and graph $G$ be a \textit{\textsf{TM} encoding structure with respect to $\mathcal{M}$ and $\ttt{input}$}.

Then, there is a $(\mathcal{V}^{\mathcal{M}}, \mathcal{E}^{\textnormal{grid}})$ labeling such that $\mathcal{C}^{\mathcal{M}}$ and $\mathcal{C}^{\textnormal{grid}}$ hold, where the input labels of the right column of the grid correspond to $\ttt{input}$.
\end{lemma}

\begin{proof}[Proof (sketch)]
Since $G$ is a grid, there is a valid $\mathcal{E}^{\textnormal{grid}}$ labeling due to Lemma~\ref{lemma:grid_labels_exist}. 
As for the $\mathcal{V}^{\mathcal{M}}$ labels, one can simply execute $\mathcal{M}$ on the given $\ttt{input}$ and write down the history of the tape as the node labels. This should be done with attention to the way the $\mathcal{V}^{\mathcal{M}}$ labels are formatted (e.g., write them as tuples $(t,h,s)$). Clearly, constraints~\ref{item:con-first}-~\ref{item:con-borders} hold.
\end{proof}

\subsection{High Level Building Blocks}
\label{subsec:high_blocks}

\subsubsection{Augmented Grid Structures}
\label{subsubsec:augmented_grid}

To recall, the problem of \cite{bgk25} is defined so that for graphs outside a certain family, there is a valid output labeling computable in a short time. This is to say, the ``real'' problem (i.e., having a grid where the outputs of the rows are consistent and equal to the inputs of the rightmost nodes) is designed especially for the graphs of that graph family. By \textit{augmented grid structures} (or AG), we refer to that family. Essentially, these are vertical grid structures, with tree-like structures glued to their columns. 

\begin{definition}[\text{\cite[Def. 6.1]{bgk25}}]
\label{def:ag}
A graph is an \textit{augmented grid structure} (AG) iff it can be obtained by:

\begin{itemize}[nosep]
\item Let $l \in \mathbb{N}$ and set $h$ to $2^l$. $w$ is some integer in $\{ 1,...,h\}$.

\item $G$ is a vertical grid structure of dimensions $h \times w$.

\item $T_0, ..., T_{w-1}$ are $w$ tree-like structures of height $l$.

\item The leaves of $T_j$ coincide with the $j$-th column of the grid $G$. That is, a node with coordinates $(l-1,i)$, $i \in [h]$ in the tree $T_j$, $j \in [w]$, is identified with the grid node whose coordinates are $(i, j)$ (i.e. $i$-th row, $j$-th column).

\end{itemize}
\end{definition}

\begin{figure}[t]
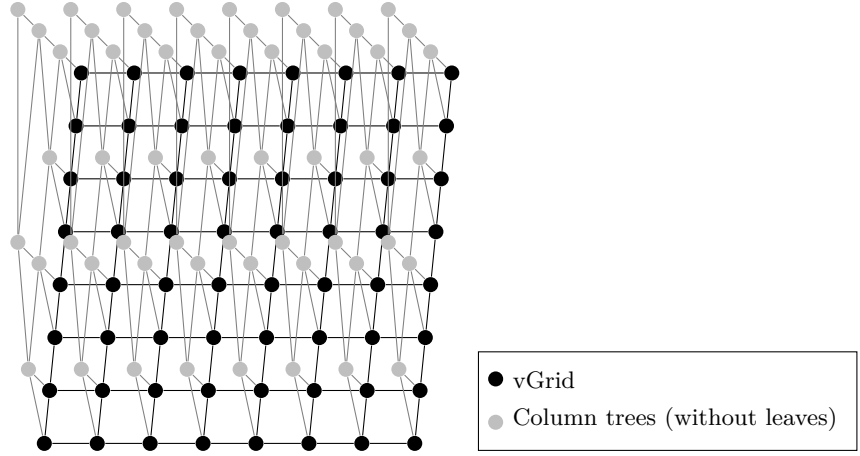

\centering
\includestandalone{figures/ag}
\caption{An example of an \textit{augmented grid} (AG) structure, where the dimensions of the (vertical) grid are $8 \times 8$.}
\label{fig:augmented_grid}
\end{figure}

\paragraph{Proving Violations.} 
Unlike previous constructions, we don't explicitly state a set of constraints for locally verifying AG structures, but in short, one should ensure that $\mathcal{C}^{\textnormal{tree}}$ and $\mathcal{C}^{\textnormal{vGrid}}$ are satisfied with respect to the column trees and the bottom grid. 

Our objective is to define an \textsf{LCL} problem $\Pi^{\textnormal{badAG}}$ where valid output labelings satisfy the following conditions. First, every connected subgraph induced by the nodes with the output label $\bot$ forms an AG equipped with a suitable input labeling. 
On the other hand, the output labeling in subgraphs induced by the nodes with non-$\bot$ outputs consists of pointer chains that lead to the violations. 

To define $\Pi^{\textnormal{badAG}}$, we make use of the low-level building blocks presented earlier. Constraints $\mathcal{C}^{\textnormal{tree}}$ and $\mathcal{C}^{\textnormal{vGrid}}$ (together with the labels associated with them) help to identify local violations of the structure. 
$\mathcal{C}^{\textnormal{badTree}}$ determines how to construct ``proofs'' of these violations, that is, pointer chains that point to the violations.

\begin{definition}[\text{\cite[Sec. 6.1]{bgk25}}]
\label{def:bad_ag}    
$\Pi^{\textnormal{badAG}}$ is an \textsf{LCL} problem $(\mathcal{V}^{\textnormal{badAG}}_{\textnormal{input}}, \mathcal{E}^{\textnormal{badAG}}_{\textnormal{input}} , \mathcal{V}^{\textnormal{badAG}}_{\textnormal{output}}, \mathcal{C}^{\textnormal{badAG}})$.
\begin{itemize}

\item $\mathcal{V}^{\textnormal{badAG}}_{\textnormal{input}} = \{ \mathsf{treeNode} \} \cup \{ ( \mathsf{treeNode}, \mathsf{gridNode} , l) | l \in \mathcal{V}^{\textnormal{vGrid}} \}$

\item $\mathcal{E}^{\textnormal{badAG}}_{\textnormal{input}} = \{ (\mathsf{treeEdge}, l) | l\in \mathcal{E}^{\textnormal{tree}} \} \cup \{ (\mathsf{gridEdge}, l) | l \in \{ \mathsf{L}, \mathsf{R}  \} \} \cup \{ ((\mathsf{gridEdge}, \mathsf{D} ), (\mathsf{treeEdge}, \mathsf{L})) \}$ \\ $\cup \{ ((\mathsf{gridEdge}, \mathsf{U} ), (\mathsf{treeEdge, \mathsf{R}})) \} $

\item $\mathcal{V}^{\textnormal{badAG}}_{\textnormal{output}} = \{ \bot, \textsf{Err}, \mathsf{TreeErr}, \mathsf{GridErr}, \mathsf{VertErr} \} \cup \{ (\mathsf{ColErr}, l) | l \in \mathcal{V}^{\textnormal{badTree}} \}$

\item Constraints $\mathcal{C}^{\textnormal{badAG}}$: see $\mathcal{C}^{\mathsf{badGraph}}$ in \cite[Sec. 6.1]{bgk25}.

\end{itemize}

\end{definition}

\begin{lemma}[\text{\cite[Lemma 6.6]{bgk25}}]
\label{lemma:bad_ag_invalid}
If a graph is equipped with $(\mathcal{V}^{\textnormal{badAG}}_{\textnormal{input}}, \mathcal{E}^{\textnormal{badAG}}_{\textnormal{input}})$ input labeling, there is an $O(\log n)$ round \textsf{LOCAL} algorithm for $\Pi^{\textnormal{badAG}}$ that returns a valid $(\mathcal{V}^{\textnormal{badAG}}_{\textnormal{output}}, \mathcal{E}^{\textnormal{badAG}}_{\textnormal{output}})$ output labeling, such that in the subgraph induced by the nodes whose output label is $\bot$, every connected component is an augmented grid structure.
\end{lemma}

\begin{lemma}[\text{\cite[Lemma 6.5]{bgk25}}]
\label{lemma:bad_ag_valid}
For an augmented grid structure, there's a $(\mathcal{V}^{\textnormal{badAG}}_{\textnormal{input}}, \mathcal{E}^{\textnormal{badAG}}_{\textnormal{input}})$ input labeling such that the only valid $(\mathcal{V}^{\textnormal{badAG}}_{\textnormal{output}}, \mathcal{E}^{\textnormal{badAG}}_{\textnormal{output}})$ output labeling is having $\bot$ in all nodes.
\end{lemma}

\subsubsection{Small-Diameter Augmented Grid Structures}
\label{subsubsec:sag}

Balliu, Censor-Hillel, Maus, Olivetti, and Suomela \cite{bcm21} designed a variant of the \textit{augmented grid structures} (AG) described in the previous section. Picture an AG where the bottom grid is not limited to being vertical, with a grid structure attached to the side of the trees and an additional tree positioned at the top (see Fig.~\ref{fig:sag}).
The motivation behind this construction is to reduce the diameter of the AG to being logarithmic. Such graphs enjoy an advantage in the ability to transmit information from one end to another. 
Every two nodes can exchange messages within $O(\log n)$ rounds.

\begin{definition}[\text{\cite[Def. 6.4]{bcm21}}]
\label{def:sag}
A graph is a \textit{small-diameter augmented grid structure} (SAG) iff it can be obtained by:

\begin{itemize}[nosep]
\item Let $l, l' \in \mathbb{N}$ and set $h$, $w$ to $2^l$, $2^{l'}$, respectively. 

\item $G_1$ is a grid structure of dimensions $h \times w$.

\item $G_2$ is a grid structure of dimensions $l \times w$.

\item $T_0, ..., T_{w-1}$ are $w$ tree-like structures of height $l$.

\item $T_{\mathsf{top}}$ is a tree-like structure of height $l'$.

\item The leaves of $T_j$ coincide with the $j$-th column of the bottom grid $G_1$: node $(l-1,i)$ in the $T_j$ is identified with the grid node whose coordinates are $(i, j)$.

\item The leftmost nodes of $T_j$ coincide with the $j$-th column of the side grid $G_2$: node $(l - 1 - i, 0)$ in $T_j$ is identified with the grid node whose coordinates are $(i, j)$.

\item The roots of $T_0, ..., T_{w-1}$ coincide with the leaves of $T_{\mathsf{top}}$: the root of $T_i$ is identified with the top tree node whose coordinates are $(l' -1 , i)$.

\end{itemize}
\end{definition}

\begin{figure}[t]
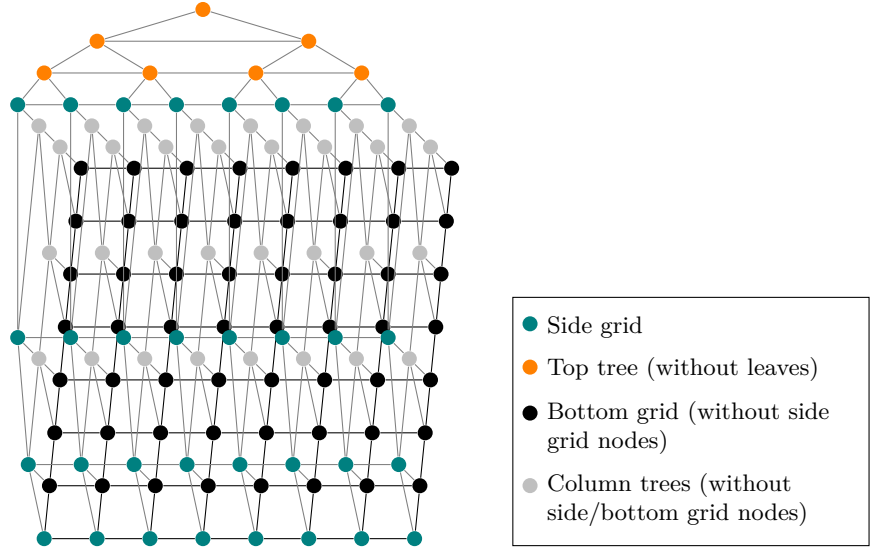

\centering
\includestandalone{figures/sag}
\caption{An example of a \textit{small-diameter augmented grid} (SAG) structure, when the dimensions of the bottom grid are $8 \times 8$.}
\label{fig:sag}
\end{figure}

\paragraph{Local Verification.}
The labels and the constraints associated with SAGs are as follows.

\begin{itemize}

\item Half-edge labels: $\mathcal{E}^{\textnormal{SAG}}$. A label is a subset of $\{ (t, l) | t \in \{ \mathsf{colTree}, \mathsf{topTree} \}, l \in \mathcal{E}^{\text{tree}}  \} \cup \{ (t, l) | t \in \{ \mathsf{bottomGrid}, \mathsf{side Grid} \} , l \in \mathcal{E}^{\text{grid}} \} $

\item Constraints $\mathcal{C}^{\textnormal{SAG}}$: see $\mathcal{C}^{\mathsf{proof}}$ in \cite[Sec. 6.2]{bcm21}.

\end{itemize}

\begin{lemma}[\text{\cite[Lemma 6.7]{bcm21}}]
\label{lemma:sag_labels_sufficient}
If a graph equipped with $\mathcal{E}^{\textnormal{SAG}}$ labels satisfies $\mathcal{C}^{\textnormal{SAG}}$, then it is a \textit{small-diameter augmented grid structure} (SAG).
\end{lemma}

\begin{lemma}[\text{\cite[Lemma 6.8]{bcm21}}]
\label{lemma:sag_labels_exist}
Any \textit{small-diameter augmented grid structure} (SAG) has a $\mathcal{E}^{\textnormal{SAG}}$ labeling that satisfies $\mathcal{C}^{\textnormal{SAG}}$.
\end{lemma}

\paragraph{Proving Violations of \boldmath$\mathcal{C}^{\textnormal{SAG}}$.}

Another attractive property of this construction is the existence of a round-efficient as well as computationally efficient algorithm that generates proofs for local errors on invalid graphs and input labelings. 
Similarly to the \textsf{LCL} problems $\Pi^{\textnormal{badTree}}$ and $\Pi^{\textnormal{badAG}}$, which served a comparable role in the context of proving violations in tree-like structures and AGs, we denote by $\Pi^{\textnormal{badSAG}}$ the \textsf{LCL} problem whose valid output labelings are either $\bot$ or proofs of violations of $\mathcal{C}^{\textnormal{SAG}}$.

\begin{definition}[\text{\cite[Sec. 6.3]{bcm21}}]
\label{def:bad_sag}    

$\Pi^{\textnormal{badSAG}}$ is an \textsf{LCL} problem $(\mathcal{V}^{\textnormal{badSAG}}_{\textnormal{input}}, \mathcal{E}^{\textnormal{badSAG}}_{\textnormal{input}} , \mathcal{V}^{\textnormal{badSAG}}_{\textnormal{output}}, \mathcal{E}^{\textnormal{badSAG}}_{\textnormal{output}}, \mathcal{C}^{\textnormal{badSAG}})$.

\begin{itemize}
\item $\mathcal{V}^{\textnormal{badSAG}}_{\textnormal{input}} = \{ \mathsf{marked}, \mathsf{unmarked} \}$

\item $\mathcal{E}^{\textnormal{badSAG}}_{\textnormal{input}} = \mathcal{E}^{\textnormal{SAG}}$

\item $\mathcal{V}^{\textnormal{badSAG}}_{\textnormal{output}} = \{ \mathsf{pointer}, \textsf{Err} , \bot\} $

\item $\mathcal{E}^{\textnormal{badSAG}}_{\textnormal{output}} = \{ \bot \} \cup \{ (\mathsf{pointer}, c, p ,t) | c \in [3], p \in \{ \mathsf{L}, \mathsf{R}, \P, \mathsf{Ch}_{\mathsf{R}} \} ,t \in \{ \mathsf{colTree}, \mathsf{topTree} \} \}$

\item Constraints $\mathcal{C}^{\textnormal{badSAG}}$. See \cite[Sec. 6.3]{bcm21}.

\end{itemize}
\end{definition}

\begin{lemma}[\text{\cite[Lemma 6.9]{bcm21}}]
\label{lemma:bad_sag_valid}
If a graph is a SAG, then there is an $(\mathcal{V}^{\textnormal{badSAG}}_{\textnormal{input}}, \mathcal{E}^{\textnormal{badSAG}}_{\textnormal{input}})$ input labeling such that the only valid $(\mathcal{V}^{\textnormal{badSAG}}_{\textnormal{output}}, \mathcal{E}^{\textnormal{badSAG}}_{\textnormal{output}})$ output labeling is having $\bot$ in every node and half-edge.
\end{lemma}

\begin{lemma}[\text{\cite[Lemmas 6.10, 6.11]{bcm21}}]
\label{lemma:bad_sag_invalid}
If a graph equipped with $(\mathcal{V}^{\textnormal{badSAG}}_{\textnormal{input}}, \mathcal{E}^{\textnormal{badSAG}}_{\textnormal{input}})$ labels either has a \textsf{marked} node, or it doesn't satisfy $\mathcal{C}^{\textnormal{SAG}}$, then there's a valid $(\mathcal{V}^{\textnormal{badSAG}}_{\textnormal{output}}, \mathcal{E}^{\textnormal{badSAG}}_{\textnormal{output}})$ output labeling with consistent pointer chains leading to the ``marked'' nodes. The output labeling can be computed by an $O(\log n)$ round \textsf{LOCAL} algorithm.
\end{lemma}

\subsubsection{Small-Diameter Augmented \texorpdfstring{\textsf{TM}}{TM} Encoding Structures}
\label{subsec:satm}

In the \textit{\textsf{TM} encoding structures} introduced in~\ref{subsubsec:tm_encode_structure}, the only special requirement from the grid was to be of large enough dimensions. This was sufficient when all we had in mind was local verification. That is, every node should be aware of errors taking place in its constant radius neighborhood. Our actual objective demands more than that. 
Similarly to the way we verify tree-like structures, AGs and SAGs, we want to define an \textsf{LCL} problem that on a valid \textsf{TM} encoding structure there will exist an input labeling for which the only acceptable output labeling is $\bot$, whereas invalid structures can output proofs for the violations they have. 
Proofs of this type imply that nodes become aware of errors far away from them. 
Unfortunately, constructing such proofs in a small number of rounds is impossible in grid structures (and therefore, also in \textsf{TM} encoding structures), as the information cannot travel so fast. 

To tackle this issue, we encode the computation of the \textsf{TM} in the labels of the bottom grid of a SAG (see Section~\ref{subsubsec:sag}). 
In addition to $C^{\textnormal{SAG}}$, we also require $\mathcal{C}^{\mathcal{M}}$ to hold (with respect to some \textsf{TM} $\mathcal{M}$ that we will choose). 
Violations of $\mathcal{C}^{\mathcal{M}}$ are handled using techniques similar to those employed for addressing local structural errors in the SAG.
Namely, the output labeling consists of pointer chains that lead to the errors, which may now be related to issues in the simulation of the \textsf{TM}. 
This will be called \textit{small-diameter augmented \textsf{TM} encoding structures} (SATM). 

Another advantage of this structure (compared to the standard \textit{\textsf{TM} encoding structures}) is that we get conditions \ref{item:def-no_d_or_no_u}, \ref{item:def-no_l_or_no_r} (of Lemma~\ref{lemma:grid_labels_sufficient}) for free.

\begin{definition}
\label{def:satm}
Let $\mathcal{M} = (Q, \Sigma, F, q_0, \delta)$ and $\ttt{input} \in (\Sigma \backslash \{ \#_{\mathsf{L}}, \#_{\mathsf{R}} \})^*$.
A graph is a \textit{small-diameter augmented \textsf{TM} encoding structures} (SATM) with respect to $\mathcal{M}$ and $\ttt{input}$ if and only if it is a \textit{small-diameter augmented grid} (SAG), and the bottom grid structure forms a \textit{\textsf{TM} encoding structure} with respect to $\mathcal{M}$ and $\ttt{input}$.
\end{definition}

\paragraph{Local Verification.}
Since SATM is a mix of two previously discussed graph structures, it makes sense to define the labels and the constraints in the following way. Note that the labels and the constraints depend on $\mathcal{M}$. For brevity, we omit $\mathcal{M}$ from the notation, as we typically focus on a single Turing machine at a time.

\begin{itemize}

\item Node labels: $\mathcal{V}^{\textnormal{SATM}} : = \{ \Pi_2 \} \cup \{ (\mathsf{tmNode}, l) | l \in \mathcal{V}^{\mathcal{M}} \}$

The $(\mathsf{tmNode}, ?)$ labels are to be assigned to the SAG's bottom grid. The rest are labeled with~$\Pi_2$.

\item Half-edge labels: $\mathcal{E}^{\textnormal{SATM}} : = \mathcal{E}^{\textnormal{SAG}}$

\item Constraints $\mathcal{C}^{\textnormal{SATM}}$:

\begin{itemize}[nosep]
\item With respect to the half-edge labels, $\mathcal{C}^{\textnormal{SAG}}$ must be satisfied.

\item In the subgraph induced by edges with $\mathcal{E}^{\text{SAG}}$ half-edge labels $(\mathsf{bottomGrid}, ?)$, constraints $\mathcal{C}^{\mathcal{M}}$ must be satisfied.

\end{itemize}
\end{itemize}

\begin{lemma}
\label{lemma:satm_labels_sufficient}
Let $\mathcal{M} = (Q, \Sigma, F, q_0, \delta)$.
If a graph equipped with $(\mathcal{V}^{\textnormal{SATM}}, \mathcal{E}^{\textnormal{SATM}})$ labels (for $\mathcal{M}$) satisfies $\mathcal{C}^{\textnormal{SATM}}$, then it is a \textit{small-diameter augmented \textsf{TM} encoding structure} (SATM) with respect to $\mathcal{M}$ and $\ttt{input}$, for \textit{some} string $\ttt{input} \in (\Sigma \backslash \{ \#_{\mathsf{L}}, \#_{\mathsf{R}} \})^*$.
\end{lemma}

\begin{proof}
Satisfying $\mathcal{C}^{\textnormal{SATM}}$ means that $\mathcal{C}^{\textnormal{SAG}}$ and $\mathcal{C}^{\mathcal{M}}$ are satisfied. By Lemma~\ref{lemma:sag_labels_sufficient}, the graph is necessarily a SAG. Moreover, by Lemma~\ref{lemma:tm_labels_sufficient}, the bottom grid is a \textit{\textsf{TM} encoding structure} with respect to our $\mathcal{M}$ and some input string appearing in the labels of the right column of the grid.
\end{proof}

\begin{lemma}
\label{lemma:satm_labels_exist}
Let $\mathcal{M} = (Q, \Sigma, F, q_0, \delta)$ and $\ttt{input} \in (\Sigma \backslash \{ \#_{\mathsf{L}}, \#_{\mathsf{R}} \})^*$.
Any \textit{small-diameter augmented \textsf{TM} encoding structure} (SATM) with respect to $\mathcal{M}$ and $\ttt{input}$ has a $(\mathcal{V}^{\textnormal{SATM}}, \mathcal{E}^{\textnormal{SATM}})$ labeling that satisfies $\mathcal{C}^{\textnormal{SATM}}$, where the labels of the right column correspond to input $\ttt{input}$.
\end{lemma}

\begin{proof}
Lemma~\ref{lemma:sag_labels_exist} provides a $\mathcal{E}^{\textnormal{SAG}}$ labeling for the half-edges such that $\mathcal{C}^{\textnormal{SAG}}$ holds. As for the node labels, Lemma~\ref{lemma:tm_labels_exist} can be applied with the purpose of generating $\mathcal{V}^{\mathcal{M}}$ node labels (which can be easily converted to $\mathcal{V}^{\textnormal{SATM}}$ labels) that stand for the execution of the given \textsf{TM} $\mathcal{M}$ on $\ttt{input}$. Together, all the constraints in $\mathcal{C}^{\textnormal{SATM}}$ are satisfied.
\end{proof}

\paragraph{Proving Violations of \boldmath$\mathcal{C}^{\textnormal{SATM}}$.}

For an appropriate graph structure, Lemma~\ref{lemma:satm_labels_exist} ensures the existence of a labeling that fully satisfies the set of local constraints $\mathcal{C}^{\text{SATM}}$. 
As in Lemma~\ref{lemma:bad_sag_valid}, this enables us to define an \textsf{LCL} problem $\Pi^{\text{badSATM}}$ in which the only valid output labeling in this case assigns $\bot$ to every node.

On the other hand, it follows from Lemma~\ref{lemma:satm_labels_sufficient} that if the graph is unsuitable, or if the labeling standing for a simulation of a \textsf{TM} is erroneous, then some constraint in $\mathcal{C}^{\text{SATM}}$ must be violated. 
In other words, errors in the \textsf{TM} simulation can be detected locally and can therefore be treated in the same way as any other local error in the graph structure.
Consequently, the argument used in the proof of Lemma~\ref{lemma:bad_sag_invalid} carries over without major modifications, yielding an efficient algorithm for generating a locally checkable proof of such errors. The locally checkable proofs also serve as valid output labelings for the problem $\Pi^{\text{badSATM}}$, which we now proceed to define.

\begin{definition}
\label{def:bad_satm}
$\Pi^{\textnormal{badSATM}}$ is an \textsf{LCL} problem $(\mathcal{V}^{\textnormal{badSATM}}_{\textnormal{input}}, \mathcal{E}^{\textnormal{badSATM}}_{\textnormal{input}} , \mathcal{V}^{\textnormal{badSATM}}_{\textnormal{output}}, \mathcal{E}^{\textnormal{badSATM}}_{\textnormal{output}}, \mathcal{C}^{\textnormal{badSATM}})$.

\begin{itemize}
\item $\mathcal{V}^{\textnormal{badSATM}}_{\textnormal{input}} = \mathcal{V}^{\textnormal{SATM}} \times \{ \mathsf{marked}, \mathsf{unmarked} \}$

\item $\mathcal{E}^{\textnormal{badSATM}}_{\textnormal{input}} = \mathcal{E}^{\textnormal{SATM}}$

\item $\mathcal{V}^{\textnormal{badSATM}}_{\textnormal{output}} = \mathcal{V}^{\textnormal{badSAG}}_{\textnormal{output}} $

\item $\mathcal{E}^{\textnormal{badSATM}}_{\textnormal{output}} = \mathcal{E}^{\textnormal{badSAG}}_{\textnormal{output}}$

\item Constraints $\mathcal{C}^{\textnormal{badSATM}}$. For $u \in V:$

\begin{enumerate}[nosep]

\item \label{item:def-local_satm_errors} \textsf{Err} is a valid output only if $u$ is \textsf{marked} or if it doesn't satisfy $\mathcal{C}^{\textnormal{SATM}}$ (rather than just $\mathcal{C}^{\textnormal{SAG}}$).

\end{enumerate}
All the constraints (except for the first one) from $\mathcal{C}^{\textnormal{badSAG}}$ in Def.~\ref{def:bad_sag} remain unchanged.

\end{itemize}

\end{definition}

\begin{lemma}
\label{lemma:bad_satm_valid}
Let $\mathcal{M} = (Q, \Sigma, F, q_0, \delta)$ and $\ttt{input} \in (\Sigma \backslash \{ \#_{\mathsf{L}}, \#_{\mathsf{R}} \})^*$.
$G$ is a SATM with respect to $\mathcal{M}$ and $\ttt{input}$.

Then, $G$ has a $(\mathcal{V}^{\textnormal{badSATM}}_{\textnormal{input}}, \mathcal{E}^{\textnormal{badSATM}}_{\textnormal{input}})$ input labeling such that the labels of the right column correspond to $\ttt{input}$, and the only valid $(\mathcal{V}^{\textnormal{badSATM}}_{\textnormal{output}}, \mathcal{E}^{\textnormal{badSATM}}_{\textnormal{output}})$ output labeling is having $\bot$ in every node and half-edge.
\end{lemma}

\begin{proof}
Lemma~\ref{lemma:satm_labels_exist} suggests that there is a $(\mathcal{V}^{\textnormal{SATM}}, \mathcal{E}^{\textnormal{SATM}})$ input labeling such that $\mathcal{C}^{\textnormal{SATM}}$ is satisfied, and the labels of the right column correspond to the string $\ttt{input}$. From constraint~\ref{item:def-local_satm_errors} of $\mathcal{C}^{\textnormal{badSATM}}$, none of the nodes output \textsf{Err}. That is, there are no marked nodes. From this point, the proof of Lemma~\ref{lemma:bad_sag_valid} can be used. It follows that the only valid output labeling is $\bot$.
\end{proof}

\begin{lemma}
\label{lemma:bad_satm_invalid}
Let $\mathcal{M} = (Q, \Sigma, F, q_0, \delta)$.
Let $G$ be a graph equipped with $(\mathcal{V}^{\textnormal{badSATM}}_{\textnormal{input}}, \mathcal{E}^{\textnormal{badSATM}}_{\textnormal{input}})$ labels (w.r.t $\mathcal{M}$) such that it either has a \textsf{marked} node, or it doesn't satisfy $\mathcal{C}^{\textnormal{SATM}}$.

Then, $G$ has a valid $(\mathcal{V}^{\textnormal{badSATM}}_{\textnormal{output}}, \mathcal{E}^{\textnormal{badSATM}}_{\textnormal{output}})$ output labeling with consistent pointer chains leading to the ``marked'' nodes. The output labeling can be computed by an $O(\log n)$ round \textsf{LOCAL} algorithm.
\end{lemma}

\begin{proof}
To satisfy constraint~\ref{item:def-local_satm_errors} in $\mathcal{C}^{\text{badSATM}}$, nodes that detect local errors must output \textsf{Err}. Thus, they count as marked nodes. 
Lemma~\ref{lemma:bad_sag_invalid}'s proof gives rise to an $O(\log n)$ round \textsf{LOCAL} algorithm that generates consistent pointer chains leading to the marked nodes.\footnote{The algorithm from Lemma~\ref{lemma:bad_sag_invalid} (originally, \cite[Lemma 6.11]{bcm21}) initially marks the nodes that fail to meet the local constraints $\mathcal{C}^{\text{SAG}}$. The difference is that we now also mark the nodes that violate the constraints $\mathcal{C}^{\mathcal{M}}$.}
\end{proof}

\subsection{Problem Definition}
\label{subsec:problem_def}

\paragraph{Hard Instances.}

The \textsf{LCL} problem $\Pi$ is defined to be hard on a graph family $\mathcal{G}$ of hard instances that we will soon define formally. Roughly speaking, these graphs consist of two components attached to each other. One is an AG and the other is a SATM (as in Fig.~\ref{fig:graph_family_g}), where the choice of the \textsf{TM} being simulated and the dimensions of the AG are aimed at making the problem hard. 

The SATM component is defined with respect to the Turing machine $\mathcal{M}$ described below. 
Remember that we assumed the existence of a public-coin $(T, \mu)$-hard-on-average problem in $\TFNP$ $(\mathcal{S}, \mathcal{D})$. 
As $\mathcal{D}$ is given to be samplable, there is an efficient sampler $D$ for $\mathcal{D}$. The sampler is assumed to be \textit{length-preserving}. I.e., $x = D(r)$ is also $|r|$ bits long. 
The length of the instance (and the random string used for sampling it) indicates how computationally hard is the distributional problem $(\mathcal{S}, \mathcal{D})$: no $T(|r|)$-time solver should be capable of finding a solution for $x = D(r)$ with probability greater than $\mu(|r|)$ when given $r$ that was sampled uniformly at random.
Additionally, there is a \textsf{TM} $V$ that on input $(x, w)$, where $|w| \leq \poly(|x|)$, verifies if it belongs to the relation $\mathcal{S}$. $V$ is deterministic and polynomial-time (in $|x|$).
Let polynomial $p(\cdot)$ be the length of the solutions to instances of $\mathcal{S}$.~\footnote{If instances of the same length have differing solution lengths, this can be resolved by padding.}
Our machine's input is of the form $(r,w,z)$, where $r$ is a string that can be used for sampling an instance, $w$ is an alleged solution for it, and $z = \Vec{0}$ is a (possibly empty) vector of zeros.
The output will be of length $|r| + |z|$. 
On input $(r, w, z)$, the \textsf{TM} $\mathcal{M}$ follows the steps in Alg.~\ref{alg:hard_tm} below.

\begin{algorithm}[H]
\caption{Description of \textsf{TM} $\mathcal{M}$}\label{alg:hard_tm}
\begin{algorithmic}[1]
\Statex \textbf{Input:} $(r, w, z)$ \Comment{$|w| = p(|r|)$}
\Statex \textbf{Output:} $|r| + |z|$ bits long string
\State Use $D$ to sample an instance $x = D(r)$. 
\State Use $V$ to verify the solution: $\ttt{ans} = V(x ,w) \in \{\text{YES}, \text{NO} \}$.
\State \textbf{return} $r \circ 1^{|z|}$ \textbf{if} $\ttt{ans} = \text{YES} \land z = 0^{|z|}$ \textbf{else} $1^{|r| + |z|}$
\end{algorithmic}
\end{algorithm}

The graphs are structured such that the rightmost column of the AG is connected to the leftmost column of the bottom grid of the SATM. The problem $\Pi$ will require the labels of the nodes in the two columns to contain the same values. Thus, in order to set the input to the rightmost column of the AG to some string $r \circ 1^{|z|}$, the adversary must be able to find a solution to the instance $x = D(r)$. For a uniformly random $r$, this $x$ is distributed according to the distribution $\mathcal{D}_{|r|}$. 
Assuming that $(\mathcal{S}, \mathcal{D})$ is public-coin $(T,\mu)$-hard-on-average (for non-uniform solvers), the task of finding a solution is infeasible for any $T(|r|)$ that is superpolynomial in $n$, even when the coins $r$ are known.

\begin{definition}[Graph Family $\mathcal{G}$]
\label{def:graph_family_g}
$(\mathcal{S}, \mathcal{D})$ is the distributional problem considered in the definition of $\Pi$. Suppose that the samplable distribution $\mathcal{D}$ has a polynomial-time sampler $D$, and $\mathcal{S}$ has a polynomial-time verification procedure $V$. Let $\mathcal{M}$ be the \textsf{TM} implementing Alg.~\ref{alg:hard_tm} for these $D$ and~$V$.
An $n$~vertex graph $G$ is in $\mathcal{G}$ iff it can be obtained by:
\begin{itemize}[nosep]

\item A SATM with respect to $\mathcal{M}$ and some string $\ttt{input} \in \{ 0,1 \}^*$. Let the bottom grid of the SATM be of dimensions $H \times W$

\item An AG with a grid of dimensions $h \times w$, where $h \leq H$.

\item The right column of the AG is connected to the leftmost column of the bottom grid of the SATM. I.e., the AG node $(j, 0)$ has an edge to the SATM node $(j, W -1)$, for every $j \in [h]$. 

\end{itemize}
\end{definition}

\begin{figure}[t]
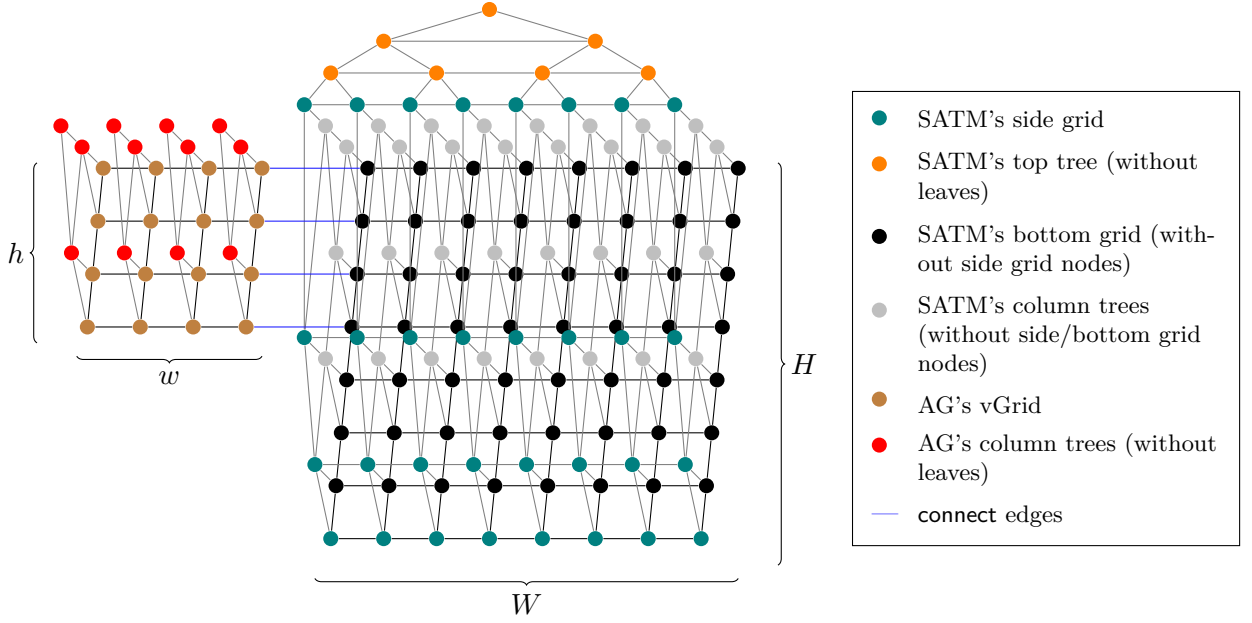

\centering
\includestandalone{figures/full}
\caption{An example of a graph in the family $\mathcal{G}$ (see Def.~\ref{def:graph_family_g}), with an AG of dimensions $h \times w = 4 \times 4$ (on the left) and a SATM with dimensions $H \times W = 8 \times 8$ (on the right).}
\label{fig:graph_family_g}
\end{figure}

\paragraph{\textsf{LCL} Problem.}

$\Pi$ is a combination of the
problems $\Pi^{\textnormal{badAG}}$ and $\Pi^{\textnormal{badSATM}}$ (introduced earlier) with a new problem $\Pi^{\text{hard}}$, which is supposed to be round consuming for the hard instances. On graphs that consist of two components, one in which the nodes output $\bot$ when solving $\Pi^{\textnormal{badAG}}$, and a second component where the nodes output $\bot$ upon solving $\Pi^{\textnormal{badSATM}}$, the problem $\Pi$ forces the nodes to solve $\Pi^{\text{hard}}$. On the other hand, for other output labelings for $\Pi^{\textnormal{badAG}}$ and $\Pi^{\textnormal{badSATM}}$ the nodes should be exempted from solving $\Pi^{\text{hard}}$. In this way, we ensure that $\Pi$ is promise-free. 

Conceptually, $\Pi^{\text{hard}}$ is non-trivial when solved over graphs in $\mathcal{G}$. 
There, the AG nodes in every row are required to output the same bit. 
Additionally, there must be at least one row where the output bit equals the input of the rightmost node. 
We don't formally define $\Pi^{\text{hard}}$, and it will be implicit in the definition of the full problem $\Pi$. 

\begin{definition}[The Problem $\Pi$]
\label{def:pi}
$\Pi$ is an \textsf{LCL} problem $(\mathcal{V}^{\Pi}_{\textnormal{input}}, \mathcal{E}^{\Pi}_{\textnormal{input}} , \mathcal{V}^{\Pi}_{\textnormal{output}}, \mathcal{E}^{\Pi}_{\textnormal{output}}, \mathcal{C}^{\Pi})$.

\noindent Assuming the existence of a public-coin $(T, \mu)$-hard-on-average problem $(\mathcal{S}, \mathcal{D})$ in $\TFNP$, where $V$ is a deterministic polynomial-time verifier and $D$ is a deterministic polynomial-time sampler, let $\mathcal{M}$ be the \textsf{TM} executing the steps in Alg.~\ref{alg:hard_tm}.

\begin{itemize}
\item $\mathcal{V}^{\Pi}_{\textnormal{input}} := 
\{ (\mathsf{agNode}, l) | l \in \{ 0,1 \} \times \mathcal{V}^{\textnormal{badAG}}_{\textnormal{input}} \} \cup \{ (\mathsf{satmNode, l}) | l \in \mathcal{V}^{\textnormal{badSATM}}_{\textnormal{input}} \}$

\item $\mathcal{E}^{\Pi}_{\textnormal{input}} := 
\{ (\mathsf{agEdge}, l) | l \in \mathcal{E}^{\textnormal{badAG}}_{\textnormal{input}} \} \cup \{ (\mathsf{satmEdge, l}) | l \in \mathcal{E}^{\textnormal{badSATM}}_{\textnormal{input}} \} \cup \{ \mathsf{connect} \}$

Half-edge input labels are either labels used in the previously defined \textsf{LCL} problems $\Pi^{\textnormal{badAG}}$ and $\Pi^{\textnormal{badSATM}}$, or a new kind of labels, denoted by $\mathsf{connect}$.

\item $\mathcal{V}^{\Pi}_{\textnormal{output}} := \mathcal{V}^{\textnormal{badAG}}_{\textnormal{output}} \cup \mathcal{V}^{\textnormal{badSATM}}_{\textnormal{output}} \cup \{  \text{YES}, \text{NO}\} \cup (\{ 0,1 \}\times\{ \text{YES}, \text{NO} \} )$

Output labels associated with $\Pi^{\textnormal{badAG}}$ or $\Pi^{\textnormal{badSATM}}$ are both valid output labels for $\Pi$. Furthermore, for the \textsf{agNode} labeled nodes, we allow $\{0,1 \} \times \{  \text{YES}, \text{NO} \}$ output labels (when they are grid nodes) and also $\{  \text{YES}, \text{NO} \}$ (if they are non-leaves tree nodes).

\item $\mathcal{E}^{\Pi}_{\textnormal{output}} = \mathcal{E}^{\textnormal{badSATM}}_{\textnormal{output}}$

\item Constraints $\mathcal{C}^{\Pi}$:\footnote{Constraint~\ref{item:con-original_constraints} is identical to \cite[Sec. 7]{bgk25}.}

\begin{enumerate}[label=\textbf{C\arabic*}, nosep]

\item \label{item:con-bad_satm} In the subgraph induced by the nodes labeled with \textsf{satmNode}, where we map $(\mathsf{satmEdge}, l)$ half-edge labels to $l$ and $(\mathsf{satmNode}, l)$ node labels to $l$: If we map all the labels not in $\mathcal{V}^{\textnormal{badSATM}}_{\textnormal{output}}$ to $\bot$, $\mathcal{C}^{\textnormal{badSATM}}$ is satisfied, where the relevant \textsf{TM} is $\mathcal{M}$.

\item \label{item:con-original_constraints} In the subgraph induced by the nodes labeled with \textsf{agNode}, where we map $(\mathsf{agEdge}, l)$ half-edge labels to $l$ and $(\mathsf{agNode}, l)$ node labels to $l$:

\begin{enumerate}[nosep]
\item \label{item:con-bad_ag} If we map all the labels not in $\mathcal{V}^{\textnormal{badAG}}_{\textnormal{output}}$ to $\bot$, $\mathcal{C}^{\textnormal{badAG}}$ is satisfied.

\item \label{item:con-grid_nodes} If a node is labeled $(\mathsf{treeNode}, \mathsf{gridNode}, l)$, and its output is not in $\mathcal{V}^{\textnormal{badAG}}_{\textnormal{output}}$, then it must be in $\{ 0,1 \} \times \{ \text{YES}, \text{NO} \}$.

\item \label{item:con-consistent_rows} If a node's output is $(b, x) \in \{ 0,1 \} \times \{ \text{YES}, \text{NO} \}$, let $v = f(u, (\mathsf{gridEdge}, \mathsf{R}))$. If $v$ exists, its output is either in $\mathcal{V}^{\textnormal{badAG}}_{\textnormal{output}}$, or $(b, x')$ (for some $x'$).

\item \label{item:con-tree_nodes} If a node is labeled with $\mathsf{treeNode}$ and the output is not in $\mathcal{V}^{\textnormal{badAG}}_{\textnormal{output}}$, then the output must be in $\{ \text{YES}, \text{NO} \}$.

\item \label{item:con-tree_propagation} If a node $u$ has an output from $\{ \text{YES}, \text{NO} \}$, let $v = f(u, (\mathsf{treeEdge}, \mathsf{Ch}_{\mathsf{L}}))$ and $z = f(u, (\mathsf{treeEdge}, \mathsf{Ch}_R))$. 
The outputs of $v$ and $z$ are in $\{ 0,1 \} \times \{ \text{YES}, \text{NO} \}$ or in $\{\text{YES}, \text{NO} \}$. $u$'s output is $\text{YES}$ iff 
$v$'s (or $z$'s) output contains $\text{YES}$.

\item \label{item:con-root} If a node has a $\{ \text{YES}, \text{NO} \}$ output but no $(\mathsf{treeEdge}, \P)$ edge, then it must be a $\text{YES}$.

\end{enumerate}

\item \label{item:con-rightmost_node} 
For an $(\mathsf{agNode}, (b_{\text{in}}, l))$ labeled node $u$ with output $(b_{\textnormal{out}}, x) \in \{ 0,1 \} \times \{\text{YES}, \text{NO}\}$ that doesn't have an $(\mathsf{agEdge}, (\mathsf{gridEdge}, \mathsf{R} ))$ edge, let $v = f(u, (\mathsf{connect}))$. 

\begin{enumerate}
\item \label{item:con-automatic_yes_output} 

If one of the following holds, $x$ in $u$'s output must be $\text{YES}$.

\begin{itemize}[nosep]
\item $v$ is not defined.

\item $v$'s input label is not of the form $(\mathsf{satmNode}, (\mathsf{tmNode}, (t,h,s)))$.

\item $u$ has no $(\mathsf{agEdge}, (\mathsf{gridEdge}, \mathsf{D}))$ but $v$'s label $(\mathsf{satmNode}, (\mathsf{tmNode}, (t,h,s)))$ has a $t$ value other than $\#_{\mathsf{R}}$.

\item $u$ has no $(\mathsf{agEdge}, (\mathsf{gridEdge}, \mathsf{U}))$ edge but $v$ has a $(\mathsf{satmEdge}, (\mathsf{gridEdge}, \mathsf{U}))$ edge.

\item $v$'s $\mathcal{V}^{\textnormal{badSATM}}_{\textnormal{output}}$ output label is not $\bot$.

\item $v$ has a half-edge with an input label $(\mathsf{satmEdge}, (\mathsf{bottomGrid}, \mathsf{L}))$.

\end{itemize}

\item \label{item:con-top_bottom}
Otherwise, if $u$ doesn't have an $(\mathsf{agEdge}, (\mathsf{gridEdge}, \mathsf{U} ))$ (or an $(\mathsf{agEdge}, (\mathsf{gridEdge}, \mathsf{D} ))$) edge: $(b_{\text{out}}, x) = (1, \text{NO})$

\item \label{item:con-input_ag_neq_input_tm} 
Otherwise, let $v$'s input label be $(\mathsf{satmNode}, (\mathsf{tmNode}, (t,h,s)))$. Then, if $t \neq b_{\text{in}}$: $x = \text{YES}$

\item \label{item:con-output_01} Otherwise, $x = \text{YES} \longleftrightarrow b_{\text{in}} = b_{\text{out}}$

\end{enumerate}

\end{enumerate}

\end{itemize}

\end{definition}

\paragraph{Digest.}
$\Pi$ is almost identical to the problem defined by \cite{bgk25}. The main difference is in the AG's rightmost column. Originally, the inputs $b_{\text{in}}$ given to the nodes in that column could be chosen freely. The newly introduced constraints \ref{item:con-bad_satm} and \ref{item:con-rightmost_node} force these inputs to correspond to the outcomes of computations of $\mathcal{M}$. We elaborate on this below.

\paragraph{Consistency between the AG Input Labels and \boldmath$\mathcal{M}$'s Output.}
Constraint~\ref{item:con-automatic_yes_output} checks if all the nodes in the AG's right column are properly connected to the last column of the SATM, which corresponds to the output tape of the \textsf{TM} simulation. That is, every node has exactly one SATM neighbor attached to it through a $\mathsf{connect}$ labeled half-edge. This neighbor is located in the left column (as it doesn't have a $\mathsf{bottomGrid}$ half-edge labeled $\mathsf{L}$) of a correctly labeled SATM (we know that the labeling is correct because the $\mathcal{V}^{\textnormal{badSATM}}_{\textnormal{output}}$ output label of the node is $\bot$). 

Additionally, we check if the node in the top-right corner of the AG is connected to the SATM node representing the beginning of the tape. This is captured by verifying that an $\mathsf{agNode}$ is missing a $\mathsf{U}$ edge if and only if its neighboring $\mathsf{satmNode}$ also lacks a $\mathsf{U}$ edge. Remember that according to $\mathcal{C}^{\mathcal{M}}$'s constraint \ref{item:con-borders}, this node is labeled with the symbol $\#_{\mathsf{L}}$.

In \ref{item:con-automatic_yes_output} we also verify that the node in the bottom-right corner of the AG is connected to a $\#_{\mathsf{R}}$ labeled SATM node.

Then, \ref{item:con-input_ag_neq_input_tm} checks if the input labels of the AG nodes in that column contain the symbols appearing in the output of the simulated \textsf{TM}.\footnote{Except for the top and the bottom rows, where we already know that the $t$ values must be $\#_{\mathsf{L}}$ and $\#_{\mathsf{R}}$, respectively. Due to constraint~\ref{item:con-top_bottom}, we ignore the inputs $b_{\text{in}}$ in these rows.}
Due to our conventions regarding the way the Turing machine operates, the outcome is formatted as $\#_{\mathsf{L}} \circ \ttt{str} \circ \#_{\mathsf{R}}^+$.
Combined, these allow us to deduce that the outcome of the machine $\ttt{str}$ is fully contained in the input labels of the AG's right column. 

If this is not the case, either a node in the AG's right column returns $\text{YES}$, or there is a violation of the local constraints.
In the former scenario, constraint~\ref{item:con-tree_propagation} forces this answer to propagate to the root, and from \ref{item:con-root} it suffices for making the output labeling valid regardless of the $\{\text{YES}, \text{NO}\}$ answers of the other rows. This stands in contrast to the adversary's interest, which is to force the nodes to generate an invalid output labeling. Hence, it has no incentive to allow any of the criteria in \ref{item:con-automatic_yes_output} or \ref{item:con-input_ag_neq_input_tm} to be met. 
This idea motivates the classification discussed in Observation~\ref{obs:possible_input_labelings} below. 

\begin{observation}[Classification of Input Labelings]
\label{obs:possible_input_labelings}
Fix a graph $G$ and a $(\mathcal{V}^{\Pi}_{\textnormal{input}}, \mathcal{E}^{\Pi}_{\textnormal{input}})$ input labeling.
We classify the input labeling according to the existence of a valid $(\mathcal{V}^{\Pi}_{\textnormal{output}}, \mathcal{E}^{\Pi}_{\textnormal{output}})$ output labeling that corresponds to one of the following (non-overlapping) possibilities.

\begin{enumerate}[nosep]

\item \label{labeling:invalid} 
% (Bad graph structure)
There is an output labeling such that no subgraph of \textsf{agNode} labeled nodes that output $\bot$ is properly connected to a subgraph of \textsf{satmNode} labeled nodes that output $\bot$.

\item \label{labeling:easy_yes} 
% (Good graph structure, except for the labeling of the columns)
\ref{labeling:invalid} doesn't hold, and the input labels of the AG's right column are \textit{inconsistent} with the SATM's left column.

\item \label{labeling:input_equals_tm_outcome} 
Both \ref{labeling:invalid} and \ref{labeling:easy_yes} do not hold.~\footnote{That is, there is a properly connected pair of an AG and a SATM, and the labels of each node in the left column of the SATM (except for the first and the last ones) are of the form $(\mathsf{satmNode},(\mathsf{tmNode}, (t, h, s)))$, where $t$ equals to the input $b_{\text{in}}$ of the AG node it is connected to (if exists).}

\end{enumerate}
\end{observation}

\paragraph{Implications of Observation~\ref{obs:possible_input_labelings}.}

With this observation in hand, we can better formalize the strategies that the adversary can adopt. We claim that it is in its best interest to select an input labeling captured by possibility~\ref{labeling:input_equals_tm_outcome}. 
To see why this is true, observe that for possibilities \ref{labeling:invalid} and \ref{labeling:easy_yes} we have round-efficient \textsf{LOCAL} algorithms that generate valid output labelings. 

In the case of possibility~\ref{labeling:invalid}, this is a straightforward application of the $O(\log n)$-round algorithms guaranteed by Lemmas~\ref{lemma:bad_ag_invalid} and \ref{lemma:bad_satm_invalid}. In the \textsf{satmNode} labeled nodes, there are no further actions to take. In the \textsf{agNode} ones, if none of them return $\bot$ we are done. Otherwise, the induced subgraph that consists of the \textsf{agNode} labeled nodes that return $\bot$ as their $\mathcal{V}_{\text{output}}^{\text{badAG}}$ output label is guaranteed (by Lemma~\ref{lemma:bad_ag_invalid}) to be a valid AG. In this scenario, we can use some fixed value as the output bit, resulting in identical outputs for all nodes within each row. No additional communication rounds are required. 
Since the SATM is not properly connected, at least one of the criteria listed in constraint~\ref{item:con-automatic_yes_output} is met. 
Constraint~\ref{item:con-automatic_yes_output} implies that when the AG is not connected to a SATM,
there are leaves in the right tree-like structure that automatically return $\text{YES}$. The answer propagates to the root because of constraint~\ref{item:con-tree_propagation}, which leads to a valid output labeling due to \ref{item:con-root}.

Handling possibility~\ref{labeling:easy_yes} is similar. This time, we have the two components (the AG and the SATM) properly connected, but the inputs provided to the right column of the AG are inconsistent with the labels of the SATM. Consequently, constraint~\ref{item:con-input_ag_neq_input_tm} suggests that there exists an AG node that can return the answer $\text{YES}$ without taking into consideration its adversarially chosen input $b_{\text{in}}$. From this point, we proceed as in the previous case.

Since the adversary aims to make any round-efficient algorithm fail, the only plausible possibility left is to select an input labeling that belongs to possibility~\ref{labeling:input_equals_tm_outcome}.

\begin{lemma}[Adapted from \text{\cite[Lemma 7.1]{bgk25}}]
\label{lemma:property_pi}
Given a graph $G \in \mathcal{G}$ and any input labeling that corresponds to possibility~\ref{labeling:input_equals_tm_outcome} in Observation~\ref{obs:possible_input_labelings}, any valid output labeling must satisfy the following:
\begin{enumerate}[nosep]
\item \label{item:def-consistent_rows} In every row in the AG's grid, all the nodes return the same output symbol.

\item \label{item:def-output_equals_input} In at least one of the rows (excluding the top and the bottom rows), the output equals to the input $b_{\text{in}} \in \{ 0,1 \}$ of the rightmost node of the row.

\end{enumerate}
\end{lemma}

\begin{proof}
Constraint~\ref{item:con-consistent_rows} in $\mathcal{C}^{\Pi}$ forces the nodes sharing the same row to output the same symbol, which means that property~\ref{item:def-consistent_rows} is satisfied.
The definition of the graph family $\mathcal{G}$ requires that the right column of the AG is properly connected to the last column of the SATM. In input labelings associated with possibility~\ref{labeling:input_equals_tm_outcome}, the values assigned to the AG's right column match those written in the last column of the SATM, indicating that in constraint~\ref{item:con-input_ag_neq_input_tm} we have $b_{\text{in}} = t$. As such, due to constraint~\ref{item:con-output_01}, the AG's right column nodes output $\text{YES}$ if and only if $b_{\text{in}} = b_{\text{out}}$.
From constraint~\ref{item:con-tree_propagation}, an inner node in a tree outputs YES if and only if one of its children does so. Thus, a YES in one of the leaves propagates until it reaches the root. This is the only way that constraint~\ref{item:con-root} gets satisfied. Hence, to satisfy $\mathcal{C}^{\Pi}$, property~\ref{item:def-output_equals_input} must hold as well.
\end{proof}

\section{Gaps Between Bounded and Unbounded Adversaries from Average-Case Hardness in \texorpdfstring{\TFNP}{TFNP}}
\label{sec:assumptions_imply_gap}

In this section we prove that under certain computational assumptions, the \textsf{LCL} problem $\Pi$ defined in Section~\ref{sec:lcl_problem} has an efficient \textsf{LOCAL} algorithm. In contrast, without these assumptions, any algorithm for solving $\Pi$ performs significantly worse. 

\roundsgap*

The proof is divided into two parts. Section~\ref{subsec:lb} is concerned with input labelings determined by computationally unbounded adversaries. We show a lower bound on the round complexity of any \textsf{LOCAL} algorithm that solves $\Pi$ in this setting. 
On the other hand, Section~\ref{subsec:ub} focuses on polynomial-time adversaries. In this setup, we prove an upper bound on the round complexity by providing an explicit algorithm that generates a valid output labeling within a bounded number of rounds. 

\subsection{Lower Bound in the Case of Unbounded Adversaries}
\label{subsec:lb}

\begin{lemma}
\label{lemma:lb_unbounded_adversary}
There is a constant $c \in (0,1/2]$ such that any \textsf{LOCAL} algorithm that for all large enough $n$ solves the \textsf{LCL} problem $\Pi$ over an $n$~vertex graph in the preset public coins model with success probability $1 - \frac{1}{n}$, when the nodes' inputs are chosen by a computationally unbounded adversary, requires $\Omega(n^{c})$ rounds
\end{lemma}

\begin{proof}
To establish a lower bound, it suffices to show that for any $o(n^c)$ round \textsf{LOCAL} algorithm, there exist infinitely many values of $n \in \mathbb{N}$ with corresponding $n$~vertex graphs where the success probability is strictly less than $1 - \frac{1}{n}$ whenever the input labeling is selected by an unbounded adversary. 
%In fact, we prove a slightly stronger result. Let $\mathcal{A}$ be an $o(n^c)$ round \textsf{LOCAL} algorithm. We show that for all large enough $n$, there exists an $n$~vertex graph on which the success probability is upper bounded. 
To this end, we use the graph family $\mathcal{G}$ (see Def.~\ref{def:graph_family_g}). 

\paragraph{Choice of a Hard Instance.}
The graph $G \in \mathcal{G}$ that we choose has an AG with a square-shaped bottom grid of dimensions $h \times h$ for some integer $h$.
We are interested in running the Turing machine $\mathcal{M}$ (described in Alg.~\ref{alg:hard_tm}) on inputs $(r,w,z)$ where $|r| + |z| = h - 2$ and $|w| = p(|r|)$.
The latter stands for the length of solutions for the search problem $\mathcal{S}$.
The lengths of $r$ and $z$ are required to sum to $h-2$ because we wish the output of $\mathcal{M}$ (which is in $\{0,1\}^{|r| + |z|}$) to fit into the labels of the right column of AG (while ignoring the first and the last nodes).
Since $\mathcal{M}$ is polynomial-time, there exists a polynomial bound on the number of computation steps $\mathcal{M}$ can perform given an input of this form. The same holds for the space complexity of $\mathcal{M}$. 
We set the width $W$ (resp. height $H$) of the bottom grid of the SATM to the bound on the time (resp. space) complexity of $\mathcal{M}$ on such inputs. 
Since the input length is polynomial in $h$, it follows that both $H$ and $W$ are $h^{\Theta(1)}$. The total size of $G$ turns out to be:
\begin{equation}
\label{eq:num_nodes}
n := n(h) =
\underbrace{O (h^2)}_{ \substack{\text {nodes in}\\\textnormal{the AG}}} + \underbrace{O( H \cdot W)}_{\substack{ \text {nodes in}\\\textnormal{the SATM} }}
=
h^{\Theta(1)}
\end{equation}
Let $c \leq 1/2$ be a constant such that $h = \Theta(n^{c})$.

\paragraph{On the Power of the Adversary.}
Being in the role of the adversary, we design an input labeling that corresponds to option~\ref{labeling:input_equals_tm_outcome} in Observation~\ref{obs:possible_input_labelings}. In short, we equip the AG and the SATM of $G$ with input labels that obey the local constraints of each structure, we properly label the edges connecting the two components, and set the inputs $b_{\text{in}}$ of the AG nodes $(1, 0)$, ..., $(h-2, 0)$ in its right column such that they are identical to the values written in the labels of the nodes in the last column of the SATM (which represents the outcome of the \textsf{TM}). 
Thus, the input of the right column of the AG is determined by the output of the \textsf{TM} when invoked on any input string of suitable length that we choose.

Since we rely on a total problem, from the way the \textsf{TM} $\mathcal{M}$ is defined in Alg.~\ref{alg:hard_tm}, the output can be any $h-2$ bits long string $r$ (this is achievable when setting the input to $(r,w,z)$ where $z$ is an empty string and $w$ is such that $(D(r), w) \in \mathcal{S}$). 
Moreover, since the adversary is unbounded, we assume that for every such string $r \in \{ 0,1 \}^{h-2}$ it can find a polynomial-sized $w$ that satisfies $(D(r),w) \in \mathcal{S}$. 
In other words, the adversary has full control over the outcome of the \textsf{TM}. 

Next, we use the fact that the bottom grid of the SATM is large enough to record the computation of $\mathcal{M}$ over every input $(r,w,z)$ for $|r| + |z| = h - 2$. 
When the number of steps in the computation is smaller than the grid's width $W$, we copy the contents of the tape at the time $\mathcal{M}$ halts as many times as necessary to fill all the columns of the grid. If the space required for the computation is smaller than the grid's height $H$, the bottom rows can be filled with blank symbols $\#_{\mathsf{R}}$. 
Consequently, the adversary has full control over the input of the right column of the AG.

\paragraph{Adversarial Strategy.}
The (randomized) strategy of the nodes is known to the adversary. So are the preset public coins $r_{\textnormal{pub}}$. Suppose that the string obtained from concatenating the output bits of the nodes in the AG's right column is $\ttt{str} \in \{ 0,1 \}^{h - 2}$ with high enough probability (over the choice of the private coins, conditioned on the preset public coins $r_{\textnormal{pub}}$). The adversarial strategy we propose aims to set the input of the AG's right column to be $\overline{\ttt{str}}$, the complement of the string. This idea is formalized below.

From this point, the proof closely follows \cite[Thm. 8.1]{bgk25} with the necessary adjustments. 
Suppose that there exists a \textsf{LOCAL} algorithm $\mathcal{A}$ that solves the problem within $o(n^{c}) = o(h)$ rounds.
We choose $h \in \mathbb{N}$ such that $\mathcal{A}$'s runtime is at most $\frac{h}{3}$ and $n := n(h)$ (a polynomial in $h$, see Eq.~(\ref{eq:num_nodes})) is large enough to guarantee that $\mathcal{A}$ succeeds with probability $1 - \frac{1}{n}$, and additionally satisfies $n > \max\{ 2h + 1, 4 \}$. 
Note that this leads to infinitely many valid choices for $n$.

Since the width of the square-shaped bottom grid is $h$, it follows that, conditioned on the preset public randomness, the outputs of nodes from opposite ends of the grid are independent.
For some row, let $u$ be the leftmost node and $v$ the rightmost one. Knowing that the algorithm succeeds with probability at least $1 - \frac{1}{n}$, it follows that:
\[
\begin{cases}
\Pr_{r_\textnormal{sec}}[\text{$u$ outputs 1} | r_{\textnormal{pub}}] \cdot (1 - \Pr_{r_\textnormal{sec}}[\text{$v$ outputs 1}| r_{\textnormal{pub}}]) < \frac{1}{n} \\
\Pr_{r_\textnormal{sec}}[\text{$v$ outputs 1}| r_{\textnormal{pub}}] \cdot (1 - \Pr_{r_\textnormal{sec}}[\text{$u$ outputs 1}| r_{\textnormal{pub}}]) < \frac{1}{n}
\end{cases}
\]
Where the probabilities are taken over the choice of the private coins. It can be inferred that (for $n \geq 4$) either $\Pr[\text{$u$ outputs 1}| r_{\textnormal{pub}}]$ and $\Pr[\text{$v$ outputs 1}| r_{\textnormal{pub}}]$ are both upper bounded by $\frac{2}{n}$, or that they are both lower bounded by $1 - \frac{2}{n}$. 

After the preset public randomness $r_{\textnormal{pub}}$ is revealed, the unbounded adversary selects an input labeling according to the following rule: $v$'s input is $1$ if and only if $\Pr[\text{$u$ outputs 1}| r_{\textnormal{pub}}] < \frac{2}{n}$. 
Here, we rely on the conclusion from the discussion we had above: $v$'s input can be freely chosen by the unbounded adversary. 
The probability that $u$'s output equals the input of $v$ (the rightmost node in $u$'s row) is at most $\frac{2}{n}$. 

The same argument applies to all the rows. From a union bound, the success probability of $\mathcal{A}$ is at most $\frac{2h}{n} = O(\frac{1}{n^{1 - c}})$. When $n > 2h + 1$, this is strictly smaller than $1 - \frac{1}{n}$, thereby reaching a contradiction.
This proves that for any $o(n^c)$~round \textsf{LOCAL} algorithm, there are infinitely many values of $n$ for which the algorithm succeeds with probability strictly smaller than $1 - \frac{1}{n}$ (on carefully chosen $n$~vertex graphs).
\end{proof}

\subsection{Upper Bound in the Case of Bounded Adversaries}
\label{subsec:ub}

\begin{lemma}
\label{lemma:ub_bounded_adversary}
Assume the existence of a public-coin $(T, \mu)$-hard-on-average problem in $\TFNP$. Let $\lambda = \omega(\log n)$ (but $\lambda = O(\poly(n))$) be a parameter such that $T(\lambda)$ is superpolynomial in $n$. 

Then, there is a negligible function $\mathsf{negl}(n)$ and an $O(\lambda)$ round \textsf{LOCAL} algorithm for solving the \textsf{LCL} problem $\Pi$ in the preset public coins model with probability $1 - \mu(\lambda) - \mathsf{negl}(n)$ for all large enough $n$, if the input is given by a $\poly(n)$-time adversary.
\end{lemma}

\begin{proof}
To solve $\Pi$ we follow the steps detailed in Alg.~\ref{alg:efficient_labeling_non_hard}. Lemma~\ref{lemma:ub_bounded_adversary} follows from Claims~\ref{claim:possibility_1}, \ref{claim:possibility_2} and \ref{claim:possibility_3}. The first two refer to scenarios in which $\Pi$ is solvable within $O(\log n)$ rounds. This suffices for our needs, as we take the parameter $\lambda$ to be superlogarithmic in $n$, which implies that we get the desired upper bound.
\end{proof}

\begin{algorithm}
\caption{A \textsf{LOCAL} Algorithm for $\Pi$} \label{alg:efficient_labeling_non_hard}
\begin{algorithmic}[1]
\Statex Executed at node $u$ in a graph $G$.

\State Run the algorithm guaranteed by Lemma~\ref{lemma:bad_ag_invalid} (resp. Lemma~\ref{lemma:bad_satm_invalid}) on the subgraph induced by the nodes labeled with \textsf{agNode} (resp. \textsf{satmNode}) to obtain some output labeling.

\State \textbf{If} $u$ is a \textsf{satmNode} \textbf{then} return that output. \Comment{From now on, assume the node is an \textsf{agNode}}

\State \textbf{If} $u$ is an \textsf{agNode} with a non-$\bot$ output \textbf{then} return that output.

\State Explore the $\log n$ neighborhood. Infer the height of the grid $h$, and the index $i$ of the row.

\If{$u$ is a grid node}

\State \textbf{If} $u$ is on the right column, and any of the criteria in \ref{item:con-automatic_yes_output} is met \textbf{then} return $(1, \text{YES})$\plabel{line:incorrect_connection}

\State \textbf{If} $i \in \{ 0, h-1 \}$ \textbf{then} return $(1, \text{NO})$\plabel{line:borders}

\If{$h < 2\lambda + 2$} learn the input bit $b$ of the rightmost node in the row, return~$(b, \text{YES})$\plabel{line:small_height}

\Else 

\State $b_{\text{out}} = \begin{cases}
    {(r_{\textnormal{pub}})}_{i-1} &\text{if } i \in \{ 2, ..., 2\lambda+1 \}  \\
    0 &\text{otherwise} 
\end{cases}$\plabel{line:out_bit}

\If{$u$ is on the right column}

\State Let $t \in \{ 0,1 \}$ be the bit written in the label of $f(u, \mathsf{connect})$

\If{$b_{\text{in}} \neq t$} return $(b_{\text{out}}, \text{YES})$\plabel{line:wrong_input}

\ElsIf{$b_{\text{in}} \neq b_{\text{out}}$} return $(b_{\text{out}}, \text{NO})$\plabel{line:input_is_output}

\EndIf

\EndIf

\State return $(b_{\text{out}}, \text{YES})$\plabel{line:yes_output}

\EndIf

\Else{ return $\text{YES}$ iff at least one of the children returns YES}\plabel{line:propagation}

\EndIf

\end{algorithmic}
\end{algorithm}

\begin{claim}
\label{claim:possibility_1}
Let $G$ be an $n$~vertex graph equipped with a $(\mathcal{V}^{\Pi}_{\textnormal{input}}, \mathcal{E}^{\Pi}_{\textnormal{input}})$ labeling that corresponds to possibility~\ref{labeling:invalid} in Observation~\ref{obs:possible_input_labelings}.
Then, Alg.~\ref{alg:efficient_labeling_non_hard} solves the \textsf{LCL} problem $\Pi$ in $O(\log n)$ rounds with probability 1 in the preset public coins model.
\end{claim}

\begin{proof}
The proof is reminiscent of the discussion that follows Observation~\ref{obs:possible_input_labelings}.
Lemma~\ref{lemma:bad_satm_invalid} says that in $O(\log n)$ rounds the network is able to detect whether it has a correctly labeled SATM component or not, and to generate a proof in the latter case. The output labeling is always valid (with respect to $\Pi$) according to constraint~\ref{item:con-bad_satm}.
Lemma~\ref{lemma:bad_ag_invalid} shows that with this many rounds the network is capable of generating an output labeling that satisfies the following property: in the subgraph induced by the \textsf{agNode} labeled nodes that return $\bot$, every connected component is a correctly labeled AG.

\textit{Case I.} 
If none of the \textsf{agNode} nodes return $\bot$ on $\Pi^{\textnormal{badAG}}$, no further action is required because the output labels of $\Pi^{\textnormal{badAG}}$ that are not $\bot$ are also valid output labels for $\Pi$ (see constraint~\ref{item:con-bad_ag}).

\textit{Case II.} 
The remaining case has to do with the existence of \textsf{agNode} labeled nodes that return $\bot$ on $\Pi^{\textnormal{badAG}}$. As mentioned earlier, they must constitute a valid AG structure. We argue that the outputs determined by Alg.~\ref{alg:efficient_labeling_non_hard} satisfy $\mathcal{C}^{\Pi}$. 

% C2
First, observe that the outputs of the grid and column trees is formatted correctly, implying that \ref{item:con-grid_nodes} and \ref{item:con-tree_nodes} cannot be violated. 
According to constraint~\ref{item:con-consistent_rows}, for any fixed row, all the nodes must output the same bit. From line~\ref{line:borders} it follows that rows $0$, $h-1$ satisfy this requirement. If the grid is of small height, the argument for the remaining rows follows trivially from \ref{line:small_height}. Otherwise, we return $b_{\text{out}}$ determined by the rule in line~\ref{line:out_bit}. Naturally, this means that per row, the nodes agree with one another.
Line~\ref{line:propagation} suggests that the behavior of the inner nodes of the column trees respects \ref{item:con-tree_propagation}.
From line~\ref{line:yes_output}, the grid nodes (except of rows $0$ and $h-1$ and the right column) output $\text{YES}$. 
Since possibility~\ref{labeling:invalid} is concerned with an AG being disconnected from the SATM, at least one of the criteria in constraint~\ref{item:con-automatic_yes_output} is met. Line~\ref{line:incorrect_connection} suggests that at least one node in the right column outputs $\text{YES}$ too. 
Having at least one $\text{YES}$ labeled node in each column and knowing that constraint~\ref{item:con-tree_propagation} holds, we infer that all the roots return $\text{YES}$. As a result, \ref{item:con-root} is satisfied.

% C3
Constraint~\ref{item:con-automatic_yes_output} is obviously satisfied due to line~\ref{line:yes_output}.
Observing that constraint~\ref{item:con-top_bottom} is satisfied is a straightforward implication of line~\ref{line:borders}. The same goes for constraint~\ref{item:con-input_ag_neq_input_tm} and line~\ref{line:wrong_input}. Lastly, line~\ref{line:input_is_output} (or \ref{line:small_height} if the height is small) suggests that constraint~\ref{item:con-output_01} is not violated.

To summarize, the output labeling generated by Alg.~\ref{alg:efficient_labeling_non_hard} satisfies $\mathcal{C}^{\Pi}$.
\end{proof}

\begin{claim}
\label{claim:possibility_2}
Let $G$ be an $n$~vertex graph equipped with a $(\mathcal{V}^{\Pi}_{\textnormal{input}}, \mathcal{E}^{\Pi}_{\textnormal{input}})$ labeling that corresponds to possibility~\ref{labeling:easy_yes} in Observation~\ref{obs:possible_input_labelings}.
Then, Alg.~\ref{alg:efficient_labeling_non_hard} solves the \textsf{LCL} problem $\Pi$ in $O(\log n)$ rounds with probability 1 in the preset public coins model.
\end{claim}

\begin{proof}
Possibility~\ref{labeling:easy_yes} is concerned with having a subgraph induced by the \textsf{agNode} labeled graph (that returns $\bot$ on $\Pi^{\text{badAG}}$), which is properly connected to a correctly labeled SATM structure. We rely on the proof of Claim~\ref{claim:possibility_1}, while making the necessary changes.

Unfortunately, since we assume that the two components are properly connected, constraint~\ref{item:con-automatic_yes_output} 
doesn't suffice for determining that one of the nodes in the AG's right column returns $\text{YES}$. 
Despite that, since the input labeling corresponds to possibility~\ref{labeling:easy_yes}, at least one of the nodes $(1,0), ..., (h-2, 0)$ in the right column disagrees with the value recorded in the label of its SATM neighbor. Consequently, constraint~\ref{item:con-input_ag_neq_input_tm} forces one of the nodes to answer $\text{YES}$. In Alg.~\ref{alg:efficient_labeling_non_hard}, this is implemented by line~\ref{line:wrong_input}. Knowing that at least of the nodes in the AG's right column outputs $\text{YES}$, the same reasoning as in the proof of Claim~\ref{claim:possibility_1} can be applied. We refrain from repeating the details.
\end{proof}

\begin{claim}
\label{claim:possibility_3}
Let $G$ be an $n$~vertex graph equipped with a $(\mathcal{V}^{\Pi}_{\textnormal{input}}, \mathcal{E}^{\Pi}_{\textnormal{input}})$ labeling that corresponds to possibility~\ref{labeling:input_equals_tm_outcome} in Observation~\ref{obs:possible_input_labelings}.
Then, there is a negligible function $\mathsf{negl}(n)$ such that Alg.~\ref{alg:efficient_labeling_non_hard} solves the \textsf{LCL} problem $\Pi$ in $O(\lambda)$ rounds with probability $1 - \mu(\lambda) - \mathsf{negl}(n)$ in the preset public coins model against a bounded adversary.
\end{claim}

\begin{proof}
As in Claim~\ref{claim:possibility_2}, we have an AG structure attached to a SATM, where both components are correctly labeled. Now, the labeling of the AG's right column is consistent with the SATM's last column. In other words, the input of the AG reflects the outcome of the Turing machine $\mathcal{M}$ simulated by the SATM. 
We prove that a computationally bounded adversary is unlikely to successfully find an input for $\mathcal{M}$ such that the outcome is the complement of the right column's output. 
Consequently, in at least one of the rows the output equals the input. 
This paves the way for justifying why Alg.~\ref{alg:efficient_labeling_non_hard} produces a valid output labeling with high probability.

We begin by noting that from the same explanation used in Claim~\ref{claim:possibility_1}, all the $\mathcal{C}^{\Pi}$ constraints continue to be satisfied except for \ref{item:con-root}, which requires all the roots to return $\text{YES}$. The reason is that we previously relied on constraints \ref{item:con-automatic_yes_output} (in Claim~\ref{claim:possibility_1}) and \ref{item:con-input_ag_neq_input_tm} (in Claim~\ref{claim:possibility_2}) to justify why at least one of the nodes in the right column answers with $\text{YES}$. Whenever the input labeling fits the third possibility~\ref{labeling:input_equals_tm_outcome}, this is no longer the case. 

The easier scenario has to do with a ``small'' AG. That is, one with height at most $2\lambda + 2$. Since the grid is vertical, if the height is at most $2\lambda + 2$, so is the width. Therefore, the grid nodes can explore the topology of the whole graph within $O(\lambda)$ rounds, and solve the problem easily after learning the input of the rightmost node in their row. Then, the nodes in the right column also output YES, resulting in the root returning YES as well, thereby satisfying \ref{item:con-root}.

In grids with greater height, suppose that for every row $i \in \{ 1,...,h - 2 \}$, the adversary manages to set the input to be the complement of the output $b_{\text{out}}$ chosen for the rightmost node in the row. 
Our Alg.~\ref{alg:efficient_labeling_non_hard} sets these output bits to $r_{\textnormal{pub}} \circ \Vec{0}$, where $|r_{\textnormal{pub}}| = 2\lambda$.
When the labeling is associated with possibility~\ref{labeling:input_equals_tm_outcome}, the input of this column is determined by the output of the \textsf{TM} $\mathcal{M}$ used in the SATM component. 
So, the adversary has to find an input on which $\mathcal{M}$ returns $\overline{r_{\textnormal{pub}}} \circ \Vec{1}$.

Allegedly, the adversary must find a string $w \in \{0,1 \}^{p(\lambda)}$ that satisfies $(D(\overline{r_{\textnormal{pub}}}), w) \in \mathcal{S}$.
This is because $\mathcal{M}$ returns $\overline{r_{\textnormal{pub}}} \circ 1^{h - 2 - 2\lambda}$ on input $(\overline{r_{\textnormal{pub}}}, w, 0^{h - 2 - 2\lambda})$.
However, this is not entirely true, as this output can also be achieved by using any prefix of $r_{\textnormal{pub}}$ obtained from removing all the zeros that appear in the end of the string. More concretely, suppose that $r_{\textnormal{pub}} = r_{\textnormal{pub}}' \circ 0^l$ for some integer $l$. 
Then, the adversary may try to find $w$ such that $(D(\overline{r_{\textnormal{pub}}'}) ,w) \in \mathcal{S}$.
On input $(r_{\textnormal{pub}}', w, 0^{h - 2 - 2\lambda +l})$, the output of $\mathcal{M}$ is still $\overline{r_{\textnormal{pub}}'} \circ 1^{h - 2 - 2\lambda + l} = \overline{r_{\textnormal{pub}}} \circ 1^{h - 2 - 2\lambda}$. 

We show that with high probability over the choice of $r_{\textnormal{pub}}$, the length of the zeros suffix $l$ cannot be too long. Specifically, we show that the probability to sample $r_{\textnormal{pub}} = r_{\textnormal{pub}}' \circ 0^l$ where $l > \lambda$ is at most negligible.
This stems from $r_{\textnormal{pub}}$ being sampled uniformly at random.

\begin{equation}
\label{eq:bad_preset_randomness}
\Pr_{r_\textnormal{pub}}\left[ r_{\textnormal{pub}} = r_{\textnormal{pub}}' \circ 0^l \land l > \lambda \right] \leq 
\sum_{l = \lambda + 1}^{2\lambda} \frac{1}{2^{l}} \leq
\frac{1}{2^\lambda} = 
\frac{1}{2^{\omega(\log n)}} = 
\frac{1}{n^{\omega(1)}}
\end{equation}
As a consequence, with probability all but negligible, $r_{\textnormal{pub}} = r_{\textnormal{pub}}' \circ 1 \circ  0^*$ where $|r_{\textnormal{pub}}'\circ 1| \geq \lambda$.
Whenever this holds, the adversary is required to solve the search problem $\mathcal{S}$ on an instance $x$ that a uniformly random string of length at least $\lambda$ was used for sampling it.

To be noted, we ignored the case of $\overline{r_{\textnormal{pub}}} = \Vec{1}$. Alg.~\ref{alg:hard_tm} specifies $\Vec{1}$ as the default output whenever the solution given as part of the input is wrong, meaning that this output can be achieved easily. Nevertheless, if $r_{\textnormal{pub}}$ is sampled uniformly at random, this event occurs only with negligible probability.
Let $\mathsf{negl}(n)$ be the probability that $r_{\textnormal{pub}} = \Vec{1}$ plus the bound specified in Eq.~(\ref{eq:bad_preset_randomness}). 

The distributional problem $(\mathcal{S}, \mathcal{D})$ is assumed to be public-coin $(T, \mu)$-hard-on-average (against non-uniform solvers). We set the parameter $\lambda$ so that $T(\lambda)$ is superpolynomially in $n$. 
The probability that a non-uniform \textsf{PPT} algorithm successfully finds a solution for a random instance (when the randomness used for sampling the instance is publicly known) is $\mu(\lambda)$. 
The \textsf{TM} assimilated in the graph returns the desired output only if the adversary manages to solve $\mathcal{S}$. Hence, with all but $\mu(\lambda) + \mathsf{negl}(n)$ probability the \textsf{PPT} adversary fails, suggesting that there exists a row $i$ that satisfies $b_{\text{in}} = b_{\text{out}}$. 
According to the algorithm, the corresponding grid node outputs YES (which is consistent with constraint~\ref{item:con-output_01}). As before, the answer propagates to the root, ensuring that constraint~\ref{item:con-root} is satisfied also in the right column tree. 
Earlier, we claimed that the other constraints are satisfied as well. 
Thus, the output labeling generated by Alg.~\ref{alg:efficient_labeling_non_hard} satisfies $C^{\Pi}$. Evidently, the round complexity of the algorithm is $O(\lambda)$.
\end{proof}

\paragraph{Gaps Between The Models.}

In \cite{bgk25}, the gap in the round complexity (between the private coins model and the oblivious model) is exponential. To be specific, the algorithm that uses shared randomness takes $O(\log n)$ rounds. 
Their technique implies a gap between the oblivious model and the preset public coins model when the adversary is bounded (see Lemma~\ref{lemma:gap_oblivious}) and one between the private coins model and the preset public coins model when the adversary is unbounded (see Lemma~\ref{lemma:gap_private}).

The main difficulty in using this technique to establish a separation between the setting that considers efficient adversaries and the one that allows them to be unbounded lies in the inability to locally determine whether the height of an AG is sufficiently large.
\cite{bgk25} address this issue by solving the problem trivially (by exploring the whole graph) whenever the height is smaller than $c\log n$ for some $c > 0$. If the grid's height is greater than that, the solution that relies on the shared randomness is valid with probability at least 
$1 - \frac{1}{n^c}$. 

This is no longer true in our setup, as the height of the grid must be at least $\lambda$.
As discussed in Corollaries~\ref{cor:poly_hardness} and~\ref{cor:subexp_hardness}, the upper bound on the round complexity against $\poly(n)$-time adversaries turns out to be a small polynomial (resp.\ polylogarithmic) if we assume (public-coin) polynomial (resp.\ subexponential) hardness-on-average of the distributional problem used for construction. That is, the gap in the preset public coins model between bounded and unbounded adversaries is subexponential (under strong assumptions) or polynomial (under weaker assumptions).

\section{Average-Case Hardness in \texorpdfstring{\TFNP/\poly}{TFNP/poly} from Gaps Between Bounded and Unbounded Adversaries}
\label{sec:probability_gap}

In this section, we assume the existence of an \textsf{LCL} problem $\Psi$ where unbounded adversaries are more harmful than efficient ones. 
We show that this implies a known computational assumption. 

\probabilitygap* 

We actually prove the following theorem, which implies Theorem~\ref{thm:probability_gap}:

\begin{theorem}
\label{thm:probabilitygap_enhanced}
Let $\Psi$ be an \textsf{LCL} problem.
Assume that there exist a negligible function $\epsilon: \mathbb{N} \rightarrow [0,1]$, a function $r: \mathbb{N} \rightarrow \mathbb{N}$ and an $r(n)$ round \textsf{LOCAL} algorithm $\mathcal{A}$ that for all large enough $n \in \mathbb{N}$ solves $\Psi$ on any $n$~vertex graph $G$ and input labeling chosen by a $\poly(n)$-time adversary with probability at least $1 - \epsilon(n)$ (over the choice of the preset randomness $r_{\text{pub}}$, the private coins $r_{\text{sec}}$ and the coin tosses of the adversary).

Moreover, assume that there exist a graph family $\mathcal{G}:=\{ G_n \}_{n \in \mathbb{N}}$, a sequence of integers $\mathcal{H}$ and a function $\varphi: \mathbb{N} \rightarrow [0,1]$ that is at least the inverse of a polynomial, such that for every $n \in \mathcal{H}$, $\mathcal{A}$ fails on $G_{n}$ and input labelings generated by unbounded adversaries with probability at least $\varphi(n)$ (over the choice of $r_{\text{pub}}$ and $r_{\text{sec}}$).

Then, there is a public-coin infinitely-often polynomially hard-on-average problem in $\TFNP/\poly$.
\end{theorem}

It's not hard to see that Theorem~\ref{thm:probabilitygap_enhanced} implies Theorem~\ref{thm:probability_gap}. Denote $\mathsf{ROUND}^{\textsf{B}}_{\epsilon(n)}(\Psi)$ by $r(n)$, and let $\mathcal{A}$ be one of the \textsf{LOCAL} algorithms that attains this round complexity. That is, for all large enough $n \in \mathbb{N}$, $\mathcal{A}$ succeeds with probability $1- \epsilon(n)$ on any $n$~vertex graph and input labeling chosen by a bounded adversary.

The assumption in Theorem~\ref{thm:probability_gap} suggests that against an unbounded adversary, there cannot be an $r(n)$~round \textsf{LOCAL} algorithm (in particular, algorithm $\mathcal{A}$) that succeeds with probability $1 - \frac{1}{p(n)}$ for all large enough $n$. 
This is the same as saying that for infinitely many values of $n$, there exists a ``hard'' $n$~vertex graph $G_{n}$ on which the algorithm fails with probability at least $\frac{1}{p(n)}$.
Let us denote this sequence of graph sizes by $\mathcal{H}$, and the lower bound on the failure probability by $\varphi(n)$. 
For convenience, we restrict $\mathcal{H}$ to integers larger than $n_0$, where $n_0$ is chosen such that for all $n \geq n_0$ the inequality $\epsilon(n) \leq \frac{1}{2} \varphi(n)$ holds. Such an $n_0$ must exist because $\varphi(n)$ is an inverse polynomial and $\epsilon(n)$ is negligible.
\[
\mathcal{H} := \left\{ 
\mathbb{N} \ni n \geq n_0 \Bigg| \substack{\text{\normalsize$\exists n$~vertex graph $G_n$ s.t $\mathcal{A}$ fails on $G_n$}\\\text{\normalsize w.p $\varphi(n)$ when the input labeling is chosen}\\\text{\normalsize by an unbounded adversary}}
\right\}
\]
The graph family $\mathcal{G} := \{ G_n \}_{n \in \mathbb{N}}$ is constructed as follows. For $n \in \mathcal{H}$, we use the $G_n$ guaranteed by the definition of $\mathcal{H}$. For any $n \in \mathbb{N} \backslash \mathcal{H}$, let $G_n$ be an arbitrary $n$~vertex connected graph.

%It follows that the assumption in Theorem~\ref{thm:probability_gap} yields the assumption in Theorem~\ref{thm:probabilitygap_enhanced}.

\subsection{Problem Definition}

\paragraph{Search Problem $\mathcal{S}$ (Informal).}

The full definition of the search problem $\mathcal{S}$ (Def.~\ref{def:search_problem}) is going to depend on the algorithm $\mathcal{A}$, the graph family $\mathcal{G}$, the sequence of graph sizes $\mathcal{H}$ and on $\mathcal{R} = \{ R_n\}_{n \in \mathbb{N}}$, a sequence of multisets of private coins $R_n = \{ r_{\textnormal{sec}}^{(i)} \}_{i=1}^{|R_n|}$. 

$\mathcal{A}$, $\mathcal{G}$ and $\mathcal{H}$ are as above. As for $\mathcal{R}$, we will go into the details later, but for the time being, one can think about $R_n \in \mathcal{R}$ as a multiset of samples of $r_{\textnormal{sec}}$, prepared in advance (so no actual coin tossing takes place while verifying a solution), where the goal is to use these samples in order to characterize the behavior of algorithm $\mathcal{A}$.

For $n \in \mathcal{H}$, given an instance $(1^n, r_{\textnormal{pub}})$, where the second item is a tape of preset randomness of length $b^{\mathcal{A}}_{\textnormal{pub}}(n)$, the task is to find an input labeling $x$. For $n \in \mathbb{N} \backslash \mathcal{H}$, the solution is trivial and it is simply $0^n$.

\begin{equation*}
\label{eq:searchproblem}
\mathcal{S} = 
\mathcal{S}_{\mathcal{A}, \mathcal{G}, \mathcal{H}, \mathcal{R}} := 
\left\{ 
\bigg((1^n, r_{\textnormal{pub}}), x \bigg) \Bigg| 
\substack{ \text{\normalsize The fraction of $r_{\textnormal{sec}}^{(i)}$'s in $R_n$ for which}\\\text{\normalsize $\mathcal{A}$'s output on $G_n$, $x$, $r_{\textnormal{pub}}$, $r_{\textnormal{sec}}^{(i)}$ is}\\\text{\normalsize invalid is at least $(\epsilon(n) + \varphi(n))/2$} }
\right\}_{n \in \mathcal{H}} \bigcup
\left\{ 
\bigg((1^n, r_{\textnormal{pub}}), 0^n \bigg)
\right\}_{n \in \mathbb{N} \backslash \mathcal{H}} 
\end{equation*}

In essence, the problem should capture the hardness of finding a hard input labeling $x$ on which the \textsf{LOCAL} algorithm fails with probability more than inverse polynomial. Think about the entity trying to solve $\mathcal{S}$ as playing the role of the adversary in the original \textsf{LCL} problem. 
Assuming that unbounded adversaries outperform efficient ones with respect to the success probability that they attain, it is reasonable to think that unbounded solvers for our search problem have an advantage over efficient solvers, laying the groundwork for establishing a computational assumption related to search problems.

\begin{remark}[Specifying the Graph Size]
\label{remark:graph_size}
The size $n$ is mentioned as part of the instance because it typically cannot be inferred from the length of $r_{\textnormal{pub}}$. Not having the graph size explicitly specified strengthens the solver, allowing it to choose $n$. This is not our intention. Another, more technical reason for mentioning the graph size, is that we measure the complexity as a function of $n$. E.g., the verification procedure would need to emulate the execution of the \textsf{LOCAL} algorithm across all $n$ nodes. Given that the runtime of each node is assumed to be $\poly(n)$, the overall verification time is polynomial in $n$ as well. To ensure that $\mathcal{S}$ is efficiently verifiable, the size of the instance must be $\Omega(n)$. Thus, it would be convenient to use $1^n$ in order to force the length to be at least $n$ ($b_{\text{pub}}^{\mathcal{A}}(n)$ is polynomially bounded, but we don't impose a lower bound over it). 
\end{remark}

\begin{remark}
[Infinitely-Often Hardness]
\label{remark:io}
Theorem~\ref{thm:probabilitygap_enhanced} establishes a distributional search problem that is only \textit{infinitely-often} hard-on-average. In other words, for any (non-uniform) \textsf{PPT} solver, the success probability is at most negligible for infinitely many $n$'s (rather than for all sufficiently large $n$, which is the more conventional notion of hardness that we use in Def.~\ref{def:hard_on_average}).
We emphasize that this stems from the fact that our distributional problem is guaranteed to be hard only for instances associated with values of $n$ that belong to $\mathcal{H}$.

The problem's hardness can be strengthened to hold for all large enough $n$, assuming that the failure probability of $\mathcal{A}$ against the unbounded adversary is at least non-negligible for all sufficiently large $n$ (instead of for infinitely many $n$'s).
\end{remark}

\paragraph{Characterizing Multisets.}
To characterize the behavior of the nodes when running the \textsf{LOCAL} algorithm $\mathcal{A}$ in an $n$~vertex graph $G_n \in \mathcal{G}$ when the preset public coins are $r_{\textnormal{pub}}$, we first sample $|R_n|$ tapes of private coins $r_{\textnormal{sec}}^{(i)}$ of length $n \cdot b^{\mathcal{A}}_{\textnormal{sec}(n)}$ (the size of $R_n$ will soon be determined). The \textit{characterizing multiset} contains the output labelings returned by $G_n$'s nodes for the different $r_{\textnormal{sec}}^{(i)}$'s in $R_n$. Roughly, the fraction of invalid output labelings within the \textit{characterizing multiset} approximates the actual probability that algorithm $\mathcal{A}$ yields an invalid output labeling (over the choice of the private coins $r_{\textnormal{sec}}$, when sampled uniformly at random by the nodes). Below, we formalize this intuition.

\begin{definition}[Expected Value for a Multiset]
\label{def:expected_value}
Let $\Psi$ be an \textsf{LCL} problem.
Let $n \in \mathbb{N}$, $G_n$ be an $n$~vertex graph and $x \in \{ 0, 1\}^{O(n)}$ be an input labeling.

$T$ is a multiset $\{ y^{(i)} \}_{i = 1}^{|T|}$ of output labelings. The \textit{expected value} for $T$ is:
\[
Q_{G_n, x}(T) =
\underset{i \stackrel{\$}{\leftarrow} [|T|]}{\mathbb{E}} \left[ \mathbbm{1}\{ y^{(i)}\text{ is valid} \} \right]
= 
\frac{1}{|T|} \sum_{i \in [|T|]} \mathbbm{1}\{ y^{(i)}\text{ is valid} \}
\]
\end{definition}

Where being valid means that the verification procedure associated with the \textsf{LCL} problem $\Psi$ accepts on the graph $G_n$ with input labeling $x$ and output labeling $y^{(i)}$.

\begin{definition}[$\gamma(n)$-Characterizing Multiset]
\label{def:characterizing_multiset}
Let $\Psi$ be an \textsf{LCL} problem, $\mathcal{A}$ be a \textsf{LOCAL} algorithm and $\gamma: \mathbb{N} \rightarrow [0,1]$.
Let $n \in \mathbb{N}$. Let $G_n$ be an $n$~vertex graph, $r_{\textnormal{pub}} \in \{ 0,1 \}^{b^{\mathcal{A}}_{\textnormal{pub}}(n)}$ be a preset public string and $x \in \{ 0, 1\}^{O(n)}$ be an input labeling.

$T$ is a $\gamma(n)$-\textit{characterizing multiset} for $\mathcal{A}$ on $G_n$, $r_{\textnormal{pub}}$, and $x$ if the following holds.
\[
\left| Q_{G_n, x}(T) - \Pr_{r_{\textnormal{sec}}} \left[ \mathcal{A}(G_n, x, r_{\textnormal{pub}}, r_{\textnormal{sec}}) \text{ is valid}\right] \right| \leq \gamma(n)
\]
Where $Q_{G_n, x}(T)$ is the \textit{expected value} of multiset $T$, and $\mathcal{A}(G_n, x, r_{\textnormal{pub}}, r_{\textnormal{sec}})$ is the output labeling generated by the algorithm $\mathcal{A}$. 
\end{definition}

\noindent \textit{The role of the parameter $\gamma(n)$}. It should be interpreted as a measure of the quality of the characterizing multiset. The smaller it is, the better the characterization. We will choose $\gamma(n)$ to be the inverse of a polynomial.

\begin{lemma}
\label{lemma:sample_characterizing_multiset}
Let $\Psi$ be an \textsf{LCL} problem, $\mathcal{A}$ be a \textsf{LOCAL} algorithm, $\gamma: \mathbb{N} \rightarrow [0,1]$ and $\delta: \mathbb{N} \rightarrow [0, 1]$.
Let $n \in \mathbb{N}$. Let $G_n$ be an $n$~vertex graph, $r_{\textnormal{pub}} \in \{ 0,1 \}^{b^{\mathcal{A}}_{\textnormal{pub}}(n)}$ be a preset public string and $x \in \{ 0, 1\}^{O(n)}$ be an input labeling.

\noindent $T$ is a multiset of output labelings generated by following the steps below.
\begin{itemize}[nosep]

\item Sample a multiset $R_n = \{ r_{\textnormal{sec}}^{(i)} \}_{i = 1}^{l}$ of $l = l(n) := 2 \cdot \gamma^{-2}(n) \cdot \ln(2\delta^{-1}(n))$ independent uniformly random binary strings of length $n \cdot b^{\mathcal{A}}_{\textnormal{sec}}(n)$. 

\item For $i \in [l]$, the output labeling $y^{(i)}$ is $\mathcal{A}(G_n, x, r_{\textnormal{pub}}, r_{\textnormal{sec}}^{(i)})$. Note, $|y^{(i)}| \leq O(n)$, because the output labels per node are of constant length.

\item $T = T^{R_n}_{x, r_\text{{pub}}} := \{ y^{(i)} \}_{i=1}^l$
\end{itemize}
Then, $T$ is a $\gamma(n)$-characterizing multiset for $\mathcal{A}$ on $G_n$, $r_{\textnormal{pub}}$, and $x$ with probability $1 - \delta(n)$.
\end{lemma}

\begin{proof}
Fix $i \in [l]$. $\chi_i$ is a random variable supported on $[-1,1]$, defined as follows.
\[
\chi_i := 
\mathbbm{1}\{ y^{(i)} \text{ is valid} \} - \Pr_{r_{\textnormal{sec}}}\left[ \mathcal{A}(G_n, x, r_{\textnormal{pub}}, r_{\textnormal{sec}}) \text{ is valid}  \right] 
\]
Be aware that the variables in $\{\chi_i \}_{i=1}^l$ are i.i.d, because the $y^{(i)}$'s are. Moreover, since the $y^{(i)}$'s are generated by sampling $r_{\textnormal{sec}}^{(i)}$ uniformly at random, we have $\mathbb{E}[\chi_i] = 0$. 

\allowdisplaybreaks
\begin{align*}
& \Pr_{R_n} \left[ \left| Q_{G_n, x}(T) - \Pr_{r_{\textnormal{sec}}}\left[  \mathcal{A}(G_n, x, r_{\textnormal{pub}}, r_{\textnormal{sec}}) \text{ is valid}\right] \right| > \gamma(n) \right]  & \text{(Def.~\ref{def:expected_value})}\\  &=
\Pr_{R_n} \left[ \left| \frac{1}{l} \sum_{i=1}^l \mathbbm{1}\{ y^{(i)} \text{ is valid} \} - \Pr_{r_{\textnormal{sec}}}\left[  \mathcal{A}(G_n, x, r_{\textnormal{pub}}, r_{\textnormal{sec}}) \text{ is valid}  \right] \right| > \gamma(n) \right]  & \text{(Symmetry)}\\ &\leq
2 \cdot \Pr_{R_n} \left[ \sum_{i=1}^l  \chi_i  > l \cdot \gamma(n) \right]  & \left( \begin{subarray}{l} \text{\normalsize Chernoff bound} \\ \text{\normalsize \cite[Thm.~A.1.16]{as00}} \end{subarray} \right)\\ &\leq 
2\cdot \text{exp}\left( - \frac{\gamma^2(n) \cdot l}{2}\right) \\ &=
2\cdot \text{exp}\left( - \frac{\gamma^2(n) \cdot 2 \cdot\gamma^{-2}(n) \cdot \ln(2\delta^{-1}(n))}{2}\right) \\ &=
\delta(n)
\end{align*}
\end{proof}

\begin{lemma}
\label{lemma:same_coins_for_all}
Let $\Psi$ be an \textsf{LCL} problem, $\mathcal{A}$ be a \textsf{LOCAL} algorithm, and $\gamma: \mathbb{N} \rightarrow [0,1]$. Let $n \in \mathbb{N}$ and let $G_n$ be an $n$~vertex graph.
Denote the upper bound on the length of the inputs labels in $\Psi$ by $d = O(1)$. Set $\delta(n)$ to be $2^{- (b^{\mathcal{A}}_{\textnormal{pub}}(n) + d \cdot n + n)}$.

$R_n = \{ r_{\textnormal{sec}}^{(i)} \}_{i=1}^{l}$ is a multiset of private coins, generated according to the instructions in Lemma~\ref{lemma:sample_characterizing_multiset}, such that it yields a $\gamma(n)$-characterizing multiset for $\mathcal{A}$ on $G_n$, $r_{\textnormal{pub}}' \in \{ 0,1 \}^{b^{\mathcal{A}}_{\textnormal{pub}}(n)}$ (preset public randomness) and $x' \in \{ 0, 1 \}^{d \cdot n}$ (an input labeling) for some $r_{\textnormal{pub}}'$ and $x'$, with probability $1 - \delta(n)$.

Then, with all but negligible probability, it gives rise to a $\gamma(n)$-characterizing multiset for $\mathcal{A}$ on $G_n$, $r_{\textnormal{pub}}$, and $x$, for \underline{every} $r_{\textnormal{pub}} \in \{ 0,1 \}^{b^{\mathcal{A}}_{\textnormal{pub}}(n)}$ and $x \in \{ 0, 1 \}^{d \cdot n}$.
\end{lemma}

\begin{proof}
Pay attention to the fact that the construction of $R_n$ in Lemma~\ref{lemma:sample_characterizing_multiset} doesn't take into consideration what preset public randomness and input labeling are being used. 
For a fixed choice of $x'$ and $ r_{\textnormal{pub}}'$, we know that $R_n$ can be used to construct a $\gamma(n)$-characterizing multiset with probability $1 - \delta(n)$. From a union bound, we can make it hold for every choice of $x, r_{\textnormal{pub}}$ simultaneously. Denote by $T_{{x', r_{\textnormal{pub}}'}}^{R_n}$ the characterizing multiset obtained by following the steps in Lemma~\ref{lemma:sample_characterizing_multiset}.
\begin{align*}
\Pr_{R_n} \left[ 
\substack{\text{\normalsize$\exists x, r_{\textnormal{pub}}$ s.t $T_{x', r_{\textnormal{pub}}'}^{{R_n}}$ is \textbf{not} a $\gamma(n)$-}\\\text{\normalsize characterizing multiset for $\mathcal{A}$}\\\text{\normalsize on $G_n$, $x$, $r_{\textnormal{pub}}$}} 
\right] 
\leq
\sum_{x, r_{\textnormal{pub}}} \delta(n) \leq
2^{b^{\mathcal{A}}_{\textnormal{pub}}(n) + d \cdot n} \cdot \delta(n) 
\leq 
\frac{1}{2^n}
\end{align*}
\end{proof}

The takeaway from the last two lemmas is that given an \textsf{LCL} problem problem $\Psi$, a \textsf{LOCAL} algorithm $\mathcal{A}$, an $n \in \mathbb{N}$ and an $n$~vertex graph $G_n$, there is a method for generating a multiset $R_n = \{ r_{\textnormal{sec}}^{(i)} \}_{i =1}^{l}$ that with high probability can be used for constructing $\gamma(n)$-characterizing multisets of output labelings for $\mathcal{A}$ on $G_n$, $r_{\text{pub}}$ and $x$ for every pair of a preset public random string $r_{\textnormal{pub}}$ and an input labeling $x$.

\paragraph{Search Problem $\mathcal{S}$ (Formal).} 

\begin{definition}[Full Definition of Problem $\mathcal{S}_{\mathcal{A}, \mathcal{G}, \mathcal{H}, \mathcal{R}}$]
\label{def:search_problem}
Let the \textsf{LCL} problem $\Psi$, the functions $r(n)$, $\epsilon(n)$, $\varphi(n)$, the \textsf{LOCAL} algorithm $\mathcal{A}$, the graph family $\mathcal{G}$ and the sequence of graph sizes $\mathcal{H}$ be as promised in Theorem~\ref{thm:probabilitygap_enhanced}. 
Additionally, we define the following:

\begin{itemize}[nosep]

\item $\gamma(n):= \frac{\varphi(n) - \epsilon(n)}{4}$ (note that this is at least the inverse of a polynomial).

\item $\theta(n) := \frac{\epsilon(n) + \varphi(n)}{2}$

\item $l(n) : = 2 \cdot \gamma^{-2}(n) \cdot (b^{\mathcal{A}}_{\textnormal{pub}}(n) + d \cdot n + n)$ (a polynomial).\footnote{This is what we get from plugging in the $\delta(n)$ chosen in Lemma~\ref{lemma:same_coins_for_all} into the expression for $l(n)$ mentioned in Lemma~\ref{lemma:sample_characterizing_multiset}.}

\item $R_n = \{ r_{\textnormal{sec}}^{(i)} \}_{i=1}^{l(n)}$ is a multiset of strings in $\{ 0,1 \}^{n \cdot b^{\mathcal{A}}_{\textnormal{sec}}(n)}$. $\mathcal{R} : = \{ R_n \}_{n \in \mathbb{N}}$.

\end{itemize}
Then, the search problem $\mathcal{S}:=\mathcal{S}_{\mathcal{A}, \mathcal{G}, \mathcal{H}, \mathcal{R}}$ is defined as follows. On instance $(1^n, r_{\textnormal{pub}})$ (where $r_{\textnormal{pub}} \in \{ 0,1 \}^{b^{\mathcal{A}}_{\textnormal{pub}}(n)}$), find $x$ such that:
\begin{itemize}[nosep]
\item If $n \in \mathcal{H}$, running algorithm $\mathcal{A}$ on graph $G_n\in \mathcal{G}$, input labeling $x \in \{ 0,1 \}^{O(n)}$ and preset public coins $r_{\textnormal{pub}}$ yields an invalid output labeling on at least $\theta(n)$ fraction of $R_n$.

\item If $n \in \mathbb{N} \backslash \mathcal{H}$, $x = 0^n$.
\end{itemize}
\end{definition}

It is worth noting that the formal definition of $\mathcal{S}_{\mathcal{A}, \mathcal{G}, \mathcal{H}, \mathcal{R}}$ doesn't require $\mathcal{R} = \{ R_n \}_{n \in \mathbb{N}}$ to be generated in a specific way. $R_n$ is only required to contain $l(n)$ multisets of binary strings of the appropriate length. When proving hardness, we restrict ourselves to certain choices of $\mathcal{R}$. Recall, our objective is to prove the \textit{existence} of a hard problem. Thus, it suffices to show that there are $\mathcal{R}$'s for which $\mathcal{S}_{\mathcal{A}, \mathcal{G}, \mathcal{H}, \mathcal{R}}$ is public-coin (polynomially) hard-on-average.

\begin{lemma}[Efficient Verification]
\label{lemma:search_problem}
For every algorithm $\mathcal{A}$, graph family $\mathcal{G} = \{ G_n \}_{n \in \mathbb{N}}$, sequence $\mathcal{H} \subseteq\mathbb{N}$ and sequence of multisets $\mathcal{R} = \{ R_n \}_{n \in \mathbb{N}}$ as in Def.~\ref{def:search_problem}, 
the search problem $\mathcal{S}_{\mathcal{A}, \mathcal{G}, \mathcal{H}, \mathcal{R}}$ is in $\FNP/\poly$.
\end{lemma}
\begin{proof}
We demonstrate how checking if $x$ is a solution to $r_{\textnormal{pub}}$ can be implemented in polynomial time. The algorithm receives as a non-uniform advice $G_n \in \mathcal{G}$, $\mathbbm{1}\{n \in \mathcal{H}\}$ and $R_n \in \mathcal{R}$.

\begin{algorithm}[H]
\caption{$\mathcal{S}_{\mathcal{A}, \mathcal{G}, \mathcal{H}, \mathcal{R}}$ Verification Procedure}\label{alg:verification}
\begin{algorithmic}[1]
\Statex \textbf{Input:} $((1^n, r_{\textnormal{pub}}), x)$
\Statex \textbf{Output:} Yes or NO
\If{$n \notin \mathcal{H}$}
    \State \textbf{return} YES \textbf{if} $x = 0^n$ \textbf{else} NO
\EndIf
\For{$i \in [l(n)]$}
\State Run $\mathcal{A}$ on graph $G_n \in \mathcal{G}$, input $x$, preset randomness $r_{\textnormal{pub}}$ and private coins $r_{\textnormal{sec}}^{(i)} \in R_n$
\State Let $y^{(i)}$ be the output labeling
\EndFor
\State Compute $q = \frac{1}{l} \sum_{i=1}^l \mathbbm{1} \{ y^{(i)} \text{ is invalid} \}$ \Comment{w.r.t the constraints $\mathcal{C}$ of $\Psi$}
\State \textbf{return} YES \textbf{if} $q > \theta(n) := \frac{\epsilon(n) + \varphi(n)}{2}$ \textbf{else} NO
\end{algorithmic}
\end{algorithm}

We take $b^{\mathcal{A}}_{\text{pub}}(n)$, the length of $r_{\textnormal{pub}}$, to be at most $\poly(n)$ (otherwise, nodes running in polynomial-time would not be able to fully read it). $x = O(n)$ since the input labels are of constant length. It follows that the input $((1^n, r_{\textnormal{pub}}), x)$ is of length $\poly(n)$.

Our modified variant of \textsf{LOCAL} expects the nodes to be efficient. Thus, simulating the execution of $\mathcal{A}$ on $n$~vertex graph $G_n$ can be done in $\poly(n)$ time per round, with the number of rounds being at most $n$.
Alg.~\ref{alg:verification} invokes $\mathcal{A}$ $l(n)$ times, where $l(n)$ is a polynomial. Overall, the total runtime is polynomial in $n$, which, as noted above, also makes it polynomial in the length of the input.
\end{proof}

\subsection{Hardness}

Apart from having an efficient procedure for verifying solutions, $\mathcal{S}_{\mathcal{A}, \mathcal{G}, \mathcal{H}, \mathcal{R}}$ possesses several other useful properties. 
In particular, we show that for certain choices of $\mathcal{R}$, when sampling $r_{\textnormal{pub}}$ uniformly at random, efficient algorithms manage to solve the problem only with negligible probability. 
In other words, for these $\mathcal{R}$'s the problem is \textit{infinitely-often potentially-vacuously polynomially hard-on-average} (see Def.~\ref{def:vacuous_hardness}). 

\begin{definition}[Usefulness]
\label{def:usefulness}
Let algorithm $\mathcal{A}$, graph family $\mathcal{G}$ and functions $\gamma(n)$ and $l(n)$ be as in Def.~\ref{def:search_problem}. 
For $n\in \mathbb{N}$, a multiset $R_n = \{ r_{\textnormal{sec}}^{(i)} \}_{i=1}^{l(n)}$ of strings in $\{ 0,1 \}^{n \cdot b^{\mathcal{A}}_{\textnormal{sec}}(n)}$ is \textsf{Useful} if it satisfies: for every $r_{\textnormal{pub}} \in \{ 0,1 \}^{b^{\mathcal{A}}_{\textnormal{pub}}(n)}$ (preset public randomness) and $x \in \{ 0,1 \}^{O(n)}$ (an input labeling), the multiset $T^{R_n}_{x, r_{\textnormal{pub}}} = \{ y^{(i)} \}_{i=1}^{l(n)}$ of output labelings (built as in Lemma~\ref{lemma:sample_characterizing_multiset}) is a $\gamma(n)$-characterizing multiset for $\mathcal{A}$ on $G_n$ (from $\mathcal{G}$), $r_{\textnormal{pub}}$ and $x$.

\noindent A sequence $\mathcal{R} = \{ R_n \}_{n \in \mathbb{N}}$ is \textsf{Useful} if $R_n \in \mathcal{R}$ is \textsf{Useful} for all $n \in \mathbb{N}$.
\end{definition}

Be aware that the definition of \textsf{Usefulness} depends on $n$, $G_n$ (from $\mathcal{G}$), $\mathcal{A}$, $\gamma(n)$ and $l(n)$. Whenever we use this term, we take for granted that these are determined as in  Def.~\ref{def:search_problem}.
Observe that a uniformly random $R_n$ (as in Lemma~\ref{lemma:sample_characterizing_multiset}) is \textsf{Useful} with high probability, according to Lemma~\ref{lemma:same_coins_for_all}. This implies the following.

\begin{corollary}
\label{cor:useful_multisets_exist}
For any $n\in \mathbb{N}$, a \textsf{Useful} multiset $R_n = \{ r_{\textnormal{sec}}^{(i)} \}_{i=1}^{l(n)}$ exists. It also follows that \textsf{Useful} sequences $\mathcal{R} = \{ R_n \}_{n \in \mathbb{N}}$ exist.
\end{corollary}

\begin{lemma}[Hardness]
\label{lemma:usefulness_implies_hardness}
Let algorithm $\mathcal{A}$, graph family $\mathcal{G}$, sequence $\mathcal{H}$ and function $l(n)$ be as in Def.~\ref{def:search_problem}.
Let $\mathcal{R} := \{ R_n \}_{n \in \mathbb{N}}$ be a \textsf{Useful} sequence of multisets $R_n = \{ r_{\text{sec}}^{(i)} \}_{i=1}^{l(n)} \subset \{ 0,1 \}^{n \cdot b_{\text{sec}}^{\mathcal{A}}(n)}$.
Set $\mathcal{S} := \mathcal{S}_{\mathcal{A}, \mathcal{G}, \mathcal{H}, \mathcal{R}}$.

Ensemble $\mathcal{U} := \{ \mathcal{U}_n \}_{n \in \mathbb{N}}$ is a sequence of probability distributions. For every $r_{\text{pub}} \in \{ 0,1 \}^{b_{\text{pub}}^\mathcal{A}(n)}$, $\underset{\mathcal{U}_n}{\Pr}[(1^n, r_{\text{pub}})] = 2^{-b_{\text{pub}}^\mathcal{A}(n)}$. Namely, this is the uniform distribution over $b_{\text{pub}}^\mathcal{A}(n)$-bits long strings, with $1^n$ augmented to them.

Then, $(\mathcal{S}, \mathcal{U})$ (a distributional search problem) is \textit{public-coin infinitely-often potentially-vacuously polynomially hard-on-average}.

That is, there is a negligible function $\mathsf{negl}(n)$ such that there is no (non-uniform) \textsf{PPT} $\Sol$ that satisfies the following for all large enough $n$:
\begin{equation}
\label{eq:success_of_s_solver}
\Pr_{\substack{\text{$\Sol$'s coins}\\r_{\textnormal{pub}} \stackrel{\$}{\leftarrow} \{ 0,1 \}^{b^{\mathcal{A}}_{\textnormal{pub}}(n)}}}
[ (r_{\textnormal{pub}}, \Sol(1^n, r_{\textnormal{pub}})) \in \mathcal{S} ] 
\geq \mathsf{negl}(n)
\end{equation}

\end{lemma}

\begin{proof}
We prove the lemma for $\mathsf{negl}(n) := \frac{4 \epsilon(n)}{3 \varphi(n)}$ (where $\epsilon(n)$ and $\varphi(n)$ are as defined in Theorem~\ref{thm:probabilitygap_enhanced}). Since $\epsilon(n)$ is negligible and $\varphi(n)$ is an inverse polynomial, the expression is, as required, a negligible function.

Let $\Sol$ be a non-uniform \textsf{PPT} algorithm that on a polynomial-length advice and an input $(1^n, r_{\textnormal{pub}})$ (an instance of $\mathcal{S}$) returns $x$, a candidate solution. 
In particular, the solver $\Sol$ receives as part of its non-uniform advice the description of the graph $G_n \in \mathcal{G}$, the multiset $R_n \in \mathcal{R}$ and a bit indicating whether $n \in \mathcal{H}$.

Towards a contradiction, assume that there exists $n_1$ such that for all $n \geq n_1$ Eq.~(\ref{eq:success_of_s_solver}) holds. 
Moreover, in Theorem~\ref{thm:probabilitygap_enhanced} we assume that there exists $n_2$ such that $\mathcal{A}$ solves $\Psi$ with probability $1- \epsilon(n)$ for all $n \geq n_2$.
Fix some $n \geq \max \{ n_1, n_2\}$ that belongs to $\mathcal{H}$ (recall that $\mathcal{H}$ is infinite, so such an $n \in \mathcal{H}$ is guaranteed to exist).
Note that since $\mathcal{R}$ is \textsf{Useful}, it follows that for the chosen $n$, $R_n$ is also \textsf{Useful}.

We construct an adversary $\mathcal{B}_{\mathcal{A}, G_n}$ for the \textsf{LCL} problem $\Psi$. Recall that the adversaries to \textsf{LCL} problems are allowed to depend on the graph size $n$, on the graph itself $G_n \in \mathcal{G}$ and on the \textsf{LOCAL} algorithm $\mathcal{A}$. 
Thus, there is no issue with having the description of the algorithm $\Sol$ hardcoded into $\mathcal{B}_{\mathcal{A}, G_n}$, including $\mathcal{A}$, $G_n$, $R_n$ and any other polynomial-length advice provided to $\Sol$.
As an input, $\mathcal{B}_{\mathcal{A}, G_n}$ receives $(1^n, r_{\textnormal{pub}})$ for $r_{\text{pub}} \in \{ 0,1 \}^{b^{\mathcal{A}}_{\textnormal{pub}}(n)}$. 
$\mathcal{B}_{\mathcal{A}, G_n}$ invokes $\Sol$ on input $(1^n, r_{\textnormal{pub}})$ and returns the input labeling $x = \Sol(1^n, r_{\textnormal{pub}})$ as is. Naturally, since $\Sol$ is polynomial-time, so is $\mathcal{B}_{\mathcal{A}, G_n}$.

For $n \in \mathcal{H}$, the definition of $\mathcal{S}$ suggests that given an instance $(1^n, r_{\textnormal{pub}})$, for a string $x$ to qualify as a solution the fraction of the strings $r_{\text{sec}}^{(i)}$ in $R_n$ for which $\mathcal{A}$ fails must be at least $\theta(n)$. 
When $R_n$ is \textsf{Useful}, this fraction is at most $\gamma(n)$ away from the actual probability (taken over $r_{\text{sec}}$). This holds regardless of the choice of $x$, $r_{\textnormal{pub}}$. 
Combining these two facts, assuming that $( (1^n, r_{\textnormal{pub}}), x) \in \mathcal{S}$ (that is, $x$ indeed qualifies as a solution) we have:
\begin{align*}
\Pr_{r_{\textnormal{sec}}}\left[\mathcal{A}(G_n, x, r_{\textnormal{pub}}, r_{\textnormal{sec}})\text{ is invalid} \right] 
\geq 
\theta(n) -  \gamma(n) 
=
\left( \frac{\epsilon(n) + \varphi(n)}{2} \right) - \left( \frac{\varphi(n) - \epsilon(n)}{4} \right) 
=
\frac{3\epsilon(n) + \varphi(n)}{4}
\end{align*}
The probability is over $r_{\text{sec}}$, where $x$, $r_{\textnormal{pub}}$ are \textit{any} fixed pair satisfying $((1^n, r_{\textnormal{pub}}), x) \in \mathcal{S}$. 

Let $X_n$ be a random variable distributed according to the output of $\Sol(1^n, r_{\textnormal{pub}})$ (it depends on the random variable $r_{\textnormal{pub}}$ and coin tosses of $\Sol$). Choosing $\Sol$'s coins and $r_{\textnormal{pub}}$ uniformly at random, conditioned on $((1^n, r_{\textnormal{pub}}), X_n) \in \mathcal{S}$, the lower bound continues to hold.
\begin{equation*}
\Pr_{\substack{\text{$\Sol$'s coins}\\r_{\textnormal{pub}}, r_{\textnormal{sec}}}}
\left[\mathcal{A}(G_n, X_n, r_{\textnormal{pub}}, r_{\textnormal{sec}})\text{ is invalid} \big| 
((1^n, r_{\textnormal{pub}}), X_n) \in \mathcal{S} \right]
\geq \frac{3\epsilon(n) + \varphi(n)}{4}
\end{equation*}
Recall that $\mathcal{B}_{\mathcal{A}, G_n}$ simply invokes $\Sol$. It follows that:
\allowdisplaybreaks
\begin{align*}
& \Pr_{\substack{\text{$\mathcal{B}_{\mathcal{A}, G_n}$'s coins}\\r_{\textnormal{pub}}, r_{\textnormal{sec}}}}
\left[\mathcal{A}(G_n, \mathcal{B}_{\mathcal{A}, G_n}(1^n, r_{\textnormal{pub}}), r_{\textnormal{pub}}, r_{\textnormal{sec}})\text{ is invalid} \right]
\\
&=\Pr_{\substack{\text{$\Sol$'s coins}\\r_{\textnormal{pub}}, r_{\textnormal{sec}}}}[\mathcal{A}(G_n, X_n, r_{\textnormal{pub}}, r_{\textnormal{sec}})\text{ is invalid}] 
\\ 
&\geq\Pr_{\substack{\text{$\Sol$'s coins}\\r_{\textnormal{pub}}, r_{\textnormal{sec}}}}
\left[\mathcal{A}(G_n, X_n, r_{\textnormal{pub}}, r_{\textnormal{sec}})\text{ is invalid} \big| 
((1^n, r_{\textnormal{pub}}), X_n) \in \mathcal{S} \right] \cdot
\Pr_{\substack{\text{$\Sol$'s coins}\\r_{\textnormal{pub}}}}
[((1^n, r_{\textnormal{pub}}), X_n) \in \mathcal{S} ] 
\\ 
&\geq \bigg( \frac{3\epsilon(n) + \varphi(n)}{4} \bigg) \cdot \Pr_{\substack{\text{$\Sol$'s coins}\\r_{\textnormal{pub}}}}[(r_{\textnormal{pub}}, \Sol(1^n, r_{\textnormal{pub}})) \in \mathcal{S} ]
\end{align*}

For any $n \geq n_2$, on any graph with $n$ vertices, algorithm $\mathcal{A}$ fails with probability at most $\epsilon(n)$ over the choice of $r_{\textnormal{pub}}$ and $r_{\textnormal{sec}}$, provided that $x$ is selected by an efficient algorithm. 
This is of course true for our algorithm $\mathcal{B}_{\mathcal{A}, G_n}$ and the $n$~vertex graph $G_n \in \mathcal{G}$.
\[
\Pr_{\substack{\text{$\Sol$'s coins}\\r_{\textnormal{pub}}}}
[ (r_{\textnormal{pub}}, \Sol(1^n, r_{\textnormal{pub}})) \in \mathcal{S} ] 
\leq
\frac{4 \epsilon(n)}{3\epsilon(n) + \varphi(n)} <
\frac{4 \epsilon(n)}{3\varphi(n)} 
\]
This is in contrast to our assumption that Eq.~(\ref{eq:success_of_s_solver}) holds for all sufficiently large $n$, thereby concluding the proof. 
Note that the distributional search problem is \textit{public-coin} infinitely-often potentially-vacuously (polynomially) hard-on-average because the coins used when sampling an instance $(1^n, r_{\text{pub}})$ are simply $r_{\text{pub}}$. That is, the coins were assumed to be publicly known the entire time.
\end{proof}

\subsection{Non-Triviality}

We prove that for all $n$, with non-negligible probability the sampled instance $(1^n, r_{\textnormal{pub}})$ has a solution with respect to the search problem $\mathcal{S}$, i.e., there is an $x$ such that $((1^n, r_{\textnormal{pub}}), x) \in \mathcal{S}$. 
Unfortunately, this still doesn't mean that $\mathcal{S}$ is in $\TFNP / \poly$. Nevertheless, it provides the foundations for constructing such a problem.

\begin{lemma}[Non-Triviality]
\label{lemma:solution_exists}
Let algorithm $\mathcal{A}$ graph family $\mathcal{G}$ and sequence $\mathcal{H}$ be as in Def.~\ref{def:search_problem}.
Suppose that $\mathcal{R}$ is \textsf{Useful}, and let $\mathcal{S} = \mathcal{S}_{\mathcal{A}, \mathcal{G}, \mathcal{H}, \mathcal{R}}$.

Then, $\forall n \in \mathbb{N}$, when $r_{\textnormal{pub}}$ is sampled uniformly at random from $\{ 0,1 \}^{b^{\mathcal{A}}_{\textnormal{pub}}(n)}$, the probability that $(1^n, r_{\textnormal{pub}})$ has a solution with respect to the search problem $\mathcal{S}$ is non-negligible. That is, 
\[
\Pr_{r_{\textnormal{pub}} \stackrel{\$}{\leftarrow} \{ 0,1 \}^{b^{\mathcal{A}}_{\textnormal{pub}}(n)}}
[\exists x\textnormal{ s.t }((1^n, r_{\textnormal{pub}}), x) \in \mathcal{S}] 
\geq 
\frac{1}{\poly(n)}
\]
\end{lemma}

\begin{proof}
For $n \notin \mathcal{H}$, this is trivial, as the definition of $\mathcal{S}$ allows $x = 0^n$ to serve as a solution in this case. From now on, we focus on $n \in \mathcal{H}$.

Denote by $\mathsf{Good}$ the event of sampling $r_{\textnormal{pub}}$ such that on \textit{some} input labeling $x$ (which may not be efficiently computable) the algorithm $\mathcal{A}$ returns an invalid output labeling for at least $\theta(n) := \frac{\epsilon(n) + \varphi(n)}{2}$ fraction of the strings $r_{\text{sec}}^{(i)}$ in $R_n$. 
Since $n \in \mathcal{H}$, this corresponds to having a solution. Our objective is to lower bound $\Pr_{r_{\textnormal{pub}}}[\mathsf{Good}]$.

$\mathcal{R}$ is given to be \textsf{Useful}, implying that $R_n \in \mathcal{R}$ can be used for constructing $\gamma$-characterizing multisets for the output labelings on every input $x$ and preset public randomness $r_{\textnormal{pub}}$. 
The \textit{expected value} (i.e.\ the fraction of output labelings in the characterizing multiset that are considered valid with respect to $\Psi$) cannot be more than $\gamma(n)$-far from the actual success probability. This implies that for $r_{\textnormal{pub}}$ for which $\mathsf{Good}$ doesn't hold, and for every $x$, the actual failure probability (over the private coins) is at most:
\begin{align*}
& \Pr_{r_{\textnormal{sec}}}[\mathcal{A}(G_n, x, r_{\textnormal{pub}}, r_{\textnormal{sec}}) \text{ is invalid}] 
\\
&\leq \Pr_{i \stackrel{\$}{\leftarrow} [l]}[\mathcal{A}(G_n, x, r_{\textnormal{pub}}, r_{\textnormal{sec}}^{(i)}) \text{ is invalid}] + \gamma(n) \\
&\leq \theta(n) + \gamma(n) = 
\frac{\epsilon(n) + \varphi(n)}{2}  +  \frac{\varphi(n) - \epsilon(n)}{4} 
=
\frac{3\varphi(n) + \epsilon(n)}{4} 
\end{align*}
This upper bound holds even when the probability is taken over both $r_{\textnormal{sec}}$ and $r_{\textnormal{pub}}$, when the input $x$ may depend on $r_{\textnormal{pub}}$, if it is conditioned on $\neg \textsf{Good}$. 
Later, we will choose $\mathcal{B}_{\mathcal{A}, G_n}$ some unbounded adversary for the \textsf{LCL} problem $\Psi$ (which may depend on the \textsf{LOCAL} algorithm $\mathcal{A}$ and the graph $G_n$). Since it is unbounded, it can be assumed to be deterministic. For this adversary, let $X_n$ be a random variable distributed according to $\mathcal{B}_{\mathcal{A}, G_n}$'s output on $(1^n, r_{\text{pub}})$ (the randomness is over the choice of $r_{\text{pub}}$).

\begin{equation*}
\Pr_{r_{\textnormal{pub}}, r_{\textnormal{sec}}}\left[ \mathcal{A}(G_n, X_n, r_{\textnormal{pub}}, r_{\textnormal{sec}}) \text{ is invalid} \big| \neg \textsf{Good} \right] \leq \frac{3\varphi(n) + \epsilon(n)}{4} 
\end{equation*}
Without assuming $\neg \textsf{Good}$, the probability of failure turns out to be upper bounded by:
\begin{align*}
& \Pr_{r_{\textnormal{pub}}, r_{\textnormal{sec}}}\left[ \mathcal{A}(G_n, X_n, r_{\textnormal{pub}}, r_{\textnormal{sec}}) \text{ is invalid} \right] \\
&\leq (1-\Pr_{r_{\textnormal{pub}}}[\textsf{Good}]) \cdot \Pr_{r_{\textnormal{pub}}, r_{\textnormal{sec}}}\left[ \mathcal{A}(G_n, X_n, r_{\textnormal{pub}}, r_{\textnormal{sec}}) \text{ is invalid} \big| \neg \textsf{Good} \right] + \Pr_{r_{\textnormal{pub}}}[\textsf{Good}] \cdot 1 \\
&\leq (1-\Pr_{r_{\textnormal{pub}}}[\textsf{Good}]) \cdot \bigg( \frac{3\varphi(n) + \epsilon(n)}{4}  \bigg) + \Pr_{r_{\textnormal{pub}}}[\textsf{Good}]
\end{align*}

From the way the graph family $\mathcal{G}$ and the algorithm $\mathcal{A}$ were chosen, we know that against unbounded adversaries, $\mathcal{A}$ fails on $G_n \in \mathcal{G}$ with probability at least $\varphi(n)$ (when restricted to $n \in \mathcal{H}$). As such, there is an unbounded adversary that finds an input labeling $x$ such that:
\[
\Pr_{r_{\textnormal{pub}}, r_{\textnormal{sec}}}\left[ \mathcal{A}(G_n, x, r_{\textnormal{pub}}, r_{\textnormal{sec}}) \text{ is invalid} \right] \geq \varphi(n)
\]

Let $\mathcal{B}_{\mathcal{A}, G_n}$ be this adversary.
Together, we can lower bound $\Pr_{r_{\textnormal{pub}}}[\textsf{Good}]$, the probability of having a ``good'' preset public randomness. 
\begin{equation*}
\varphi(n) \leq (1-\Pr_{r_{\textnormal{pub}}}[\textsf{Good}]) \cdot \bigg( \frac{3\varphi(n) + \epsilon(n)}{4}  \bigg) + \Pr_{r_{\textnormal{pub}}}[\textsf{Good}]
\end{equation*}
$\epsilon(n)$ is negligible and $\varphi(n)$ is the inverse of some polynomial. Recall that for simplicity we defined $\mathcal{H}$ to include only integers for which it holds that $\epsilon(n) \leq \frac{1}{2}\varphi(n)$.
\begin{equation*}
\Pr_{r_{\textnormal{pub}}}[\textsf{Good}] \geq 
\frac{\varphi(n) - \epsilon(n)}{4 -  3 \varphi(n) - \epsilon(n)} \geq 
\frac{\varphi(n) - \frac{1}{2}\varphi(n)}{4} = 
\frac{1}{8}\varphi(n)
\end{equation*}
In words, the probability to sample a ``good'' preset public randomness is non-negligible.
\end{proof}

\begin{corollary}
\label{cor:not_total_search_problem}
Let algorithm $\mathcal{A}$, graph family $\mathcal{G}$ and sequence $\mathcal{H} \subseteq \mathbb{N}$ be as in Def.~\ref{def:search_problem}.
Set $\mathcal{S}$ to be $\mathcal{S}_{\mathcal{A}, \mathcal{G}, \mathcal{H}, \mathcal{R}}$ as defined in Def.~\ref{def:search_problem}, where $\mathcal{R}$ is as promised in Corollary~\ref{cor:useful_multisets_exist}.
Ensemble $\mathcal{U}$ is as in Lemma~\ref{lemma:usefulness_implies_hardness}.

Then, $(\mathcal{S}, \mathcal{U})$ is a \textit{non-trivial}, \textit{public-coin infinitely-often potentially-vacuously polynomially hard-on-average} distributional problem in $\FNP / \poly$.
\end{corollary}

\begin{proof}
This is a direct result of Lemmas~\ref{lemma:search_problem} (efficient verification), \ref{lemma:usefulness_implies_hardness} (hardness) and \ref{lemma:solution_exists} (non-triviality). Note that for Lemma~\ref{lemma:usefulness_implies_hardness} to apply, we need $\mathcal{R}$ to be \textsf{Useful}. The existence of such sequences $\mathcal{R}$ is guaranteed by Corollary~\ref{cor:useful_multisets_exist}.
\end{proof}

\begin{proof}[Proof of Theorem~\ref{thm:probabilitygap_enhanced}]
From Corollary~\ref{cor:not_total_search_problem} it holds that $(\mathcal{S}, \mathcal{U})$ is a non-trivial, public-coin infinitely-often potentially-vacuously polynomially hard-on-average distributional search problem in $\FNP/\poly$.
Use the transformation of Lemma~\ref{lemma:transformation_total}. The transformation maintains the hardness of the problem (the success probability of a \textsf{PPT} solver grows by at most by a polynomial multiplicative factor, so it remains negligible) and the ability to verify solutions efficiently, together with making it \textit{total}. 
Namely, it yields a public-coin infinitely-often polynomially hard-on-average distributional problem in $\TFNP/\poly$.
\end{proof}

\blockanon{

\section*{Acknowledgments}

We thank Merav Parter for insightful suggestions that helped improve the presentation of our results, and the RANDOM 2025 anonymous referees for their useful comments.
}

\printbibliography[heading=bibintoc, title={References}]

\appendix

\section{Appendices}

\subsection{Proof of Lemma~\ref{lemma:gap_oblivious}}
\label{subsec:proof_gap_oblivious}

\oblivious*

\begin{proof}
Let $\Pi_1$ be the \textsf{LCL} problem defined in \cite[Sec. 7]{bgk25}. The upper bound in the oblivious model is stated in Theorem~8.2. As for the lower bound, the proof of Theorem~8.1 will have to undergo a minor modification.
We assume familiarity with the graph family defined in \cite[Sec. 6]{bgk25}, which we refer to as \textit{augmented grid structures} (AG). See Section~\ref{subsubsec:augmented_grid} for details. They observe (in Lemma~7.1) that whenever a graph belongs to that family, there is an input labeling such that the only valid output labeling (with respect the constraints of $\Pi_1$) satisfies: (1) the outputs in each row are consistent, (2) in at least one of the rows, the output equals to the input of the rightmost node. Henceforth, we focus on this input labeling.

The original proof relies on the outputs of nodes located in opposite ends of the AG being independent, if the algorithm terminates quickly enough. 
The difference lies in the fact that in the preset public coins model we can no longer claim that the outputs of nodes, even if they are very far apart, are independent. This is due to the presence of public randomness. This can be easily circumvented by observing that conditioned on the public random string, the outputs are indeed independent. 

Formally, let $G$ be an $n$~vertex AG with a grid of size $h \times w$ for $h = w$ (implying that $w = \Omega(\sqrt{n})$). Let $\mathcal{A}$ be a \textsf{LOCAL} algorithm that solves $\Pi_1$ within $o(\sqrt{n})$ rounds with probability $1 - \frac{1}{n}$. We restrict ourselves to large enough values of $n$ such that $\mathcal{A}$'s runtime is at most $\frac{\sqrt{n}}{3}$. For an arbitrary row, let $u$ (resp. $v$) be the leftmost (resp. rightmost) node. Since $v$ is not contained in $u$'s view after $\frac{\sqrt{n}}{3}$ rounds (and vice versa), their outputs (conditioned on the preset public coins $r_{\textnormal{pub}}$) are independent. Since the algorithm succeeds with probability at least $1 - \frac{1}{n}$, it follows that:
\[
\begin{cases}
\Pr_{r_{\text{sec}}}[\text{$u$ outputs 1} | r_{\textnormal{pub}}] \cdot (1 - \Pr_{r_{\text{sec}}}[\text{$v$ outputs 1}| r_{\textnormal{pub}}]) < \frac{1}{n} \\
\Pr_{r_{\text{sec}}}[\text{$v$ outputs 1}| r_{\textnormal{pub}}] \cdot (1 - \Pr_{r_{\text{sec}}}[\text{$u$ outputs 1}| r_{\textnormal{pub}}]) < \frac{1}{n}
\end{cases}
\]
Where the probabilities are taken over the choice of the private coins. It can be inferred that either $\Pr[\text{$u$ outputs 1}| r_{\textnormal{pub}}]$ and $\Pr[\text{$v$ outputs 1}| r_{\textnormal{pub}}]$ are both upper bounded by $\frac{2}{n}$, or that they are lower bounded by $1 - \frac{2}{n}$. 
The bounded adversary selects an input labeling that satisfies $\Pr[\text{$u$ outputs 1}| r_{\textnormal{pub}}] < \frac{2}{n}$ iff $v$'s input is $1$. The probability that $u$'s output equals to $v$'s input is at most $\frac{2}{n}$. 
The same argument applies to each of the $h = O(\sqrt{n})$ rows. From a union bound, the success probability of $\mathcal{A}$ is at most $O(\sqrt{n}) \cdot \frac{2}{n} = O(\frac{1}{\sqrt{n}})$. $n$ can be chosen such that this is strictly smaller than $1 - \frac{1}{n}$, thereby reaching a contradiction. 
This proves that for any $o(\sqrt{n})$~round \textsf{LOCAL} algorithm, there are infinitely many values of $n$ for which the algorithm succeeds with probability strictly smaller than $1 - \frac{1}{n}$ (on carefully chosen $n$~vertex graphs).
\end{proof}

\subsection{Proof of Lemma~\ref{lemma:gap_private}}
\label{subsec:proof_gap_private_coins}

\privatecoins*

\begin{proof}[Proof (sketch)]
We outline the primary changes between the proof of Lemma~\ref{lemma:gap_private} and the proofs of \cite[Thm. 8.1, 8.2]{bgk25}. We use a similar \textsf{LCL} problem to demonstrate the separation between the models. In high level, the hard instances are a variant of the \textit{augmented grid structures} (AG) from Section~\ref{subsubsec:augmented_grid}, but in the input labeling for the \textsf{LCL} problem that we define we don't assume an additional input bit $b_{\text{in}} \in \{ 0,1 \}$ associated with each grid node. Instead, we encode these bits using subgraphs that we attach to the nodes. Their exact structure doesn't matter, as long as the diameter is $O(1)$. For example, consider the two graph structures proposed in Fig.~\ref{fig:01_encoding}. 

\begin{figure}[t]%
\centering
\subfloat[\centering 0-Gadget]{
\begin{tikzpicture}[every node/.style={circle, fill=black, inner sep=2pt}]
\node[label={30:$v$}] (v) at (0,0) {};
\node (a) at (3,0) {};
\node (b) at (4,1) {};
\node (c) at (4,-1) {};
\node (d) at (5,0) {};
\draw[color=black] (v) -- (a);
\draw[color=black] (a) -- (b);
\draw[color=black] (a) -- (c);
\draw[color=black] (b) -- (d);
\draw[color=black] (c) -- (d);
\end{tikzpicture}
}%
\qquad\qquad
\subfloat[\centering 1-Gadget]{
\begin{tikzpicture}[every node/.style={circle, fill=black, inner sep=2pt, label distance=0.2pt}]
\node[label={30:$v$}] (v) at (0,0) {};
\node (a) at (3,0) {};
\node (b) at (4.5,1) {};
\node (c) at (4.5,-1) {};
\draw[color=black] (v) -- (a);
\draw[color=black] (a) -- (b);
\draw[color=black] (b) -- (c);
\draw[color=black] (a) -- (c);
\end{tikzpicture}
}%
\caption{Graph structures to be attached to grid node $v$. Should be interpreted as if the node $v$ is associated with the input bit $0$ or $1$, respectively.}%
\label{fig:01_encoding}%
\end{figure}
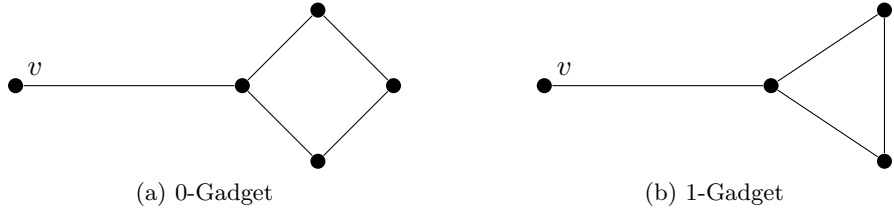

An AG in which all the grid nodes are connected to such gadgets is called \textit{augmented grid structure with gadgets}, or AGG. We don't provide a formal definition. The \textsf{LCL} problem $\Pi^{\text{badAGG}}$ defined below is a modification of $\Pi^{\text{badAG}} = (\mathcal{V}^{\textnormal{badAG}}_{\textnormal{input}}, \mathcal{E}^{\textnormal{badAG}}_{\textnormal{input}} , \mathcal{V}^{\textnormal{badAG}}_{\textnormal{output}}, \mathcal{C}^{\textnormal{badAG}})$ (see Def.~\ref{def:bad_ag}).

\begin{definition}
\label{def:bad_agg}
$\Pi^{\textnormal{badAGG}}$ is an \textsf{LCL} problem $(\mathcal{V}^{\textnormal{badAGG}}_{\textnormal{input}}, \mathcal{E}^{\textnormal{badAGG}}_{\textnormal{input}} , \mathcal{V}^{\textnormal{badAGG}}_{\textnormal{output}}, \mathcal{C}^{\textnormal{badAGG}})$.
\begin{itemize}

\item $\mathcal{V}^{\textnormal{badAGG}}_{\textnormal{input}} =  \{ \mathsf{gadgetNode}\} \cup \mathcal{V}^{\textnormal{badAG}}_{\textnormal{input}} $

\item $\mathcal{E}^{\textnormal{badAGG}}_{\textnormal{input}} = \{ \mathsf{gadgetEdge}\} \cup \mathcal{E}^{\textnormal{badAG}}_{\textnormal{input}}$

\item $\mathcal{V}^{\textnormal{badAGG}}_{\textnormal{output}} = \{ \mathsf{GadgetErr} \} \cup \mathcal{V}^{\textnormal{badAG}}_{\textnormal{output}}$

\item Constraints $\mathcal{C}^{\textnormal{badAGG}}$:

\begin{itemize}[nosep]

\item On the subgraph induced by the nodes with input labels \textsf{treeNode} and $(\mathsf{treeNode}, \mathsf{gridNode}, l)$ (for $l \in \mathcal{V}^{\text{vGrid}}$), the constraints $\mathcal{C}^{\text{badAG}}$ are satisfied.

\item \label{item:con-tree_node_and_gadget} Let $v \in V$ be a \textsf{treeNode} labeled node. Then, if it has \textsf{gadgetEdge} labeled half-edges, $v$'s output is \textsf{Err}.

\item \label{item:con-grid_node_and_gadget} Let $v \in V$ be a $(\mathsf{treeNode}, \mathsf{gridNode}, l)$ labeled node. Then, if the number of \textsf{gadgetEdge} labeled half-edges is not exactly one, $v$'s output is \textsf{Err}. 
Additionally, let $u = f(v, \mathsf{gadgetEdge})$. If $u$'s input label is not \textsf{gadgetNode} or if $u$'s output label is not $\bot$, then $v$'s output is \textsf{Err}.

\item \label{item:con-valid_gadget} Consider a connected subgraph of \textsf{gadgetNode} labeled nodes. If the subgraph is from one of the structures depicted in Fig.~\ref{fig:01_encoding} (this can be verified by inspecting the $O(1)$ neighborhood), all the \textsf{gadgetNode} labeled nodes in the subgraph return $\bot$. Otherwise, they return \textsf{GadgetErr}.

\end{itemize}

\end{itemize}

\end{definition}

Analogues of Lemmas~\ref{lemma:bad_ag_invalid},~\ref{lemma:bad_ag_valid} suggest that there exists an input labeling for $\Pi^{\text{badAGG}}$ such that the subgraph induced by the nodes that return $\bot$ as their outputs is an AGG. Moreover, the output labels that are not $\bot$ constitute pointer chains directed towards violations of the constraints $\mathcal{C}^{\textnormal{badAGG}}$. The output labeling can be computed by an $O(\log n)$ round \textsf{LOCAL} algorithm.

Next, we present the definition of the \textsf{LCL} problem $\Pi_2$ that demonstrates the gap stated in the lemma. In words, this is exactly the same problem as in \cite[Sec. 7]{bgk25}, with the following change: grid nodes in the right column output $\text{YES}$ if and only if their output bit is $b$ and they are attached to a $b$-Gadget. That is, the gadgets serve as a replacement for the input bits used in the original problem.

\begin{definition}
\label{def:pi_for_agg}
$\Pi_2$ is an \textsf{LCL} problem $(\mathcal{V}^{\Pi_2}_{\textnormal{input}}, \mathcal{E}^{\Pi_2}_{\textnormal{input}} , \mathcal{V}^{\Pi_2}_{\textnormal{output}}, \mathcal{C}^{\Pi_2})$.

\begin{itemize}
\item $\mathcal{V}^{\Pi_2}_{\textnormal{input}} := \mathcal{V}^{\text{badAGG}}_{\text{input}}$, 
$\mathcal{E}^{\Pi_2}_{\textnormal{input}} := \mathcal{E}^{\text{badAGG}}_{\textnormal{input}}$

\item $\mathcal{V}^{\Pi_2}_{\textnormal{output}} := \mathcal{V}^{\textnormal{badAGG}}_{\textnormal{output}} \cup \{  \text{YES}, \text{NO}\} \cup (\{ 0,1 \}\times\{ \text{YES}, \text{NO} \} )$

Output labels associated with $\Pi^{\textnormal{badAGG}}$ are valid with respect to $\Pi_2$. Furthermore, the grid nodes may have output labels in $\{0,1 \} \times \{  \text{YES}, \text{NO} \}$, and tree nodes may have output labels in $\{  \text{YES}, \text{NO} \}$.

\item Constraints $\mathcal{C}^{\Pi_2}$:\footnote{The main difference with respect to \cite[Sec. 7]{bgk25} is in constraint~\ref{item:con-gadget_yes}.}

\begin{enumerate}[nosep]
\item If we map all the labels not in $\mathcal{V}^{\textnormal{badAGG}}_{\textnormal{output}}$ to $\bot$, $\mathcal{C}^{\textnormal{badAGG}}$ is satisfied.

\item If a node is labeled $(\mathsf{treeNode}, \mathsf{gridNode}, l)$, and its output is not in $\mathcal{V}^{\textnormal{badAGG}}_{\textnormal{output}}$, then it must be in $\{ 0,1 \} \times \{ \text{YES}, \text{NO} \}$.

\item If the output of a node $u$ is $(b, x) \in \{ 0,1 \} \times \{ \text{YES}, \text{NO} \}$, then if $f(u, (\mathsf{gridEdge}, \mathsf{R}))$ exists and its output is not in $\mathcal{V}^{\textnormal{badAGG}}_{\textnormal{output}}$, then its output is $(b, x')$ (for some $x'$).

\item If a node is labeled with $\mathsf{treeNode}$ and the output is not in $\mathcal{V}^{\textnormal{badAGG}}_{\textnormal{output}}$, then the output must be in $\{ \text{YES}, \text{NO} \}$.

\item \label{item:con-gadget_yes} Let $u$ be a $(\mathsf{treeNode}, \mathsf{gridNode}, l)$ labeled node, with $f(u, (\mathsf{gridEdge}, \mathsf{R})) = \bot$. $u$'s output is $(b, \text{YES})$ for $b \in \{ 0,1 \}$ iff $u$ is attached to a $b$-Gadget (see Fig.~\ref{fig:01_encoding}).

\item If a node $u$ has an output from $\{ \text{YES}, \text{NO} \}$, let $v = f(u, (\mathsf{treeEdge}, \mathsf{Ch}_{\mathsf{L}}))$ and $z = f(u, (\mathsf{treeEdge}, \mathsf{Ch}_R))$. 
The outputs of $v$ and $z$ are in $\{ 0,1 \} \times \{ \text{YES}, \text{NO} \}$ or in $\{\text{YES}, \text{NO} \}$. $u$'s output is $\text{YES}$ iff 
$v$'s (or $z$'s) output contains $\text{YES}$.

\item If a node has a $\{ \text{YES}, \text{NO} \}$ output but no $(\mathsf{treeEdge}, \P)$ edge, then it must be a $\text{YES}$.

\end{enumerate}

\end{itemize}
\end{definition}

Relying on the observations made above, it's not hard to see that for every AG graph $G$ there is an input labeling such that in the only valid output labeling the rows must be consistent (i.e., the output bits of all the nodes in a row must be identical), and in at least one of the rows the output is $b$ when the rightmost node is connected to a $b$-Gadget.

The lower bound in the private coins model follows from the same argument used in \cite[Thm.~8.1]{bgk25} (the fact that we now use $b$-Gadgets instead of providing the input bits of the grid nodes explicitly doesn't matter for the sake of the proof). 

For the upper bound in the case of the preset public coins model, remember that the graph $G$ is fixed before the preset randomness is revealed. This is in contrast to the input labeling, which can be determined after $r_{\text{pub}}$ is already known. This implies that $r_{\text{pub}}$ is independent of the $b$-Gadgets attached to the grid nodes. 

The \textsf{LOCAL} algorithm that the nodes run executes steps similar to those in \cite[Thm.~8.2]{bgk25}. First, they run the $O(\log n)$ round algorithm promised earlier for solving $\Pi^{\text{badAGG}}$. In every connected subgraph made from nodes that return $\bot$ in $\Pi^{\text{badAGG}}$, the nodes first explore their $O(\log n)$ neighborhood in order to detect whether the height of the grid is at most $\lceil \log n \rceil$. 
If this is the case, the nodes can learn the topology of the entire network within $O(\log n)$ rounds (recall that the grid is \textit{vertical}, as described in Sec.~\ref{subsubsec:vgrid}). This allows each grid node to learn the structure of the $b$-Gadget connected to the rightmost node in the row, and to answer accordingly. 
Otherwise, the nodes learn the index $i$ of the row they are in. This can be done in $O(\log n)$ rounds by inspecting their position with respect to the column tree to which they belong. 
Afterwards, they output the $i$-th bit of the preset random string $(r_{\text{pub}})_i$. As claimed in the beginning, $r_{\text{pub}}$ is independent of $G$. In particular, $(r_{\text{pub}})_i$ is independent of the $b$-Gadget connected to the rightmost node in the $i$-th row. 
With probability $1/2$ the $b$-Gadget of that node corresponds to $(r_{\text{pub}})_i$. The probability of failing in every row is at most $(1/2)^{\log n} = 1/n$. 
Clearly, this \textsf{LOCAL} algorithm takes $O(\log n)$ rounds.
\end{proof}

\subsection{Proof of Lemma~\ref{lemma:transformation_total}}
\label{subsec:proof_transformation_total}

\transformation*

\begin{proof}
The proof closely follows \cite[Thm.~4.1]{hny17}. $(\mathcal{S}, \mathcal{D})$ is a \textit{non-trivial}, \textit{public-coin infinitely-often potentially-vacuously polynomially hard-on-average} distributional problem in $\FNP / \poly$. Let $D$ be the efficient sampler associated with the probability ensemble $\mathcal{D}$. 
We show how to turn $\mathcal{S}$ into a total problem. This gives rise to a public-coin infinitely-often polynomially hard-on-average distributional problem $(\mathcal{S}', \mathcal{D}')$ such that $\mathcal{S}' \in \TFNP/ \poly$.

Let $\lambda \in \mathbb{N}$.
W.l.o.g, $D$ is assumed to be \textit{length-preserving}. We use $\lambda$ to denote both the instance size and the size of the random strings used for sampling these instances. Because $(\mathcal{S}, \mathcal{D})$ is a \textit{non-trivial}, for some polynomial $p(\lambda)$ the following must hold.
\begin{equation*}
    \Pr_{r \stackrel{\$}{\leftarrow} \{ 0,1 \}^{\lambda} } \big[\exists y \text{ s.t }(x,y)\in\mathcal{S} | D(r) = x\big] \geq \frac{1}{p(\lambda)}
\end{equation*}

Set $k:= \lambda \cdot p(\lambda)$. Let $s_1, ..., s_{k}$ be a multiset of $k$ Boolean strings of length $\lambda$. The search problem ${\mathcal{S}'}_{s_1, ..., s_k}$ is defined as:
\begin{equation*}
{\mathcal{S}'}_{s_1, ..., s_k} : = 
\left\{
(r, y) \bigg| \lor_{i \in [k]} \big( (x_i, y) \in \mathcal{S} \big), \forall i: D(r \oplus s_i) = x_i
\right\}
\end{equation*}
Having $s_1, ...,s_k$ as a non-uniform advice, solutions for this problem can be verified efficiently. Instances of the search problem are sampled according to $\mathcal{D}' = \{ \mathcal{D}_{\lambda}' \}_{\lambda \in \mathbb{N}}$, the ensemble of uniform distributions over $\lambda$ bits long strings. 

\paragraph{Totality.}
Fix $\lambda \in \mathbb{N}$.
Given a fixed $r$, when $s_i \sim \text{Uniform}(\{ 0,1 \}^{\lambda})$ for every $i \in [k]$, the $\lambda$ bits long string $r \oplus s_i$ is also uniformly random.

\begin{equation*}
\Pr_{s_i}[\exists y \text{ s.t }(x_i,y)\in\mathcal{S} | D(r\oplus s_i) = x_i] \geq \frac{1}{p(\lambda)}
\end{equation*}
Moreover, the $s_i$'s are independent. From the product rule, the probability that each of them yields a problem without a solution is:
\begin{align*}
\Pr_{s_1, ..., s_k}[\forall i \in [k]: \not\exists  y \text{ s.t }(x_i,y)\in\mathcal{S} | D(r\oplus s_i) = x_i] &=\\
\left( \Pr_{s_i}[\not\exists y \text{ s.t }(x_i,y)\in\mathcal{S} | D(r\oplus s_i) = x_i] \right)^k
&\leq\\
\left( 
1 - \frac{1}{p(\lambda)}
\right)^{\lambda \cdot p(\lambda)}
&\leq\\
e^{-\lambda}
\end{align*}

For it to hold for every $r \in \{ 0,1 \}^{\lambda}$, we apply the union bound.
\begin{equation*}
\Pr_{s_1, ..., s_k} \big[ \exists r \text{ } \forall i \in [k]: \not\exists  y \text{ s.t }(x_i,y)\in\mathcal{S} | D(r\oplus s_i) = x_i \big] 
\leq 
2^{\lambda} \cdot e^{-\lambda}
< 1
\end{equation*}
The last inequality implies that there exists a choice of $s_1, ..., s_k$ for which there is a solution for every $r$. Pick $\{s_i\}_{i=1}^k$ that satisfies this, and let the problem $\mathcal{S}'$, for this specific $\lambda$, be defined as $\mathcal{S}'_{s_1, ..., s_k}$. 
By following these steps for every $\lambda \in \mathbb{N}$, we get a problem $\mathcal{S}'$ that is \textit{total}.
We consider both verifiers and solvers that are \textit{non-uniform}. In particular, the strings $\{s_i\}_{i=1}^k$ can be assumed to be part of their auxiliary inputs.

\paragraph{Hardness.}
The proof of hardness is by a reduction to the \textit{infinitely-often potentially-vacuous polynomial hardness-on-average} of the original search problem (against non-uniform solvers).
Let $\mathcal{A}'$ be a non-uniform \textsf{PPT} that solves $(\mathcal{S}', \mathcal{D}')$ with non-negligible probability (say, $1 / q(\lambda)$ for some polynomial $q(\lambda)$) for infinitely many $\lambda \in \mathbb{N}$:
\begin{equation*}
\Pr_{\substack{r \leftarrow \mathcal{D}_{\lambda}'\\\text{$\mathcal{A}'$'s coins}}}\left[ (r, \mathcal{A}'(r)) \in \mathcal{S}'\right] 
\geq
\frac{1}{q(\lambda)}
\end{equation*}
We claim that $\mathcal{A}$ described below is a $\poly(\lambda)$-time algorithm for solving $(\mathcal{S}, \mathcal{D})$ for infinitely many values of $\lambda \in \mathbb{N}$.

\begin{algorithm}[H]
\caption{$(\mathcal{S}, \mathcal{D})$ Solver $\mathcal{A}$}\label{alg:solve_dist_problem}
\begin{algorithmic}[1]
\Statex \textbf{Input:} $r \in \{ 0,1 \}^{\lambda}$
\Statex \textbf{Output:} $y$ such that $(x,y) \in \mathcal{S}$, where $x = D(r)$
\Statex Let $\{ s_i \}_{i=1}^{k(\lambda)} \subset \{ 0,1 \}^{\lambda}$ be the set of strings associated with $\mathcal{S}'$ for this specific $\lambda$.
\State $j \stackrel{\$}{\leftarrow} [k(\lambda))]$
\State $r': = r \oplus s_j$
\State \textbf{return} $y \leftarrow \mathcal{A}'(r')$
\end{algorithmic}
\end{algorithm}

\noindent The output of the given algorithm $\mathcal{A}'$ on a random $r \leftarrow \mathcal{D}_{\lambda}'$ satisfies the following:

\begin{align*}
& \Pr_{\substack{r\\\text{$\mathcal{A}'$'s coins}}}\left[ (r, \mathcal{A}'(r)) \in \mathcal{S}' \right]  & \text{(Def. $\mathcal{S}$ and $\mathcal{S}'_{s_1, ..., s_k}$)} \\
&=\Pr_{\substack{r\\\text{$\mathcal{A}'$'s coins}}}\left[ \lor_{i \in [k]} (x_i, \mathcal{A}'(r)) \in \mathcal{S} \big| D(r\oplus s_i) = x_i \right]  & \text{(Union bound)} \\
&\leq \sum_{i \in [k]} \Pr_{\substack{r\\\text{$\mathcal{A}'$'s coins}}} \left[ (x_i, \mathcal{A}'(r)) \in \mathcal{S} \big| D(r\oplus s_i) = x_i \right]  & \text{(Identically distributed)}\\
&= \sum_{i \in [k]} \Pr_{\substack{r\\\text{$\mathcal{A}'$'s coins}}} \left[ (x, \mathcal{A}'(r \oplus s_i)) \in \mathcal{S} \big| D(r) = x \right]  & \text{($j \sim \text{Uni}([k])$)}\\
&= k \cdot \Pr_{\substack{r, j\\\text{$\mathcal{A}'$'s coins}}} \left[ (x, \mathcal{A}'(r \oplus s_j)) \in \mathcal{S} \big| D(r) = x \right]  & \text{(Alg.~\ref{alg:solve_dist_problem})}\\
&= k \cdot \Pr_{\substack{r\\\text{$\mathcal{A}$'s coins}}}\left[ (x, \mathcal{A}(r)) \in \mathcal{S} \big| D(r) = x \right]
\end{align*}
Knowing that $\mathcal{A}'$ succeeds with probability at least $\frac{1}{q(\lambda)}$, we have:

\begin{equation*}
\Pr_{\substack{r\\\text{$\mathcal{A}$'s coins}}}\left[ (x, \mathcal{A}(r)) \in \mathcal{S} \big| D(r) = x \right] \geq 
\frac{1}{k \cdot q(\lambda)}
\end{equation*}
Recall that $k$ is a polynomial. That is, our \textsf{PPT} $\mathcal{A}$ solves $(\mathcal{S}, \mathcal{D})$ with non-negligible probability for the same $\lambda$ values as $\mathcal{A}'$. This comes as a contradiction to the
\textit{infinitely-often potentially-vacuous (polynomial) hardness-on-average} of
$(\mathcal{S}, \mathcal{D})$.

Note that $(\mathcal{S}', \mathcal{D}')$ is in fact a \textit{public-coin} infinitely-often polynomially hard-on-average distributional problem. This is evident from the way the instances are sampled.
\end{proof}
\end{document}